\documentclass[smallcondensed,a4paper]{svjour3}
\pdfoutput=1
\smartqed
\usepackage{graphicx}

\usepackage{newtxtext,newtxmath} 

\usepackage[utf8]{inputenc}
\usepackage[T1]{fontenc}
\usepackage[ngerman,english]{babel}
\usepackage{llncsstuffthm}
\usepackage{cwdefs}
\usepackage{amssymb}
\usepackage{amsfonts}
\usepackage{multirow}
\usepackage{longtable}
\usepackage{booktabs}
\usepackage{amsmath}
\usepackage{bigdelim}
\usepackage{comment}
\usepackage{xspace}
\usepackage{booktabs}
\usepackage{setspace}
\usepackage{upgreek}
\usepackage{enumitem}
\usepackage{rotating}
\usepackage{tikz}
\usetikzlibrary{shapes.geometric}
\usetikzlibrary{shapes.symbols}
\usetikzlibrary{positioning}
\usepackage{array}
\usepackage{xcolor}
\usepackage{colortbl}
\usepackage{capt-of}
\usepackage[labelfont=bf,hypcap=false]{caption}
\usepackage{pdfpages}

\usepackage{hyperref}
\hypersetup{
hyperfootnotes=false,
breaklinks=true,
hidelinks,
linktoc=all,
colorlinks=false
}

\title{Craig Interpolation and Access Interpolation\\ with Clausal First-Order Tableaux}

\author{Christoph Wernhard}
\institute{\email{info@christophwerhard.com}}

\date{Revision: October 18, 2018}
\renewcommand{\makeheadbox}{}

\newcommand{\binop}{\otimes}

\newcounter{equivcounter}
\newcounter{entailcounter}

\newcommand{\tup}[1]{\boldsymbol{#1}}

\newcommand{\uus}{\tup{u}}
\newcommand{\vvs}{\tup{v}}

\newcommand{\ffs}{\tup{f}}
\newcommand{\ggs}{\tup{g}}

\renewcommand{\xxs}{\tup{x}}

\newcommand{\yys}{\tup{y}}

\newcommand{\Qs}{\tup{Q}}
\newcommand{\Rs}{\tup{R}}

\newcommand{\Ss}{\tup{S}}
\renewcommand{\tts}{\tup{t}}

\newcommand{\emptyseq}{\epsilon}

\def\squareforxed{$\vartriangleleft$\hspace*{-0.01em}}

\def\xed{\ifmmode\squareforqed\else{\unskip\nobreak\hfil
\penalty50\hskip1em\null\nobreak\hfil\squareforxed
\parfillskip=0pt\finalhyphendemerits=0\endgraf}\fi}

\renewcommand{\subst}[3]{\warnSubstIsUndefined}

\newcommand{\var}[1]{\f{var}(#1)}

\newcommand{\parg}[1]{\f{arg}(#1)}
\newcommand{\garg}[1]{\f{garg}(#1)}

\renewcommand{\setminus}{\backslash}

\tikzset{
itria/.style={
  draw,
  solid,
  thin,
  isosceles triangle,
  isosceles triangle apex angle=60,
  shape border rotate=90,yshift=-6.8ex}
}

\newcommand{\tria}[1]
{{ node[itria]
    {\small \makebox[0.9em]{\rule[-0.55ex]{0pt}{2.2ex}#1} }}}

\newcommand{\linked}{ edge from parent[very thick] }

\newcommand{\closednode}[1]{node [label={below: \raisebox{1.4ex}[1.4ex]{$*$}}]
  {#1} }

\newcommand{\tableauscale}{0.7}

\newcommand{\nodeMarkLeafonly}[1]
{ node [draw,rectangle,inner sep=0pt,line width=1.5pt] {#1} }

\newcommand{\nodeMarkContig}[1]
{ node [draw,ellipse,inner sep=1pt,line width=1pt] {#1} }

\newcommand{\nodeMarkRegular}[1]
{ node [draw,tape,inner sep=1pt,line width=1pt] {#1} }

\newenvironment{tableaufig}[1]
{\noindent\begin{minipage}{\linewidth}%
\captionsetup[figure]{font=small,labelfont=bf}%
\centering%
\captionof{figure}{ #1}}
{\end{minipage}}

\newenvironment{tableaufigTwoCol}[1]
{\noindent\begin{minipage}[t]{0.5\linewidth}%
\captionsetup[figure]{font=small,labelfont=bf}%
\centering%
\captionof{figure}{#1}}
{\end{minipage}}

\newcommand{\rewrite}{\Rightarrow}

\newcommand{\aipol}[1]{\f{acc\hyph ipol}(#1)}

\newcommand{\SetOfS}{\mathcal{S}}

\newcommand{\ACIT}{ACI-Tableau\xspace}
\newcommand{\ACITX}{ACI-Tableaux\xspace}
\newcommand{\acit}{ACI-tableau\xspace}
\newcommand{\acitx}{ACI-tableaux\xspace}

\newenvironment{tabularlabtext}
{\begin{center}\begin{tabular}{R{2em}@{\hspace{0.5em}}L{30em}}}
{\end{tabular}\end{center}}

\newenvironment{arrayprf}
{\begin{array}{Z{2.5em}@{\hspace{1em}}X{32em}}}
{\end{array}}

\newenvironment{arrayprflong}
{\begin{longtable}{Z{2.5em}@{\hspace{1em}}X{32em}}}
{\end{longtable}}

\newcommand{\centerlong}[1]{\ \par
\vspace{-12pt}{\centering #1}\vspace{-10pt}}

\newenvironment{arrayprfeq}
{\begin{array}{Z{2.5em}@{\hspace{1em}}Y{1.5em}X{30.5em}}}
{\end{array}}

\newcommand{\algoinput}{\smallskip \noindent\textsc{Input: }}
\newcommand{\algooutput}{\smallskip \noindent\textsc{Output: }}
\newcommand{\algomethod}{\smallskip \noindent\textsc{Method: }}
\newcommand{\algoskip}{\vspace{2pt}}

\newcommand{\nlit}[1]{\f{lit}(#1)}
\newcommand{\nclause}[1]{\f{clause}(#1)}
\newcommand{\nparent}[1]{\f{parent}(#1)}
\newcommand{\nside}[1]{\f{side}(#1)}
\newcommand{\ntgt}[1]{\f{tgt}(#1)}
\newcommand{\nipol}[1]{\f{ipol}(#1)}
\newcommand{\nbranch}[2]{\f{branch}_{#1}(#2)}
\newcommand{\ncopy}[1]{\f{copy}(N)}

\newcommand{\ncode}[1]{\f{path\hyph string}(#1)}
\newcommand{\nbadlits}[1]{\f{bad\hyph literals}(#1)}

\newcommand{\poss}[1]{\f{pos}(#1)}
\newcommand{\subform}[2]{#1|_{#2}}

\newcommand{\bpatt}{binding pattern\xspace}
\newcommand{\bpatts}{binding patterns\xspace}

\newcommand{\bpp}[1]{\f{bp}(#1)}

\newcommand{\RQFO}{RQFO\xspace}
\newcommand{\RQFOT}{RQFOT\xspace}

\newcommand{\vout}{\mathbf{v}}

\newcommand{\dx}{\mathbf{x}}

\newcommand{\ddp}{\f{d}}

\newcommand{\eeq}{\f{e}}
\newcommand{\skf}{\f{f}}
\newcommand{\skg}{\f{g}}

\newcommand{\DEFP}[1]{\f{DEF}(#1)}
\newcommand{\DEFPNOT}[1]{\f{DEF}(\lnot #1)}
\newcommand{\defp}[1]{\f{def}_{#1}(F)}

\definecolor{tcolaaa}{rgb}{0.9,0.0,0.0}
\definecolor{tcolbbb}{rgb}{0.0,0.0,0.9}

\newcommand{\taaa}[1]{\textcolor{tcolaaa}{#1}}
\newcommand{\tbbb}[1]{\textcolor{tcolbbb}{#1}}

\definecolor{tcolaaabg}{rgb}{0.9,0.9,0.9}
\definecolor{tcolbbbbg}{rgb}{1.0,1.0,1.0}

\newcommand{\vbar}{\raisebox{-0.60ex}{\rule{0pt}{2.35ex}}}

\renewcommand{\taaa}[1]{\colorbox{tcolaaabg}{$#1\vbar$}}
\renewcommand{\tbbb}[1]{\colorbox{tcolbbbbg}{$#1\vbar$}}

\newcommand{\taaatxt}[1]{\colorbox{tcolaaabg}{#1}}

\newcommand{\dom}[1]{\f{dom}(#1)}
\newcommand{\rng}[1]{\f{rng}(#1)}

\newcommand{\revsubst}[2]{#1 \la #2^{-1} \ra}
\newcommand{\revsubststandalone}[1]{\la #1^{-1} \ra}
\newcommand{\toprevsubst}[2]{#1 \la\!\la #2^{-1} \ra\!\ra}

\newcommand{\pred}[1]{\mathcal{V}_{\f{P}}(#1)}
\newcommand{\const}[1]{\mathcal{V}_{\f{C}}(#1)}
\newcommand{\fun}[1]{\mathcal{V}_{\f{F}}(#1)}
\newcommand{\lit}[1]{\mathcal{L}(#1)}

\renewcommand{\pred}[1]{\mathcal{P}(#1)}
\renewcommand{\const}[1]{\mathcal{C}(#1)}
\renewcommand{\fun}[1]{\mathcal{F}(#1)}
\renewcommand{\var}[1]{\mathcal{V}(#1)}

\renewcommand{\pred}[1]{\f{pred}(#1)}
\renewcommand{\const}[1]{\f{const}(#1)}
\renewcommand{\fun}[1]{\f{fun}(#1)}
\renewcommand{\var}[1]{\f{var}(#1)}
\renewcommand{\lit}[1]{\f{lit}(#1)}

\newcommand{\violet}[1]{{\color{\colfgviolet}#1}}

\newcommand{\red}[1]{{\color{\colfgred}#1}}
\newcommand{\blue}[1]{{\color{\colfgblue}#1}}

\renewcommand{\violet}[1]{#1}
\renewcommand{\red}[1]{#1}
\renewcommand{\blue}[1]{#1}

\newcommand{\F}{\red{F}}
\newcommand{\G}{\blue{G}}
\renewcommand{\H}{\violet{H}}
\newcommand{\FA}{\red{F_{\aaa}}}
\newcommand{\FB}{\blue{F_{\bbb}}}
\newcommand{\lnotFB}{\blue{\lnot F_{\bbb}}}
\newcommand{\aaa}{\red{\f{red}}}
\newcommand{\bbb}{\blue{\f{blue}}}

\newcommand{\nbranchA}[1]{\red{\f{branch}_{\aaa}(#1)}}
\newcommand{\nbranchB}[1]{\blue{\f{branch}_{\bbb}(#1)}}

\renewcommand{\aaa}{\mathsc{*L}}
\renewcommand{\bbb}{\f{R}}

\newcommand{\emptysubst}{\f{Identity}}

\newcommand{\sterm}{\text{-term}}

\newcommand{\sided}{two-sided\xspace}
\newcommand{\Sided}{Two-Sided\xspace}

\renewcommand{\aaa}{\f{L}}
\renewcommand{\bbb}{\f{R}}

\newcommand{\xxsq}{\tup{x}} 
\newcommand{\yysq}{\tup{y}}

\newcommand{\xxsa}{\tup{u}} 
\newcommand{\yysa}{\tup{v}}

\newcommand{\stsk}{\phi}    
\newcommand{\stskf}{\stsk|_{\yysq}}    
\newcommand{\stskg}{\stsk|_{\xxsq}}  
\newcommand{\stski}{\phi}   

\newcommand{\stz}{\lambda}           
\newcommand{\stzig}{\stz|_{\yysq}}   
\newcommand{\stzif}{\stz|_{\xxsq}}   

\newcommand{\stt}{\mu}           
\newcommand{\stig}{\stt|_{\yysa}}   
\newcommand{\stif}{\stt|_{\xxsa}}   

\newcommand{\sth}{\eta}  

\newcommand{\stv}{\sigma} 
\newcommand{\stren}{\rho} 

\newcommand{\CTIG}{CTI\xspace}
\newcommand{\CTIF}{CTI\xspace}

\makeatletter
\newlength{\negph@wd}
\DeclareRobustCommand{\negphantom}[1]{%
  \ifmmode
    \mathpalette\negph@math{#1}%
  \else
    \negph@do{#1}%
  \fi
}
\newcommand{\negph@math}[2]{\negph@do{$\m@th#1#2$}}
\newcommand{\negph@do}[1]{%
  \settowidth{\negph@wd}{#1}%
  \hspace*{-\negph@wd}%
}
\makeatother

\newwrite\keyslogfile
\openout\keyslogfile=\jobname.log.prkeys

\newcommand{\prset}[2]
           {\pgfkeysifdefined{/prlabval/#1}
             {\write\keyslogfile{WARNING Duplicate: #1}}
             {\pgfkeyssetvalue{/prlabval/#1}{#2}}}

\IfFileExists{./\jobname.aux.prkeys}{\input{\jobname.aux.prkeys}}{}
\newwrite\keysfile
\openout\keysfile=\jobname.aux.prkeys

\newcounter{prlabcounter}

\newcommand{\prlReset}[1]
           {\pgfkeyssetvalue{/prlabcontext}{#1}
             \setcounter{prlabcounter}{0}}

\prlReset{default}

\newcommand{\prl}[1]
           {\stepcounter{prlabcounter}(\theprlabcounter)%
             \immediate\write\keysfile{%
               \unexpanded{\prset}{\pgfkeysvalueof{/prlabcontext}:#1}{\theprlabcounter}}}

\newcommand{\pref}[1]
           {\pgfkeysifdefined{/prlabval/\pgfkeysvalueof{/prlabcontext}:#1}
             {(\pgfkeysvalueof{/prlabval/\pgfkeysvalueof{/prlabcontext}:#1})%
               \immediate\write\keyslogfile{Referenced: \pgfkeysvalueof{/prlabcontext}:#1}}
             {(??)%
               \immediate\write\keyslogfile{WARNING Undefined: \pgfkeysvalueof{/prlabcontext}:#1}}}
           
\newcommand{\prefglobal}[1]
           {\pgfkeysifdefined{/prlabval/#1}%
             {(\pgfkeysvalueof{/prlabval/#1})}%
             {(??)}}

\newcommand{\ic}{\f{c}}
\newcommand{\id}{\f{d}}
\newcommand{\ie}{\f{e}}
\newcommand{\ia}{\f{t}}
\newcommand{\ig}{\f{b}}
\newcommand{\iq}{\f{q}}

\newcommand{\inject}[1]{#1^*}

\newcommand{\hparenc}{\hphantom{\{}}

\renewcommand{\NG}{G^{\lnot}}

\newcommand{\nand}{\uparrow}

\newcommand{\casesection}[1]
           {\paragraph{\textbf{#1.}}}

\newlist{caselist}{description}{2}
\setlist[caselist]
        {leftmargin=1.5em,font=\itshape,style=nextline,itemsep=1.0ex,parsep=0pt}

\newcommand{\casepara}[1]{\item[#1:]}
           
\newcommand{\waicond}{WAI\xspace}

\newcommand{\LL}{\f{L}}
\newcommand{\RR}{\f{R}}
\newcommand{\XS}{s}

\newcommand{\stx}{\theta}

\newcommand{\stcr}{\sigma}

\newcommand{\sti}{\mu}

\newcommand{\extabscale}{0.5}
\newcommand{\extabld}{10ex}

\newcommand{\dk}{\f{d}}
\newcommand{\rk}{\f{r}}
\newcommand{\sk}{\f{s}}
\newcommand{\ffk}{\f{f}}
\newcommand{\ggk}{\f{g}}

\newcommand{\pk}{\f{p}}
\newcommand{\hk}{\f{h}}
\newcommand{\ak}{\f{a}}
\newcommand{\bk}{\f{b}}

\begin{document}

\includepdf[pages={1-2}]{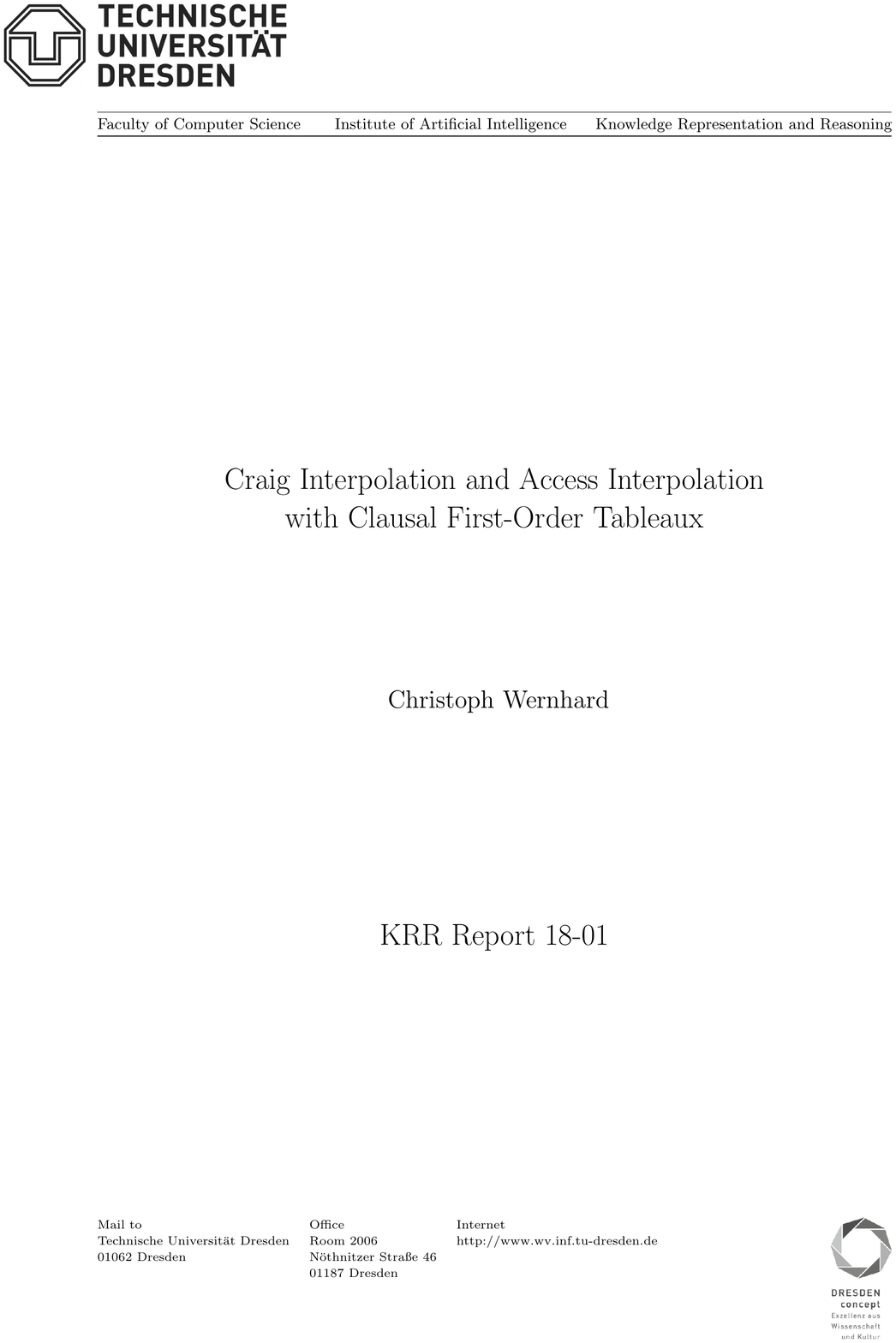}
\setcounter{page}{1}

\maketitle

\begin{abstract}
  We develop foundations for computing Craig interpolants and similar
  intermediates of two given formulas with first-order theorem provers that
  construct clausal tableaux.  Provers that can be understood in this way
  include efficient machine-oriented systems based on calculi of two families:
  goal-oriented like model elimination and the connection method, and
  bottom-up like the hyper tableau calculus.  The presented method for
  Craig-Lyndon interpolation involves a lifting step where terms are replaced
  by quantified variables, similar as known for resolution-based
  interpolation, but applied to a differently characterized ground formula and
  proven correct more abstractly on the basis of Herbrand's theorem,
  independently of a particular calculus.  Access interpolation is a recent
  form of interpolation for database query reformulation that applies to
  first-order formulas with relativized quantifiers and constrains the
  quantification patterns of predicate occurrences.  It has been previously
  investigated in the framework of Smullyan's non-clausal tableaux.  Here, in
  essence, we simulate these with the more machine-oriented clausal tableaux
  through structural constraints that can be ensured either directly by
  bottom-up tableau construction methods or, for closed clausal tableaux
  constructed with arbitrary calculi, by postprocessing with restructuring
  transformations.

\end{abstract}

\section{Introduction}
\label{sec-intro}
By Craig's interpolation theorem \cite{craig:linear}, for two first-order
formulas $F$ and $G$ such that $F$ entails $G$ there exists a third
first-order formula $H$ that is entailed by $F$, entails $G$ and is such that
all predicate and function symbols occurring in it occur in both $F$ and $G$.
Such a \name{Craig interpolant}~$H$ can be \emph{constructed} from given
formulas $F$ and $G$, for example by a calculus that allows to extract~$H$
from a proof that $F$ entails $G$, or, equivalently, that the implication $F
\imp G$ is valid.  Automated construction of interpolants has many
applications, in the area of computational logic most notably in symbolic
model checking, initiated with \cite{mcmillan:2003}, and in query
reformulation
\cite{marx:2007,nash:2010,borgida:2010,toman:wedell:book,benedikt:guarded,benedikt:etal:2014:generating,toman:2015:tableaux,benedikt:book,benedikt:2017,toman:2017}.
The foundation for the latter application field is the observation that a
reformulated query can be viewed as a \emph{definiens} of a given query where
only symbols from a given set, the target language of the reformulation, occur
in the definiens.  The existence of such definientia, that is, definability
\cite{tarski:35}, or \name{determinacy} as it is called in the database
context, can be expressed as validity and their synthesis as interpolant
construction. For example, a \name{definiens~$H$ of a unary predicate~$\pk$
  within a first-order formula~$F$} can be characterized by the following
conditions:
\begin{enumerate}
\item $F$ entails $\forall x\, (\pk(x) \equi H)$.
\item $\pk$ does not occur in  $H$.
\end{enumerate}
The variable $x$ is allowed there to occur free in $H$. We further assume that
$x$ does not occur free in $F$ and let $F^\prime$ denote $F$ with $\pk$ replaced
by a fresh symbol $\pk^\prime$. Now the characterization of \name{definiens} by
the two conditions given above can be equivalently expressed as
\begin{center}
\name{$H$ is a
  Craig interpolant of the two formulas $F \land \pk(x)$ and $\lnot (F^\prime
  \land \lnot \pk^\prime(x))$}.
\end{center}
A definiens $H$ exists if and only if it is
valid that the first formula implies the second one.

The construction of Craig interpolants of given first-order formulas has been
elegantly specified in the framework of tableaux by Smullyan
\cite{smullyan:book:68,fitting:book}. Although this has been taken as
foundation for applications of interpolation in query reformulation
\cite{toman:wedell:book,benedikt:book}, it has been hardly used as a basis for
the practical computation of first-order interpolants with automated reasoning
systems, where the focus so far has been on interpolant extraction from
specially constrained resolution proofs (see
\cite{bonacina:15:on:ipol,kovacs:17} for recent overviews and discussions).

Here we approach the computation of interpolants from another paradigm of
automated reasoning, the construction of a \name{clausal tableau}.
Expectations are that, on the one hand, the elegance of Smullyan's
interpolation method for non-clausal tableaux can be utilized and, on the
other hand, the foundation for efficient practical implementations is laid.
Various efficient theorem proving methods can be viewed as operating by
constructing a clausal tableau \cite{handbook:tableaux:letz} (or \name{clause
  tableau} \cite{handbook:ar:haehnle}). They can be roughly divided into two
major families: First, methods that are goal-sensitive, typically proceeding
with the tableau construction ``top-down'', by ``backward chaining'', starting
with clauses from the theorem in contrast to the axioms.  Aside of clausal
tableaux in the literal sense, techniques to specify and investigate such
methods include model elimination \cite{loveland:1969}, the connection method
\cite{bibel:1981}, and the Prolog technology theorem prover \cite{pttp}.  One
of the leading first-order proving systems of the 1990s, \name{SETHEO}
\cite{setheo}, followed that approach.  The \name{leanCoP} system
\cite{leancop} along with its recent derivations
\cite{kaliszyk15:tableaux,femalecop} as well as the \name{CM} component of
\name{PIE} \cite{cw-mathlib,cw-pie} are implementations in active duty today.
The second major family of methods constructs clausal tableaux ``bottom-up'',
in a ``forward-chaining'' manner, by starting with positive axioms and
deriving positive consequences.  With the focus of their suitability to
construct model representations, these methods have been called
\name{bottom-up model generation (BUMG)} methods \cite{bumg}.  They include,
for example, \name{SATCHMO} \cite{satchmo} and the hyper tableau calculus
\cite{hypertab}, with implementations such as \name{Hyper}, formerly called
\name{\mbox{E-KRHyper}} \cite{cw-ekrhyper,cw-krhyper,hyper:2013}.
Hyper tableau methods
are also used in high-per\-for\-mance description logic reasoners
\cite{dl:hypertab}.  It appears that the chase method from the database field,
which recently got attention anew in knowledge representation (see, e.g.,
\cite{grau:2013:acyclicity}), can also be understood as such a bottom-up
tableau construction. Methods of the instance-based approach to theorem
proving (see \cite{baumgartner:2010:ibased} for an overview) should in general
be applicable to construct a clausal tableau after proving, from the instances
involved in the proof, although the proof construction itself might not
proceed by tableau construction.  For a systematic overview of different
variants of tableaux structures and methods, including clausal tableaux with
respect to both considered major paradigms see \cite{handbook:ar:haehnle}.

An essential distinction of clausal tableau methods from resolution-based
methods is that at the tableau construction only instances of \emph{input
  clauses} are created and incorporated.  Clauses are not broken apart and
joined as in a resolution step. Nevertheless, clausal tableau methods might be
complemented by preprocessors that perform such operations.
An essential distinction from non-clausal tableau methods is that with the
clausal form only a particularly simple formula structuring has to be
considered, in essence sets of clauses. Through preprocessing with conversion
to prenex form and Skolemization, the handling of quantifications amounts for
clausal tableau methods just to the handling of free variables.

The tableau-based method for Craig interpolation presented here proceeds in
two stages, with some similarity to resolution-based methods discussed in
\cite{huang:95,baaz:11,bonacina:15:on:ipol,kovacs:17} that compute in a first
stage a so-called \name{relational}, \name{weak} or \name{provisional}
interpolant which satisfies the vocabulary restriction on interpolants with
respect to predicate symbols but not necessarily with respect to function and
constant symbols. The result of the first stage is in the second stage lifted
to an actual interpolant of the original input formulas by replacing terms
with variables and prepending a specific quantifier prefix.  In our
tableau-based method the two stages are separated at a different place, more
directly related to Herbrand's theorem, without need of an additional notion
such as \name{relational interpolant}. In the first stage an actual Craig
interpolant of a finite unsatisfiable subset of the Herbrand expansion of the
Skolemized and clausified input formulas is constructed. The involved ground
clauses can be obtained as instances of clauses of the closed tableau computed
by a first-order prover for a set of first-order clauses.  With respect to
interpolation, the closed clausal tableau can be considered just as
\emph{given}, abstracting from the method by which it has been constructed.
This leads to a lean formalism for interpolation and justifies the practical
implementation of Craig interpolation with arbitrary high-performance
first-order theorem provers that construct clausal tableaux, without need to
modify inference rules or other prover internals.

There are many known ways to strengthen Craig's interpolation theorem by
ensuring that for given formulas~$F$ and~$G$ that satisfy certain syntactic
restrictions there exists an interpolant~$H$ that also satisfies certain
syntactic restrictions. For example, that predicates occur in~$H$ only with
\emph{polarities} with which they occur in both $F$ and $G$.  (A predicate
occurs with \name{positive} (\name{negative}) polarity in a formula if it
occurs there in the scope of an even (odd) number of negation operators.)  The
respective strengthened interpolation theorem has been explicated by Lyndon
\cite{lyndon}, hence we call Craig interpolants that meet this restriction
\name{Craig-Lyndon interpolants}.
\name{Access interpolation} \cite{benedikt:book} is a variant of Craig-Lyndon
interpolation that applies to formulas in which quantifiers only occur
relativized by atoms, as for example in
\begin{equation}
\label{ex-relat}
\forall x\, (\f{r}(x) \imp \exists y \exists z\, (\f{s}(x,y,z) \land \mathit{true})).
\end{equation}
With each occurrence of a relativizing atom a \defname{binding pattern} or
\defname{access pattern} is associated, which comprises the predicate, the
polarity of the occurrence and the argument positions of those variables that
are not quantified by the associated quantifier.  For example,
in~(\ref{ex-relat}) we have for the occurrence of~$\f{r}(x)$ the predicate
$\f{r}$ in negative polarity with the empty set of argument positions and for
the occurrence of $\f{s}(x,y,z)$ the predicate $\f{s}$ in positive polarity
and the set~$\{1\}$ of argument positions, because $x$ at the first argument
position in the occurrence of $\f{s}(x,y,z)$ is not quantified by $\exists y
\exists z$.  Positions specified in the set are also called \name{input}
positions, while the quantified positions are \name{output} positions,
corresponding to their role in a naive formula evaluation.  Access
interpolation strengthens Craig-Lyndon interpolation by requiring that also
the binding patterns occurring in the interpolant formula are subsumed by
binding patterns occurring in a specific way in the input formulas.

In \cite{benedikt:book} it has been shown that many tasks in database query
reformulation can be expressed in terms of access interpolation, applied to
construct definientia of queries that are in a certain vocabulary and involve
only certain binding patterns which makes them evaluable in a certain sense.
A variant of Craig-Lyndon interpolation by Otto \cite{otto:interpolation:2000}
has been suggested in \cite{nash:2010} as a technique to take relativization
into account. In \cite{benedikt:book} access interpolation is presented as a
generalization of Otto's interpolation and constructively proven on the basis
of Smullyan's tableau method following the presentation in
\cite{fitting:book}.

Access interpolants involve only relativized quantification, which seems
incompatible with a global quantifier prefix as computed by the lifting
technique sketched above for Craig interpolation, at least if predicates used
as relativizers are permitted to have empty extensions.\footnote{If
  relativizer predicates are assumed to have nonempty extensions, quantifiers
  together with their relativizing literals can be moved to the prefix,
  justified in essence by the following entailments shown here for a unary
  relativizer predicate $r$, but holding analogously also for relativizers
  with larger arity.  If $F$ and $G$ are formulas such that $x$ does not occur
  free in $G$, then:
\[
\begin{array}{rcl}
\exists x\, r(x) & \entails &
(\exists x\, (r(x) \land F) \lor G)
\equi
\exists x\, (r(x) \land (F \lor G)).\\
\exists x\, r(x) & \entails &
(\forall x\, (r(x) \imp F) \land G)
\equi
\exists x\, (r(x) \imp (F \land G)).
\end{array}
\]}
Hence, the method for access interpolation presented here extracts the
interpolant from a tableau in a single stage, where a form of lifting that
only applies to subformulas corresponding to scopes of relativized quantifiers
is incorporated.  In essence, Smullyan's techniques for non-clausal tableau
are simulated with the more machine-oriented clausal tableaux and variable
handling through Skolemization. Correspondence to Smullyan's tableaux is
achieved by a structure preserving normal form and certain structural
requirements on the clausal tableaux. These are already met by hyper tableaux.
In the general case they can be ensured with restructuring transformations,
applied in a postprocessing step to closed clausal tableaux obtained from
provers.

The contributions of this work can be summarized as follows:
\begin{enumerate}
  
\item Foundations to perform Craig interpolation and related forms of
  interpolation for first-order logic with clausal tableau methods are
  developed. They provide:
\begin{enumerate}
\item A basis for implementing interpolation with efficient machine-oriented
  theorem provers for first-order logic that can be understood as constructing
  clausal tableaux. With methods and systems of two main families,
  goal-oriented ``top-down'' and forward-chaining ``bottom-up'', there is a
  wide range of potential applications.

\item A relatively simple framework to prove constructively the existence of
  interpolants with further syntactic properties, beyond the restriction on
  symbols required by Craig interpolants.  The involved constructions are,
  moreover, suited for realization by practical systems. In the paper such
  constructions are shown for Craig-Lyndon interpolation, interpolation from a
  Horn formula, and, with access interpolation, for a form of quantifier
  relativization.
\end{enumerate}
  
\item Interpolant lifting, which is in principle known from resolution-based
  approaches since the mid-nineties, is placed at a new and apparently more
  natural position within the overall task of first-order interpolation, where
  it is independent of a particular calculus. A detailed correctness proof
  that resides on a small technical basis is presented.

\item For access interpolation, a key technique for query reformulation, the
  first practically implementable methods are described.  
  
\item Conversions between closed clausal tableaux are developed that transform
  arbitrarily structured inputs to clausal tableaux with a restricted
  structure that in essence simulates non-clausal tableaux or tableaux that
  are constrained in specific ways, as, for example, computed by hyper tableau
  methods. They justify the application of practical methods that construct
  unrestricted clausal tableaux, such as, for example, goal-oriented
  ``top-down'' first-order theorem proving methods, to tasks like access
  interpolation which require a certain tableau structuring.

\end{enumerate}

Proofs are given for all theorem, lemma and proposition statements that do not
pertain to the considered logics in general. Proofs which involve intricacies
or subtleties are given in detail.

The rest of this paper is structured in two main parts:
Sections~\ref{sec-notation}--\ref{sec-cli-related} are concerned with
Craig-Lyndon interpolation and Sections~\ref{sec-ai}--\ref{sec-ai-conclusion}
with access interpolation.  After notation and basic terminology have been
specified in Sect.~\ref{sec-notation}, precise accounts of \name{clausal
  tableau} and related notions are given in Sect.~\ref{sec-tableaux}.  In
Sect.~\ref{sec-ipol-basic} the extraction of ground interpolants from closed
clausal ground tableaux is specified and proven correct.  The generalization
of this method to first-order formulas, which involves preprocessing by
Skolemization and postprocessing of ground interpolants by lifting is
specified and proven correct in Sect.~\ref{sec-ipol-lift}, and in
Sect.~\ref{sec-lift-related} compared with related approaches from the
literature.  In Sect.~\ref{sec-tab-constraints} constraints on clausal
tableaux are specified that characterize positive hyper tableaux, which are
typically computed by ``bottom-up'' methods. On this basis a construction of
Craig-Lyndon interpolants that inherit the Horn property from the first
interpolation input is shown.  Section~\ref{sec-cli-related} concludes the
part on Craig-Lyndon interpolation with a discussion of possible refinements
of our method and issues for further research. We then turn to access
interpolation.  In Sect.~\ref{sec-ai} a brief overview on our approach is
given, underlying notions from the literature are recapitulated, and a
structure-preserving clausal normalization of the relativized input formulas
is described.  The extraction of an access interpolant from a closed clausal
ground tableau that is for such clauses and meets certain structural
constraints is then specified and proven correct in
Sect.~\ref{sec-access-extract}.  These structural constraints are met by
positive hyper tableaux. For the general case they can be ensured with tableau
transformations, specified in Sect.~\ref{sec-access-convert} and illustrated
with examples in Sect.~\ref{sec-access-convert-examples}.
Section~\ref{sec-ai-conclusion} concludes the part on access interpolation
with a discussion of possible refinements of our method, issues for further
research and related work.  Section~\ref{sec-conclusion} concludes the paper
with an abstract view on its main contributions.

A work-in-progress poster of this research at an earlier stage was
presented at the\linebreak \name{TABLEAUX~2017} conference.

\section{Notation and Basic Terminology}
\label{sec-notation}

We basically consider first-order logic without equality.\footnote{This does
  not preclude to represent equality as a predicate with axioms that express
  reflexivity, symmetry, transitivity and substitutivity.}  Atoms are of the
form $p(t_1,\ldots,t_n)$, where $p$ is a \defname{predicate symbol} (briefly
\defname{predicate}) with associated arity $n \geq 0$ and $t_1,\ldots,t_n$ are
terms formed from \defname{function symbols} (briefly \name{functions}) with
associated arity~$\geq 0$ and \defname{individual variables} (briefly
\name{variables}).  Function symbols with arity $0$ are also called
\defname{individual constants} (briefly \name{constants}).

Unless specially noted, a \defname{formula} is understood as a formula of
first-order logic without equality, constructed from atoms, constant operators
$\true$, $\false$, the unary operator $\lnot$, binary operators $\land, \lor$
and quantifiers $\forall, \exists$ with their usual meaning.  Further binary
operators~$\imp$, $\equi$ as well as $n$-ary versions of $\land$ and $\lor$
can be understood as meta-level shorthands.  Also quantification upon a set of
variables is used as shorthand for successive quantification upon each of its
elements.  The operators $\land$ and $\lor$ bind stronger than $\imp$ and
$\equi$. The scope of $\lnot$, the quantifiers, and the $n$-ary connectives is
the immediate subformula to the right.  Formulas in which no functions with
exception of constants occur are called \defname{relational}.  Formulas in
which no predicates with arity larger than zero and no quantifiers occur are
called \defname{propositional}.

A subformula occurrence has in a given formula \defname{positive (negative)
  polarity}, or is said to occur \defname{positively (negatively)} in the
formula, if it is in the scope of an even (odd) number of negations.  If $E$
is a term or a formula, then the set of variables that occur \defname{free}
in~$E$ is denoted by $\var{E}$, the set of functions occurring in~$E$ by
$\fun{E}$, and the set of constants occurring in~$E$ by $\const{E}$.  If $F$
is a formula, then the set of pairs of predicates occurring in~$F$ coupled
with an identifier of the respective polarity of the atom in which they occur
is denoted by $\pred{F}$, the set of pairs of atoms occurring in $F$ coupled
with an identifier of the respective polarity in which they occur as
$\lit{F}$, and the set of terms that occur as argument of a predicate (in
contrast to just as argument of a function) as $\parg{F}$.
The notation $\var{E}$, $\fun{E}$ and $\const{E}$ is also used with sets~$E$
of terms or formulas, where it stands for the union of values of the
respective function applied to each member of $E$.  A formula without free
variables is called a \defname{sentence}.  A term or quantifier-free formula
in which no free variable occurs is called \defname{ground}. A ground formula
is thus a special case of a sentence.  Symbols not present in the formulas and
other items under discussion are called \defname{fresh}.

A \defname{literal} is an atom or a negated atom.  If $A$ is an atom, then the
\name{complement} of $A$ is $\lnot A$ and the complement of $\lnot A$ is $A$.
The complement of a literal~$L$ is denoted by $\du{L}$.
A \defname{clause} is a (possibly empty) disjunction of literals.  A
\defname{clausal formula} is a (possibly empty) conjunction of clauses, called
the \defname{clauses in} the formula.

The notion of \name{substitution} used here follows
\cite{baader:snyder:unificationtheory}: A \defname{substitution} is a mapping
from variables to terms which is almost everywhere equal to identity.  If
$\sigma$ is a substitution, then the \defname{domain} of $\sigma$ is the set
of variables $\dom{\sigma} \eqdef \{x \mid x\sigma \neq x\}$, the
\defname{range} of $\sigma$ is $\rng{\sigma} \eqdef \bigcup_{x \in
  \dom{\sigma}} \{x\sigma\}$, and the \defname{restriction} of $\sigma$ to a
set~$\xxs$ of variables, denoted by $\sigma|_{\xxs}$, is the substitution
which is equal to the identity everywhere except over $\xxs \cap
\dom{\sigma}$, where it is equal to~$\sigma$.  The identity substitution is
denoted by $\emptysubst$.  A substitution can be represented as a function by
a set of bindings of the variables in its domain, e.g., $\{x_1 \mapsto t_1,
\ldots, x_n \mapsto t_n\}$.  The application of a substitution~$\sigma$ to a
term or a formula~$E$ is written as $E\sigma$, $E\sigma$ is called an
\defname{instance} of $E$ and $E$ is said to \defname{subsume} $E\sigma$.
Composition of substitutions is written as juxtaposition. Hence, if $\sigma$
and $\gamma$ are both substitutions, then $E\sigma\gamma$ stands for
$(E\sigma)\gamma$.

\label{sec-revsubst}
For injective substitutions we use the following additional notation: If
$\sigma$ is an injective substitution and $E$ is a term or a formula, then
$\revsubst{E}{\sigma}$ denotes $E$ with all occurrences of subterms~$s$ that
are in the range of $\sigma$ and are not a strict subterm of another subterm
in the range of $\sigma$ replaced by the variable that is mapped by~$\sigma$
to $s$.
As an example let $\sigma = \{x \mapsto f(a),\, y \mapsto g(f(a))\}$. Then
 \[\revsubst{p(h(f(a),g(f(a))))}{\sigma} = p(h(x,y)).\]

The \defname{principal functor} of a term that is not a variable is its
outermost function symbol.  If $S$ is a set of function symbols, then a term
with a principal functor in $S$ is also called an \defname{$S$-term}.

We write $F \entails G$ for \name{$F$ entails $G$}; $\valid F$ for \name{$F$
  is valid}; and $F \equiv G$ for \name{$F$ is equivalent to $G$}, that is,
\name{$F \entails G$ and $G \entails F$}.  On occasion we write a sequence of
statements with these operators where the right and left, respectively,
arguments of subsequent statements are identical in a chained way, such as,
for example, $F \entails G \entails H$ for \name{$F \entails G$ and $G
  \entails H$}.

\newpage
\section{Clausal First-Order Tableaux}
\label{sec-tableaux}

The following definition makes the variant of clausal tableaux that we use as
basis for interpolation precise. It is targeted at modeling tableau structures
produced by efficient fully automated first-order proving systems based on
different calculi.\hspace{-1em}
\begin{defn}[Clausal Tableau and Related Notions]

\sdlab{def-tab}
Let $F$ be a clausal formula.  A \defname{clausal tableau} (briefly
\name{tableau}) \defname{for} $F$ is a finite ordered tree whose nodes~$N$
with exception of the root are labeled with a literal, denoted by
$\nlit{N}$, such that the following condition is met: For each node~$N$ of the
tableau the disjunction of the labels of all its children in their
left-to-right order, denoted by $\nclause{N}$, is an instance of a clause
in~$F$. A value of $\nclause{N}$ for a node~$N$ in a tableau is called a
\defname{clause of} the tableau.

\smallskip

\sdlab{def-tab-closed} A node~$N$ of a tableau is called \defname{closed} if
and only if it has an ancestor~$N^\prime$ with $\nlit{N^\prime} =
\du{\nlit{N}}$.  With a closed node~$N$, a particular such ancestor $N^\prime$
is associated as \defname{target of}~$N$, written~$\ntgt{N}$.  A tableau is
called \defname{closed} if and only if all of its leaves are closed.

\smallskip

\sdlab{def-tab-ground} A tableau is called \defname{ground} if and only if for
all its nodes~$N$ it holds that $\nlit{N}$ is ground.
\end{defn}
The most immediate relationship of clausal tableaux to the semantics of
clausal formulas is that the universal closure of a clausal formula is
unsatisfiable if and only if there exists a closed clausal tableaux for the
clausal formula. Knowing that there are sound and complete calculi that
operate by constructing a closed clausal tableau for an unsatisfiable clausal
formula, and taking into account Herbrand's theorem we can state the following
proposition:%
\hspace{-10pt}
\begin{prop}[Unsatisfiability and Computation of Closed Clausal Tableaux]
\label{prop-tab-complete}
There is an effective method that computes from a clausal formula~$F$ a closed
clausal tableau for~$F$ if and only if $\forall x_1 \ldots \forall x_n\, F$,
where $\{x_1,\ldots,x_n\} = \var{F}$, is unsatisfiable.  Moreover, this also
holds if terms in the literal labels of tableau nodes are constrained to
ground terms formed from functions occurring in~$F$ and, in case there is no
constant occurring in $F$, an additional fresh constant.

\end{prop}

Our objective is here interpolant construction on the basis of clausal
tableaux produced by fully automated systems. This has effect on some aspects
of our formal notion of \name{clausal tableau}: All occurrences of variables
in the literal labels of a tableau according to Definition~\ref{def-tab} are
free and the scope of these variables spans all literal labels of the whole
tableau.  In more technical terms, this means that the tableaux are \name{free
  variable tableaux} (see \cite[p.~158ff]{handbook:tableaux:letz}) with
\name{rigid} variables (see \cite[p.~114]{handbook:ar:haehnle}).  Tableaux
with only clause-local variables can, however, of course be expressed by just
using different variables in each tableau clause. Thus, although our notion of
tableaux involves rigid variables, this does not in any way imply that
interpolant computation based on it applies only to tableaux whose
\emph{construction} by a prover had involved rigid variables.

Another aspect concerns the definition of \name{closed} for nodes and for
tableaux: A tableau is closed if all of its \emph{leaves} are closed, which
does, however, not exclude that also an inner node of a closed tableau might
be closed.  For the \emph{construction} of a closed tableau in theorem proving
it is pointless to attach children to an already closed node. In our context,
however, operations such as instantiating literal labels and certain tableau
transformations might introduce inner closed nodes.  To let the results of
such operations be tableaux again, we thus have to permit closed inner nodes.
A tableau simplification to eliminate these is shown in
Sect.~\ref{sec-access-convert}.

\section{Ground Interpolant Extraction from Clausal Tableaux}
\label{sec-ipol-basic}

As shown by Craig \cite{craig:linear}, for first-order sentences $\F$ and $\G$
such that $\F \entails \G$, an ``intermediate'' sentence $H$ such that $F
\entails H \entails G$ can be constructed, whose predicates and functions are
occurring in both $F$ and $G$. That this also holds if in addition the
polarities of predicate occurrences in $H$ are constrained to polarities in
which they occur in both $F$ and $G$ is attributed to Lyndon \cite{lyndon},
such that formulas $\H$ are sometimes called \name{Lyndon interpolants} in
analogy to \name{Craig interpolants}.  We call them here \name{Craig-Lyndon
  interpolants}:
\begin{defn}[Craig-Lyndon Interpolant]
\label{def-cli}
Let $F, G$ be sentences such that $F \entails G$.  A \defname{Craig-Lyndon
  interpolant} of $\F$ and $\G$ is a sentence $\H$ such that
\begin{enumerate}
\item \label{def-cli-sem} $\F \entails \H \entails \G$.
\item \label{def-cli-pred} 
$\pred{H} \subseteq \pred{F} \cap \pred{G}$.
\item \label{def-cli-fun}
$\fun{H} \subseteq \fun{\F} \cap \fun{\G}$.
\end{enumerate}
\end{defn}
The notion of Craig-Lyndon interpolant is specified here for \emph{sentences}
in contrast to \emph{formulas}~$F$, $G$ and $H$. This is without loss of
generality because free variables in~$F$, $G$ and $H$ would, with respect to
interpolation, be handled exactly like constants.

Smullyan \cite{smullyan:book:68} specifies in his framework of non-clausal
tableaux an elegant technique to extract a Craig-Lyndon interpolant from a
tableau that represents a proof of $\F \entails \G$, which is also presented
in Fitting's book \cite{fitting:book}.  The handling of propositional
connectives in this method can be straightforwardly transferred to clausal
tableaux.  Quantifiers, however, have to be processed differently to match
their treatment in clausal tableaux by conversion to prenex form and
Skolemization.
The overall interpolant extraction from a closed clausal tableau then proceeds
in two stages, analogously as described for resolution-based methods in
\cite{huang:95,baaz:11,bonacina:15:on:ipol,kovacs:17}.  In the first stage a
``rough interpolant'' is constructed which needs postprocessing by replacing
terms with variables and prepending a quantifier prefix on these variables to
yield an actual interpolant.  This second stage will be specified in
Sect.~\ref{sec-ipol-lift} and discussed further in
Sect.~\ref{sec-lift-related}.  As we will see now, on the basis of clausal
tableaux the first stage can be specified and verified with proofs by a
straightforward adaption of Smullyan's method in an almost trivially simple
way.

Our interpolant construction is based on a variant of clausal tableaux where
nodes have an additional \name{side} label that is shared by siblings and
indicates whether the tableau clause is an instance of an input clause derived
from the formula of the left side or the formula on the right side of the
entailment underlying the interpolation:
\begin{defn}[\Sided Clausal Tableau and Related Notions]

\sdlab{def-coltab} Let $\FA, \FB$ be clausal formulas.  A \defname{\sided
  clausal tableau for} $\FA$ \defname{and} $\FB$ (or briefly \defname{tableau}
for the \defname{two} formulas) is a clausal tableau for $\FA \land \FB$ whose
nodes~$N$ with exception of the root are labeled additionally with a
\name{side} $\nside{N} \in \{\aaa, \bbb\}$, such that the following conditions
are met:
\begin{enumerate}
\item If $N$ and
$N^\prime$ are siblings, then $\nside{N} = \nside{N^\prime}$.
\item If $N^\prime$ is a child of $N$, then $\nclause{N}$ is an instance of a
 clause in $F_{\nside{N^\prime}}$.
\end{enumerate}
The \name{side} of a clause~$\nclause{N}$ in a tableau is the value of the
side label of the children of~$N$.

\medskip

\sdlab{def-branch}
For $S \in \{\aaa,\bbb\}$ and nodes $N$ of a \sided clausal tableau define
\[\nbranch{S}{N}\; \eqdef \bigwedge_{N^\prime \in B \text{ and }
\nside{N^\prime} = S}\hspace{-3em} \nlit{N^\prime},\] where $B$ is the union
of $\{N\}$ with the set of the ancestors of $N$.
\end{defn}
The following definition specifies an adaption of the handling of
propositional connectives in \cite[Chap.~XV]{smullyan:book:68} and
\cite[Chap.~8.12]{fitting:book} to construct interpolants from non-clausal
tableaux. Differently from these works, the specification is here not in terms
of tableau manipulation rules that deconstruct the tableau bottom-up, but
inductively, as a function that maps a node to a formula.
\begin{defn}[Interpolant Extraction from a Clausal Ground Tableau]
  \label{def-ipol}
Let $N$ be a node of a closed \sided clausal ground tableau.  The value of
\defname{$\nipol{N}$} is a ground formula, defined inductively as follows:
\begin{enumerate}[label={\roman*}.,leftmargin=2em]
\item If $N$ is a leaf, then the value of $\nipol{N}$ is determined by the
  values of $\nside{N}$ and $\nside{\ntgt{N}}$ as specified in the
  following table:
\end{enumerate}
\[
\begin{array}{c@{\hspace{1em}}c@{\hspace{1em}}c}
\nside{N} &  \nside{\ntgt{N}} & \nipol{N}\\\midrule
\aaa & \aaa & \false\\[0.5ex]
\aaa & \bbb & \nlit{N}\\[0.5ex]
\bbb & \aaa & \du{\nlit{N}}\\[0.5ex]
\bbb & \bbb & \true
\end{array}
\]

\begin{enumerate}[label={\roman*}.,leftmargin=2em]
\setcounter{enumi}{1}
\item \label{item-children} If $N$ is an inner node with children $N_1,
  \ldots, N_n$ where $n \geq 1$, then the value of $\nipol{N}$ is composed
  from the values of $\f{ipol}$ for the children, disjunctively or
  conjunctively, depending on the side label of the children (which is the
  same for all of them), as specified in the following table:
\end{enumerate}
\[
\begin{array}{c@{\hspace{1em}}c}
\nside{N_1} & \nipol{N}\\\midrule
\aaa & \bigvee_{i=1}^n \nipol{N_i}\\[1ex]
\bbb & \bigwedge_{i=1}^n \nipol{N_i}
\end{array}
\]
\end{defn}

The following lemma associates semantic and syntactic properties with the
formula obtained as value of applying $\f{ipol}$ to the root of a closed
ground tableau. These properties imply the conditions required from a
Craig-Lyndon interpolant (Definition~\ref{def-cli}).
\begin{lem}[Correctness of Interpolant Extraction from 
Clausal Ground Tableaux]
\label{lem-ground-ipol-correct}
Let $\FA, \FB$ be clausal ground formulas and let $T$ be a closed \sided
clausal ground tableau for $\FA$ and $\FB$. If $N$ is the root of $T$,
then 
\begin{enumerate}
\item $\FA \entails \violet{\nipol{N}} \entails \lnotFB$.
\item $\lit{\violet{\nipol{N}}} \subseteq
\lit{\FA} \cap \lit{\lnotFB}$.
\end{enumerate}
\end{lem}
\begin{proof}
We show the following property of $\f{ipol}$ that invariantly holds for all
nodes of the tableau, including the root, which immediately implies the
proposition: For all nodes~$N$ of $T$ it holds that

\begin{enumerate}[label={(\alph*)}]
\item \label{item-lem-gi-sem}
 $\red{\FA \land \nbranchA{N}} \entails \violet{\nipol{N}} 
 \entails \blue{\lnotFB  \lor \lnot \nbranchB{N}}$.
\item \label{item-lem-gi-syn}
$\lit{\violet{\nipol{N}}} \subseteq
\lit{\red{\FA \land \nbranchA{N}}} \cap
\lit{\blue{\lnotFB  \lor \lnot \nbranchB{N}}}$.
\end{enumerate}

\noindent
This is proven by induction on the tableau structure, proceeding from leaves
upwards.  We prove the base case, where $N$ is a leaf, by showing
\ref{item-lem-gi-sem} and \ref{item-lem-gi-syn} for all possible values of
$\nside{N}$:

  \begin{itemize}
  \item Case $\nside{N} = \aaa$:
    \smallskip
    \begin{itemize}
    \item Case $\nside{\ntgt{N}} = \aaa$: Immediate since then
      $\nbranchA{N} \entails \false = \violet{\nipol{N}}$.
    \item Case $\nside{\ntgt{N}} = \bbb$: Then $\violet{\nipol{N}} = \nlit{N}$.
      Properties \ref{item-lem-gi-sem} and \ref{item-lem-gi-syn} follow
      because $\nlit{N}$ is a conjunct of $\nbranchA{N}$ and $\du{\nlit{N}}$
      is a conjunct of $\nbranchB{N}$.
    \end{itemize}
    \smallskip
  \item Case $\nside{N} = \bbb$:
    \smallskip
    \begin{itemize}
      \item Case $\nside{\ntgt{N}} = \aaa$: Then $\violet{\nipol{N}} =
        \du{\nlit{N}}$.  Properties \ref{item-lem-gi-sem} and
        \ref{item-lem-gi-syn} follow because $\du{\nlit{N}}$ is a conjunct of
        $\nbranchA{N}$ and $\nlit{N}$ is a conjunct of $\nbranchB{N}$.
      \item Case $\nside{\ntgt{N}} = \bbb$: 
        Immediate since then
        $\violet{\nipol{N}} = \true \entails \lnot \nbranchB{N}$.
    \end{itemize}
  \end{itemize}

\noindent
To show the induction step, assume that $N$ is an inner node with children
$N_1, \ldots, N_n$. Consider the case where the side of the children is
$\aaa$.  The induction step for the case where the side of the children is
$\bbb$ can be shown analogously.  By the induction hypothesis we can assume
that for all $i \in \{1,\ldots n\}$ it holds that
\[\red{\FA \land \nbranchA{N_i}} \entails \violet{\nipol{N_i}}
 \entails \blue{\lnotFB \lor \lnot \nbranchB{N_i}},\]
which, since $\nside{N_i} = \aaa$, is equivalent to
\[\red{\FA \land \nbranchA{N} \land \nlit{N_i}} \entails \violet{\nipol{N_i}}
 \entails \blue{\lnotFB  \lor \lnot \nbranchB{N}}.\]
Since $\nipol{N} = \bigvee_{i=1}^{n} \nipol{N_i}$ it follows that
\[\red{\FA \land \nbranchA{N} \land \bigvee_{i = 1}^{n}{\nlit{N_i}}} 
\entails \violet{\nipol{N}}
 \entails \blue{\lnotFB \lor \lnot \nbranchB{N}}.\] 
Because $\bigvee_{i = 1}^{n}{\nlit{N_i}} = \nclause{N}$ is an instance of a
clause in $\FA$ and thus entailed by $\FA$ the semantic requirement
\ref{item-lem-gi-sem} of the induction conclusion follows:
\[\red{\FA \land \nbranchA{N}} \entails \violet{\nipol{N}}
 \entails \blue{\lnotFB \lor \lnot \nbranchB{N}}.\]
The syntactic requirement \ref{item-lem-gi-syn} follows from the induction
assumption and because in general for all nodes~$N$ of a \sided clausal
ground tableau for clausal ground formulas $\FA$ and $\FB$ it holds that all
literals in $\nbranchA{N}$ occur in some clause of $\FA$ and all literals in
$\nbranchB{N}$ occur in some clause of $\FB$.  \qed
\end{proof}

\noindent
Lemma~\ref{lem-ground-ipol-correct} immediately yields a construction method
for Craig-Lyndon interpolants of propositional and, more general, ground
formulas, or, in other words, quantifier-free first-order formulas.  We call
the method \name{\CTIG}, suggesting \name{Clausal Tableau Interpolation}.  In
Sect.~\ref{sec-ipol-lift} below it will be generalized to first-order
sentences in full.
\begin{proc}[The \CTIG Method for Craig-Lyndon Interpolation on
    Ground Formulas]\hspace{-10pt}
  \label{proc-cti-ground}

  \algoinput Ground formulas $F$ and $G$ such that
  $F \entails G$.

  \algomethod Convert $F$ and $\lnot G$ to equivalent clausal ground formulas
  and compute a closed \sided clausal ground tableau for them.  Let $N$ be the
  root of the tableau and compute the value of $\nipol{N}$.
  
  \algooutput Return the value of $\nipol{N}$.  The output is a ground formula
  that is a Craig-Lyndon interpolant of the input formulas.
\end{proc}

\noindent
The procedure is correct: The existence of a closed \sided clausal tableau as
required follows from Proposition~\ref{prop-tab-complete}, that the result is ground
and is a Craig-Lyndon interpolant of $F$ and~$G$ follows from
Lemma~\ref{lem-ground-ipol-correct} and Definition~\ref{def-cli}.

\section{First-Order Interpolant Extraction from Clausal Tableaux}
\label{sec-ipol-lift}

Procedure~\ref{proc-cti-ground} provides a method to compute Craig-Lyndon
interpolants of ground formulas.  We now generalize it to first-order
sentences with arbitrary quantifications.
The starting point is a ground interpolant obtained from a closed clausal
ground tableaux according to Lemma~\ref{lem-ground-ipol-correct}.  The tableau
is now for two clausal formulas that have been obtained from first-order
sentences by Skolemization, conversion to clausal form and instantiation. By a
postprocessing lifting operation, the ground interpolant is converted to an
interpolant of the two original first-order input sentences.  Terms with
function symbols that do not occur in both of them are there replaced by
variables and a suitable quantifier prefix upon these variables is prepended.
The postprocessing is easy to implement, it effects at most a linear increase
of the formula size and its computational effort amounts to sorting the
replaced terms according to their size. Similar lifting techniques have been
shown for resolution-based methods in \cite{huang:95} and \cite[Lemma
  8.2.2]{baaz:11}. We discuss the relationship to these in
Sect.~\ref{sec-lift-related}.

Before we specify the first-order interpolation procedure and prove its
correctness we note that to capture the semantics of Skolemization and to
eliminate function symbols that occur only in one the two interpolation inputs
we use second-order quantification upon functions and predicates in
intermediate formulas, that is, formulas used in the procedure specification
and within the correctness proof.  In particular, we apply the following
properties:
\begin{prop}[Second-Order Skolemization]
\label{prop-second-order-skolemization}
Let $F$ be a formula. Assume that $x_1, \ldots, x_n, y$ are variables that do
not occur bound in $F$ and that $f$ is an $n$-ary function symbol that does
not occur at all in $F$. Then
\[\forall x_1 \ldots \forall x_n \exists y\, F\; \equiv\;
\exists f \forall x_1 \ldots \forall x_n\, F\{y \mapsto f(x_1,\ldots,x_n)\}.\]
\end{prop}

\begin{prop}[Inessential Quantifications in Entailments]
  \label{prop-qu-entailment}
  Let $F,G$ be formulas
and let $\xxs, \yys$ be sets of predicate and function symbols such that $\xxs
\cap (\pred{G} \cup \fun{G}) = \yys \cap (\pred{F} \cup \fun{F}) = \emptyset$.
Then 
\[\exists \xxs\,F \entails \forall \yys\, G\;
\text{ if and only if }\; F \entails G.\]
\end{prop}
Proposition~\ref{prop-qu-entailment} includes the special case of
quantification upon nullary functions, that is, constants, which is actually
\emph{first-order} quantification upon them in the role of variables.  On the
right side of the equivalence stated by the proposition, where they occur
free, they can be viewed as constants or as free variables.  Notice that
$\pred{G}$ and $\pred{F}$ in the preconditions take polarity into account.
That is, if a predicate $p$ occurs in $F$ only with, say, positive polarity
and in $G$ only with negative polarity, then, by
Proposition~\ref{prop-qu-entailment} it holds that $\exists p\, F \entails \forall
p\, G$ holds if and only if $F \entails G$, although $p$ occurs in $F$ as well
as in $G$.

We are now ready to specify the \CTIF method in full, which generalizes
Procedure~\ref{proc-cti-ground} by allowing first-order sentences with
arbitrary quantifications as inputs:
\begin{proc}[The \CTIF Method for Craig-Lyndon Interpolation]
\label{proc-cti-fol}

 \algoinput First-order sentences $F$ and $G$ such that $F \entails G$.

 \algomethod Clausify $F$ and $\lnot G$ to obtain equivalent sentences $\exists
 \ffs_{\ic} \forall \uus_{\ic}\, F_{\ic}$ and $\exists \ggs_{\ic} \forall
 \vvs_{\ic}\, \NG_{\ic}$, respectively, where $\ffs_{\ic}$ and $\ggs_{\ic}$
 are the introduced Skolem functions and $F_{\ic}$ and $\NG_{\ic}$ are clausal
 formulas whose variables are $\uus_{\ic}$ and $\vvs_{\ic}$, respectively.
 Assume w.l.o.g. that $\ffs_{\ic}$ and $\ggs_{\ic}$ are disjoint.  Let $k$ be
 a fresh constant.  Construct a closed \sided clausal ground tableau for
 $F_{\ic}$ and $\NG_{\ic}$ in which all literal labels are instantiated with
 terms formed from $k$ and functions that occur in $F_{\ic}$ or in
 $\NG_{\ic}$. Let~$F_{\ig}$ be the conjunction of the clauses of the tableau
 with side~$\aaa$ and let~$\NG_{\ig}$ be the conjunction of the clauses of the
 tableau with side~$\bbb$.  Let~$H_{\ig}$ be $\nipol{N}$, where $N$ is the
 root of the tableau.  Define:
\[
\begin{array}{r@{\hspace{0.5em}}c@{\hspace{0.5em}}l}
  \ffs & \eqdef & \fun{F_{\ic}} \setminus \fun{\NG_{\ic}}.\\
  \ggs & \eqdef & (\fun{\NG_{\ic}} \setminus \fun{F_{\ic}}) \cup \{k\}.
\end{array}
\]
(Alternatively, it is also possible to place $k$ into $\ffs$ instead
of~$\ggs$. Further possibilities are discussed in
Sect.~\ref{sec-cli-related-choices} below.)
Let $\xxsa$ and $\yysa$ be fresh sequences of variables and let $\stt$ be an
injective substitution with domain $\xxsa \cup \yysa$ such that
\[
\begin{array}{r@{\hspace{0.5em}}c@{\hspace{0.5em}}l}
\rng{\stif} & = & \{t \mid t \text{ is a } \ggs\sterm \text{ occurring in }
H_{\ig} \text{ in a position other than}\\
&& \hphantom{t \mid t} \text{as strict
  subterm of another } \ffs\sterm \text{ or } \ggs\sterm\}.\\
\rng{\stig} & = & \{t \mid t \text{ is an } \ffs\sterm \text{ occurring in }
H_{\ig} \text{ in a position other than}\\
&& \hphantom{t \mid t} \text{as strict
  subterm of another } \ffs\sterm \text{ or } \ggs\sterm\}.
\end{array}
\]
Construct $H_{\iq}$ as
\[H_{\iq}\; \eqdef\; \revsubst{H_{\ig}}{\stt}.\]
Construct the quantifier prefix $Q_1 z_1 \ldots Q_n z_n$ as follows: Let
$\{z_1,\ldots,z_n\}$ be the members of $\xxsa \cup \yysa$ that occur in
$H_{\iq}$ ordered such that for $i,j \in \{1,\ldots,n\}$ it holds that if
$z_i\stt$ is a strict subterm of $z_j\stt$, then $i < j$ and, for $i \in
\{1,\ldots,n\}$, let $Q_i \eqdef \forall$ if $z_i \in \xxsa$ and let $Q_i
\eqdef \exists$ if $z_i \in \yysa$.

\algooutput Return
\[Q_1 z_1
 \ldots Q_n z_n\, H_{\iq}.\] The output is a Craig-Lyndon interpolant of the
 input sentences.
\end{proc}

\medskip

Procedure~\ref{proc-cti-fol} indeed generalizes
Procedure~\ref{proc-cti-ground}: For ground inputs both procedures proceed
identically.  Correctness of the procedure is stated with the following
theorem, which will be proven in detail.  The proof is followed by
Example~\ref{examp-thm-lifting}, which illustrates items mentioned in the
proof for a pair of concrete input sentences.
\begin{thm}[Correctness of the \CTIF Method]
  \label{thm-lifting}
  If $F$ and $G$ are first-order sentences such that $F \entails G$, then
  Procedure~\ref{proc-cti-fol} applied to $F$ and $G$ outputs a Craig-Lyndon
  interpolant of $F$ and $G$.
\end{thm}  
\prlReset{thm-lifting}

\begin{proof}
  Let symbols have the denotation according to the procedure specification.
  In addition we will specify further clausal formulas, sets of variables, and
  substitutions, that relate to the items in the procedure specification and
  are overviewed in the following two graphs:
  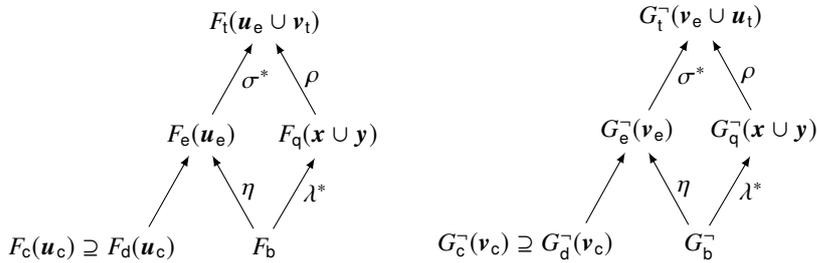
\begin{figure}[h]
    \caption{Clausal formulas
      and substitutions used to prove interpolant lifting.}
    \centering
    \label{fig-lifting}
    $\hphantom{F_{\ic}(\uus_{\ic}) \supseteq}$
    \begin{tikzpicture}[scale=1.0,
        sibling distance=5em,level distance=10ex,
        every node/.style = {transform shape,anchor=mid},
        edge from parent/.style={draw,latex-}]]
        \node {$F_{\ia}(\uus_{\ie} \cup \vvs_{\ia})$}
        child { node {$F_{\ie}(\uus_{\ie})$}
          child { node {$\negphantom{F_{\ic}(\uus_{\ic}) \supseteq}
              F_{\ic}(\uus_{\ic}) \supseteq
              F_{\id}(\uus_{\ic})$} }
          child { node {$F_{\ig}$}
        edge from parent node[right] {$\sth$}}
          edge from parent node[right] {$\inject{\stv}$}
        }
    child { node {$F_{\iq}(\xxsq \cup \yysq)$}
      child { node {$\rule{0pt}{1.95ex}$}
        edge from parent node[right] {$\inject{\stz}$}}
      child[missing] { node {$F$} }
      edge from parent node[right] {$\stren$}
    };
    \end{tikzpicture}
    \hspace{5em}
    \begin{tikzpicture}[scale=1.0,
        sibling distance=5em,level distance=10ex,
        every node/.style = {transform shape,anchor=mid},
        edge from parent/.style={draw,latex-}]]
        \node {$\NG_{\ia}(\vvs_{\ie} \cup \uus_{\ia})$}
        child { node {$\NG_{\ie}(\vvs_{\ie})$}
          child { node {$\negphantom{\NG_{\ic}(\vvs_{\ic}) \supseteq}
              \NG_{\ic}(\vvs_{\ic}) \supseteq
              \NG_{\id}(\vvs_{\ic})$} }
          child { node {$\NG_{\ig}$}
        edge from parent node[right] {$\sth$}}
          edge from parent node[right] {$\inject{\stv}$}
        }
    child { node {$\NG_{\iq}(\xxsq \cup \yysq)$}
      child { node {$\rule{0pt}{1.95ex}$}
        edge from parent node[right] {$\inject{\stz}$}}
      child[missing] { node {$\NG$} }
      edge from parent node[right] {$\stren$}
    };
    \end{tikzpicture}
  \end{figure}

  \noindent
Variables allowed in the respective formulas are shown there in
parentheses. Formulas~$F_{\ig}$ and $\NG_{\ig}$ are ground.  Sets of variables
denoted by different symbols (including differences in the subscript) in the
figure are disjoint.  The superset symbol $\supseteq$ indicates that all
clauses of the formula on the right are clauses of the formula on the left.
Arrows (\raisebox{0.55ex}{\begin{tikzpicture} \draw[-latex] (-0.4,0) --
    (0,0);\end{tikzpicture}}) represent the \name{instance of} relationship,
where the formula at the arrow tip under the substitution shown as arrow label
is the formula at the arrow origin.  Substitutions that are injections are
marked with an asterisk~($*$). The shown substitutions have the following
domains:
\[
\begin{array}{r@{\hspace{0.5em}}c@{\hspace{0.5em}}l}
  \dom{\stv} & = &  \uus_{\ia} \cup \vvs_{\ia}.\\
  \dom{\stren} & = & \uus_{\ie} \cup \vvs_{\ia} \cup \vvs_{\ie} \cup
  \uus_{\ia}.\\
  \dom{\stz} & = & \xxsq \cup \yysq.\\
  \dom{\sth} & = & \uus_{\ie} \cup \vvs_{\ie}.\\
\end{array}  
\]
The following additional syntactic constraints are imposed on the involved
formulas:
\begin{center}
  \begin{tabular}{l}
    Members of $\ggs$ do not occur in $F_{\ia}, F_{\ie}, F_{\ic},
    F_{\id}, \NG_{\ia}, \NG_{\iq}$.\\
    Members of $\ffs$ do not occur in $\NG_{\ia}, \NG_{\ie}, \NG_{\ic},
    \NG_{\id}, F_{\ia}, F_{\iq}$.\\
\end{tabular}
\end{center}

We proceed to show the construction of the involved items, stepping out from
those mentioned in the procedure description.
Sentences~$F$ and $G$ are given as input. The conversion to $\exists
\ffs_{\ic} \forall \uus_{\ic}\, F_{\ic}$ and $\exists \ggs_{\ic} \forall
\vvs_{\ic}\, \NG_{\id}$ can be obtained by usual first-order normal form
transformation.  Skolemization can there be understood as equivalence
preserving rewriting with
Proposition~\ref{prop-second-order-skolemization}. It has to be applied here
independently to $F$ and to $\lnot G$, which is possible since these sentences
do not share quantified variables. The required disjointness conditions on
sets of variables and Skolem functions can be achieved easily by renaming
bound variables.  The sets of functions~$\ffs, \ggs$ can then be constructed
from $F_{\ic}$ and $\NG_{\ic}$.  The following semantic relationships hold:
\[\begin{arrayprf}
\prl{d:1} & F\; \equiv\; \exists \ffs_{\ic} \forall \uus_{\ic}\, F_{\ic}
  \;\entails\; \exists \ffs \forall \uus_{\ic}\, F_{\ic}
  \;\entails\; \forall \ggs \exists \vvs_{\ic}
  \lnot \NG_{\ic}
  \;\entails\; \forall \ggs_{\ic} \exists \vvs_{\ic}\,
  \lnot \NG_{\ic}\; \equiv\; G.
\end{arrayprf}
\]
Given that $\forall \uus_{\ic} \forall \vvs_{\ic}\, (F_{\ic} \land \NG_{\ic})$
is unsatisfiable, which follows from \pref{d:1}, the existence of a closed
\sided ground tableau as specified in the procedure description follows from
``completeness'' of clausal ground tableau construction as implied by
Proposition~\ref{prop-tab-complete} (or, in essence, by Herbrand's theorem).
Formulas $F_{\ig}$, $\NG_{\ig}$ and $H_{\ig}$ as specified then exist, since
they can be extracted from the tableau.
Formulas $F_{\id}$ and $\NG_{\id}$ contain clauses of $F_{\ic}$ and
$\NG_{\ic}$, respectively, such that each clause with side $\aaa$ of the
tableau is an instance of a clause in~$F_{\id}$ and each clause with side
$\bbb$ is an instance of a clause in $\NG_{\id}$.  The following semantic
relationships hold:
\[\begin{arrayprf}
\prl{d:2} & \forall \uus_{\ic}\, F_{\ic}
\;\entails\; \forall \uus_{\ic}\, F_{\id}\; \entails\; F_{\ig}
\;\entails\; \lnot \NG_{\ig}
\;\entails\; \exists \vvs_{\ic}\, \lnot \NG_{\id}
\;\entails\; \exists \vvs_{\ic}\, \lnot \NG_{\ic}.
\end{arrayprf}
\]

We define $\xxsa, \yysa$ and $\stt$, which are specified in the procedure
description, on the basis of larger sets of variables and a substitution with
increased domain that are needed for the further internal proceeding of the
proof: Let $\xxsq$ and $\yysq$ be fresh sequences of variables and let~$\stz$
be an injective substitution with domain $\xxsq \cup \yysq$ such that
\[
\begin{array}{l}
\rng{\stzif} = \{t \mid t \text{ is a } \ggs\sterm \text{ occurring in }
F_{\ig} \text{ or in } \NG_{\ig} \},\\
 \rng{\stzig} = 
\{t \mid t \text{ is an } \ffs\sterm \text{ occurring in }
F_{\ig} \text{ or in } \NG_{\ig} \}.\\
\end{array}
\]
Define~$\xxsa$ as the subset of all members $x$ of $\xxsq$ such that
$x\lambda$ meets the conditions on the range of $\stt$ stated in the procedure
description, and, analogously, define $\yysa$ as the subset of all members~$y$
of $\yysq$ such that $y\lambda$ meets the conditions on the range of $\stt$
stated in the procedure description. Define
\[\stt\;\eqdef\; \lambda|_{\xxsa \cup \yysa}.\]
  The construction of the remaining items specified in the procedure
  description, that is, the formula $H_{\iq}$ and a quantifier prefix $Q_1
  z_1 \ldots Q_n z_n$, is straightforward.

We now consider further items introduced with Fig.~\ref{fig-lifting}.  The
clauses of $F_{\ig}$ are ground instances of clauses of $F_{\id}$. Hence,
there must exist a clausal formula $F_{\ie}$ that subsumes both formulas
$F_{\id}$ and $F_{\ig}$.  Specifically, the formula~$F_{\ie}$ can be
understood as conjunction of ``copies'' of clauses of $F_{\id}$, that is,
clauses of $F_{\id}$ with variables renamed to fresh symbols. The set of
variables $\uus_{\ie}$ consists of all variables occurring in these copies.
Analogous considerations hold for $\NG_{\ie}$. We can then supplement the
semantic relationships in~\pref{d:2} to
\[\begin{arrayprf}
\prl{d:3} & \forall \uus_{\ic}\, F_{\ic}
\entails \forall \uus_{\ic}\, F_{\id}
\equiv \forall \uus_{\ie}\, F_{\ie}
\entails F_{\ig}
\entails \lnot \NG_{\ig}
\entails \exists \vvs_{\ie}\, \lnot \NG_{\ie}
\equiv \exists \vvs_{\ic}\, \lnot \NG_{\id}
\entails \exists \vvs_{\ic}\, \lnot \NG_{\ic}.
\end{arrayprf}
\]
Define $F_{\iq} \eqdef \revsubst{F_{\ig}}{\stz}$ and $\NG_{\iq} \eqdef
\revsubst{\NG_{\ig}}{\stz}$, in analogy to the specification of $H_{\iq}$ in
the procedure description.
The formula $F_{\ia}$ subsumes both $F_{\ie}$
and $F_{\iq}$. Together with the substitution $\stv$ it can be characterized
as follows: Let $\stv$ be an injective substitution such that
\[\rng{\stv|_{\vvs_{\ia}}}\; =\; \{t \mid t \text{ is an } \ffs\sterm
\text{ occurring in } F_{\ie}\},\] and define $F_{\ia} \eqdef
\revsubst{F_{\ie}}{\stv}$. Intuitively, $F_{\ia}$ can be understood as obtained
from $F_{\ie}$ by replacing each term whose principal function symbol does not
occur in $G$ (which includes the Skolem functions $\ffs_{\ic}$) and which is
not a proper subterm of another such term with a dedicated variable from
$\vvs_{\ia}$.  Analogous considerations apply to $\NG_{\ia}$. We complete the
characterization of $\stv$ with
\[\rng{\stv|_{\uus_{\ia}}}\; =\; \{t \mid t \text{ is a } \ggs\sterm
\text{ occurring in } \NG_{\ie}\}\] and define $\NG_{\ia} \eqdef
\revsubst{\NG_{\ie}}{\stv}$.  It still needs to be shown that $F_{\iq}$ is an
instance of $F_{\ia}$ and that $\NG_{\iq}$ is an instance of $\NG_{\ia}$
obtained by applying the substitution $\stren$. Define
\[\stren\; \eqdef\; \{x \mapsto \revsubst{x\stv\sth}{\stz} \mid
x \in \uus_{\ie} \cup \vvs_{\ia} \cup \vvs_{\ie} \cup \uus_{\ia}\}.\]
That $F_{\ia}\stren = F_{\iq}$ can then be shown as follows: Since the range
of $\stz$ only includes $\ffs$-terms and $\ggs$-terms, whereas in $F_{\ia}$
members of $\ffs$ and $\ggs$ do not occur at all it holds that
$F_{\ia}\stv\sth\revsubststandalone{\stz} = F_{\ia}\{x \mapsto
\revsubst{x\stv\sth}{\stz} \mid x \in \uus_{\ie} \cup \vvs_{\ia}\}$.  Since
members of $\vvs_{\ie} \cup \uus_{\ia}$ do not occur in $F_{\ia}$ it follows
that $F_{\ia}\stv\sth\revsubststandalone{\stz} = F_{\ia}\{x \mapsto
\revsubst{x\stv\sth}{\stz} \mid x \in \uus_{\ie} \cup \vvs_{\ia} \cup
\vvs_{\ie} \cup \uus_{\ia}\} = F_{\ia}\stren$.  As Fig.~\ref{fig-lifting}
makes evident, $F_{\ia}\stv\sth = F_{\ig}$.  By definition $F_{\iq} =
F_{\ig}\revsubststandalone{\stz}$.  Hence $F_{\ia}\stren =
F_{\ia}\stv\sth\revsubststandalone{\stz} = F_{\ig}\revsubststandalone{\stz} =
F_{\iq}$. With analogous considerations it follows that $\NG_{\ia}\stren =
\NG_{\iq}$.

We are now done with showing the construction of the items introduced in the
procedure description and in Fig.~\ref{fig-lifting}.  It remains to show on
this basis that the constructed output formula $Q_1 z_1 \ldots Q_n z_n
H_{\iq}$ is indeed a Craig-Lyndon interpolant.  Let $\Qs$ be a shorthand for
$Q_1 z_1 \ldots Q_n z_n H_{\iq}$.  From the construction of $\Qs H_{\iq}$ by
replacing in a ground formula ground terms with variables that are bound by a
prepended quantifier prefix it follows that $\Qs H_{\iq}$ is a sentence.  The
further syntactic properties of a Craig-Lyndon interpolant, as specified with
items (2.)  and (3.) of Definition~\ref{def-cli}, are:
\[\begin{arrayprf}
\prl{b:1} & \pred{\Qs H_{\iq}} \subseteq \pred{F} \cap
\pred{G}.\\
\prl{b:2} & \fun{\Qs H_{\iq}} \subseteq
\fun{F} \cap \fun{G}.
\end{arrayprf}\]
They can be shown as follows: Recall that $H_{\ig}$ is a Craig-Lyndon
interpolant of $F_{\ig}$ and $\lnot \NG_{\ig}$. Hence $\pred{H_{\ig}}
\subseteq \pred{F_{\ig}} \cap \pred{\lnot \NG_{\ig}}$. Statement~\pref{b:1}
then follows since $\pred{\Qs H_{\iq}} = \pred{H_{\iq}} = \pred{H_{\ig}}$,
$\pred{F_{\ig}} \subseteq \pred{F_{\ic}} \subseteq \pred{F}$ and $\pred{\lnot
  \NG_{\ig}} \subseteq \pred{\lnot \NG_{\ic}} \subseteq \pred{G}$.
All members of $\fun{H_{\ig}}$ that are not in $\fun{F} \cap \fun{G}$ are in
$\ffs \cup \ggs$.  Statement~\pref{b:2} then follows since $H_{\iq}$ is defined
as $\revsubst{H_{\ig}}{\stt}$, which implies $\fun{\Qs H_{\iq}} =
\fun{H_{\iq}} \subseteq \fun{H_{\ig}}$ and, with the specification of $\stt$,
that there are no occurrences of members of $\ffs \cup \ggs$ in $H_{\iq}$.

It remains to prove that $\Qs H_{\iq}$ has the semantic characteristics of a
Craig-Lyndon interpolant as specified with item~(1.)  of
Definition~\ref{def-cli}. Let $\{w_1,\ldots,w_m\}$ be $\xxsq \cup \yysq$ ordered
such that for $i,j \in \{1,\ldots,m\}$ it holds that if $w_i\stz$ is a strict
subterm of $w_j\stz$ then $i < j$ and the ordering of $\{z_1,\ldots,z_n\}$ is
extended, that is, if $w_a = z_c$, $w_b = z_d$ and $c < d$, then $a < b$. For
$i \in \{1,\ldots,m\}$ let $R_i \eqdef \forall$ if $w_i \in \xxsq$ and let
$R_i \eqdef \exists$ if $w_i \in \yysq$. Let $\Rs \eqdef R_1 w_1 \ldots R_m
w_m$.  Since $F_{\ig} \entails H_{\ig} \entails \lnot \NG_{\ig}$, $F_{\iq} =
\revsubst{F_{\ig}}{\stz}$, $\NG_{\iq} = \revsubst{\NG_{\ig}}{\stz}$, and
$H_{\iq} = \revsubst{\NG_{\ig}}{\stt} = \revsubst{\NG_{\ig}}{\stz}$ it follows
that $F_{\iq} \entails H_{\iq} \entails \lnot \NG_{\iq}$.  Hence $\Rs F_{\iq}
\entails \Rs H_{\iq} \entails \Rs \lnot \NG_{\iq}$.  Since the quantifier
prefix $\Qs$ consists of exactly those quantifications in $\Rs$ that are upon
variables occurring in $H_{\iq}$, in the same order as in $\Rs$, it holds that
$\Rs H_{\iq} \equiv \Qs H_{\iq}$. Thus
\[\begin{arrayprf}
\prl{b:5} & \Rs F_{\iq} \entails \Qs H_{\iq} \entails \Rs \lnot \NG_{\iq}.
\end{arrayprf}\]
The semantic property of a Craig-Lyndon interpolant that we are going to prove
is $F \entails \Qs H_{\iq} \entails G$.  Given~\pref{b:5}, this follows from
$F \entails \Rs F_{\iq}$ and $\Rs \lnot \NG_{\iq} \entails G$. We will now
show the first of these entailments, $F \entails \Rs F_{\iq}$.  For this we
need a further substitution, $\stsk$, which is not displayed in
Fig.~\ref{fig-lifting}.  Its key properties are stated as~\pref{d:6}
and~\pref{d:7} below. They will later serve to justify the base cases of an
induction. Their proof depends on a further property of $\stren$:
\[\begin{arrayprf}
\prl{d:4} & \stv\stren = \stren.
\end{arrayprf}\]
Equality~\pref{d:4} can be shown as follows: For members~$x$ of $\dom{\stren}
\setminus \dom{\stv} = \uus_{\ie} \cup \vvs_{\ie}$ it is evident that
$x\stv\stren = x\stren$. It remains to consider further members of $\dom{\stv}
\cup \dom{\stren} = \dom{\stren} = x \in \uus_{\ie} \cup \vvs_{\ia} \cup
\vvs_{\ie} \cup \uus_{\ia}$, that is, members of $\vvs_{\ia} \cup \uus_{\ia}$.
Let $x$ be a member of this set.  Observe that $\var{x\stv} \subseteq
\uus_{\ie} \cup \vvs_{\ie}$.  Since $\dom{\stv} \cap (\uus_{\ie} \cup
\vvs_{\ie}) = \emptyset$ it holds that $x\stv\stv = x\stv$. Hence, by the
definition of $\stren$ it holds that $x\stv\stren =
\revsubst{x\stv\stv\sth}{\stz} = \revsubst{x\stv\sth}{\stz} = x\stren$, which
concludes the proof of~\pref{d:4}.
We now define the substitution $\stsk$ with $\dom{\stsk} = \xxsq \cup \yysq$
by
\[\stsk\; \eqdef\; \{x \mapsto \revsubst{x\stz}{\stz |_{\yysq}} \mid x \in \xxsq\}
\cup \{y \mapsto \revsubst{y\stz}{\stz |_{\xxsq}} \mid y \in \yysq\}.\]
It then holds that
\[\begin{arrayprf}
\prl{d:5} & \forall \uus_{\ie}\, F_{\ie}\; =\;
\forall \uus_{\ie}\, F_{\ia}\stv \;\entails\;
\forall \xxsq\, F_{\ia}\stv\stren\stskf\; =\;
\forall \xxsq\, F_{\ia}\stren\stskf\; =\;
\forall \xxsq\, F_{\iq}\stskf.
\end{arrayprf}
\]
That $\forall \uus_{\ie}\, F_{\ia}\stv \entails \forall \xxsq\,
F_{\ia}\stv\stren\stskf$ follows since
$F_{\ia}\stv\stren\stskf$ is an instance of $F_{\ia}\stv$ and the free
variables in both formulas are universally quantified on both sides of the
entailment.  That $\forall \xxsq\, F_{\ia}\stv\stren\stskf = \forall
\xxsq\, F_{\ia}\stren\stskf$ follows from~\pref{d:4}.  The remaining
identities in~\pref{d:5} are immediate from the relationships displayed in
Fig.~\ref{fig-lifting}.  From \pref{d:5}, \pref{d:3} and \pref{d:1}
it follows that
\[\begin{arrayprf}
\prl{d:6} & 
F \;\entails\; \exists \ffs \forall \xxsq\, F_{\iq}\stskf.
\end{arrayprf}
  \]
Analogously, it can be shown that
\[\begin{arrayprf}
\prl{d:7} &
\forall \ggs \exists \yysq\, \lnot \NG_{\iq}\stskg \;\entails\; G.
\end{arrayprf}
  \]

In preparation of an inductive argument, we define fragments of $\Rs$, $\xs$
and $\stsk$ for all $i \in \{0,\ldots,m\}$:
\[
\begin{array}{r@{\hspace{0.5em}}c@{\hspace{0.5em}}l}
\Rs_i & \eqdef & R_1 w_1 \ldots R_i w_i.\\
\xxsq_i & \eqdef & \xxsq \cap \{w_{i+1},\ldots,w_m\}.\\
\stski_i & \eqdef & \stsk|_{\yysq \cap \{w_{i+1},\ldots,w_m\}}.
\end{array}\]
Observe that then
\[\begin{arrayprf}
\prl{b:9} & \Rs_0 = \emptyseq.\;\; \xxsq_0 = \xxsq.\;\; \stski_0 = \stsk|_{\yysq}.\\
\prl{b:12} & \Rs_m = \Rs.\;\; \xxsq_m  = \{\}.\;\; \stski_m = \emptysubst.
\end{arrayprf}\]
We now show by induction that for all $i \in \{0,\ldots,m\}$ it holds
that
\[\begin{arrayprf}
(\text{IP}) &
F \entails \exists \ffs \Rs_i \forall \xxsq_i\, F_{\iq}\stski_i.
\end{arrayprf}
\]
If $i=m$, then, by~\pref{b:12} the entailment~(IP) equals $F \entails \exists
\ffs \Rs F_{\iq}$. Since $\fun{F_{\iq}} \cap \ffs = \emptyset$ this
is equivalent to $F \entails \Rs F_{\iq}$, the statement to prove.
In the base case $i=0$ of the induction, the entailment~(IP) is by~\pref{b:9}
identical with $F \entails \exists \ffs \forall \xxsq\, F_{\iq}\stskf$, which
we have already shown as~\pref{d:6}.  To show the induction step we assume as
induction hypothesis that (IP) holds for some $i \in \{0,\ldots,m-1\}$.
The variable $w_{i+1}$ must be either in $\xxsq$ or in $\yysq$.  In the case
$w_{i+1} \in \xxsq$ it holds that $\xxsq_i = \{w_{i+1}\} \cup \xxsq_{i+1}$ and
that $\stski_{i+1} = \stski_i$. Thus
\[\begin{arrayprf}
\prl{b:15} & \exists \ffs \Rs_i \forall \xxsq_i\, F_{\iq}\stski_i
\;\equiv\; \exists \ffs \Rs_i \forall w_{i+1} \forall \xxsq_{i+1}\,
  F_{\iq}\stski_{i+1}
\;\equiv\; \exists \ffs \Rs_{i+1} \forall \xxsq_{i+1}\,
  F_{\iq}\stski_{i+1}.
\end{arrayprf}
\]
In the case $w_{i+1} \in \yysq$ it holds that $\xxsq_{i+1} = \xxsq_{i}$ and
$\stski_i = \{w_{i+1} \mapsto w_{i+1}\stsk\}\stski_{i+1}$.
Moreover, it holds that
\[
\begin{arrayprf}
\prl{b:22} & w_{i+1} \notin \var{\rng{\stski_{i+1}}}.\\ 
\prl{b:23} & \xxsq_{i+1} \cap \var{w_{i+1}\stsk} = \emptyset.\\
\end{arrayprf}
\]
Statement~\pref{b:22} follows since $w_{i+1} \in \yysq$ and
$\var{\rng{\stski_{i+1}}} \subseteq \var{\rng{\stsk|_{\yysq}}} \subseteq
\xxsq$. Statement~\pref{b:23} can be shown as follows: Assume that~\pref{b:23}
does not hold.  Then there must be a number $j$ such that $w_j \in \xxsq_{i+1}
\cap \var{w_{i+1}\stsk}$.  By the definition of $w_i$ it follows that $j >
i+1$. It also follows that $w_j\stz$ is a strict subterm of
$w_{i+1}\stsk\stz$, and thus, since $\stsk\stz = \stz$, which is not hard to
verify, also of $w_{i+1}\stz$. From the specification of the ordering of
$\{w_1, \ldots, w_m\}$ it follows that $j < i+1$, which contradicts with the
condition $j > i+1$ just derived. Hence~\pref{b:23} must hold.  We now can
state the following relationships, where the entailment step is justified
by~\pref{b:22} and~\pref{b:23}:
\[
\begin{arrayprfeq}
\prl{b:18} &&  \Rs_i \forall \xxsq_i\, F_{\iq}\stski_i\\
& \equiv & \Rs_i \forall \xxsq_{i+1}\, F_{\iq}\{w_{i+1} \mapsto w_{i+1}\stsk\}\stski_{i+1}\\
& \entails & \Rs_i \exists w_{i+1} \forall \xxsq_{i+1}\, F_{\iq}\stski_{i+1}\\
& \equiv & \Rs_{i+1} \forall \xxsq_{i+1}\, F_{\iq}\stski_{i+1}.\\
\end{arrayprfeq}    
\]
Given the induction hypothesis $F \entails \exists \ffs \Rs_i \forall
\xxsq_i\, F_{\iq}\stski_i$, the induction conclusion \[F \entails \exists \ffs
\Rs_{i+1} \forall \xxsq_{i+1}\, F_{\iq}\stski_{i+1}\] follows in the case
$w_{i+1} \in \xxsq$ from~\pref{b:15} and in the case $w_{i+1} \in \yysq$
from~\pref{b:18}. Hence we have established
\[
\begin{arrayprf}
  \prl{e:1} & F \entails \Rs F_{\iq}.
\end{arrayprf}
\]
Analogously it can be shown that
\[
\begin{arrayprf}
  \prl{e:2} &  \Rs \lnot \NG_{\iq} \entails G.
\end{arrayprf}
\]
Combining \pref{b:1}, \pref{b:2}, \pref{b:5}, \pref{e:1} and \pref{e:2} and
recalling that $\Qs$ was defined as shorthand for $Q_1 z_1 \ldots Q_n z_n$ we
can finish the proof of Theorem~\ref{thm-lifting} by concluding that the
output of the \CTIF procedure is indeed a Craig-Lyndon interpolant:
\[
\begin{arrayprf}
  \prl{e:3} &
Q_1 z_1 \ldots Q_n
z_n\, H_{\iq} \text{ is a Craig-Lyndon interpolant of } F \text{ and } G.
\end{arrayprf}
\qed
\]
\end{proof}

The following example illustrates the proof of Theorem~\ref{thm-lifting}:
\begin{examp}[Items in the Proof of Theorem~\ref{thm-lifting}]
  \label{examp-thm-lifting}
  Consider computation of a Craig-Lyndon interpolant by the \CTIF method for
  the sentences:
  \[
  \begin{array}{r@{\hspace{0.5em}}c@{\hspace{0.5em}}l}
    F & = & \forall x_1 \forall x_2\, \pk(x_1, \hk(\ffk_1(x_1)), x_2)).\\
    G & = & \exists x_1 \exists x_2\,  (\pk(\hk(\ggk_2(x_1)), x_2, \ggk_1) \land \pk(\ggk_1, x_1, \hk(\ggk_2(x_1)))).
  \end{array}
  \]
  The common symbols of both sentences are the predicate $\pk$, in positive
  polarity, and the function $\hk$.  Alternatively, the non-common functions
  $\ffk_1, \ggk_1, \ggk_2$ might be viewed as Skolem functions for original sentences
  \[
  \begin{array}{r@{\hspace{0.5em}}c@{\hspace{0.5em}}l}
    F' & = & \forall x_1 \exists y \forall x_2\, \pk(x_1, \hk(y), x_2).\\
    G' & = & \forall
  x_1 \exists y_1 \forall x_2 \exists y_2\, (\pk(\hk(x_2), y_2, x_1) \land \pk(x_1,
  y_1, \hk(x_2))).
  \end{array}
  \]
  Under this view, however, $F'$ and $G'$ themselves both qualify as
  Craig-Lyndon interpolants of $F'$ and $G'$. Nevertheless, the proceeding in
  the example can also be understood as computing a further interpolant of
  $F'$ and $G'$ which actually is strictly weaker than $F'$ and strictly
  stronger than $G'$.

  We return back to the original view of $\ffk_1, \ggk_1, \ggk_2$ as functions occurring
  in just one of the interpolation inputs $F$ and $G$ and show the respective
  values of the items mentioned the description of
  Procedure~\ref{proc-cti-fol} and in the proof of its correctness,
  Theorem~\ref{thm-lifting}.  Converting $F$ and $\lnot G$ to clausal form
  yields the following formulas, variables and sets of distinguished functions:
  \[
  \begin{array}{r@{\hspace{0.5em}}c@{\hspace{0.5em}}l}
    F_{\ic} & = & \pk(u^{\ic}_1, \hk(\ffk_1(u^{\ic}_1)), u^{\ic}_2)).\\
    \NG_{\ic} & = &
    \lnot \pk(\hk(\ggk_2(v^{\ic}_1)), v^{\ic}_2, \ggk_1) \lor \lnot \pk(\ggk_1, v^{\ic}_1,
    \hk(\ggk_2(v^{\ic}_1))).\\
    \uus_{\ic} & = & \{u^{\ic}_1, u^{\ic}_2\}.\\
    \vvs_{\ic} & = & \{v^{\ic}_1, v^{\ic}_2\}.\\
    \ffs & = & \{\ffk_1\}.\\
    \ggs & = & \{\ggk_1, \ggk_2, k\}.\\
  \end{array}
  \]
  Formulas $F_{\ig}$ and $\NG_{\ig}$ are clausal ground formulas
  obtained from instantiating clauses of $F_{\ic}$ and $\NG_{\ic}$.  Actually
  it is easy to verify syntactically that $F_{\ig} \equiv \lnot \NG_{\ig}$,
  hence $F_{\ig} \land \NG_{\ig}$ is unsatisfiable, implying that a \sided
  ground tableau for $F_{\ig}$ and $\NG_{\ig}$ can be constructed.  From that
  tableau we can extract $H_{\ig}$, a Craig-Lyndon interpolant of $F_{\ig}$
  and $\lnot \NG_{\ig}$. Since $F_{\ig}$, $\NG_{\ig}$ and $H_{\ig}$ are built
  up from the same two ground atoms, we introduce shorthands $A, B$ for these
  to facilitate readability:
      \[
      \begin{array}{r@{\hspace{0.5em}}c@{\hspace{0.5em}}l}
        A & = & \pk(\hk(\ggk_2(\hk(\ffk_1(\ggk_1)))), \hk(\ffk_1(\hk(\ggk_2(\hk(\ffk_1(\ggk_1)))))), \ggk_1).\\
        B & = & \pk(\ggk_1, \hk(\ffk_1(\ggk_1)), \hk(\ggk_2(\hk(\ffk_1(\ggk_1))))).\\
    F_{\ig}\;=\; H_{\ig}\; & = & A \land B.\\
    \NG_{\ig} & = & \lnot A \lor \lnot B.
  \end{array}
  \]
  Formulas $F_{\iq}, \NG_{\iq}, H_{\iq}$ can be viewed as obtained from
  $F_{\ig}, \NG_{\ig}, H_{\ig}$ by replacing $\ffs$-terms and $\ggs$-terms in
  occurrences that are not as subterm of another such term with dedicated
  variables, that is, different terms are replaced by different variables and
  identical terms with the same variable:
  \[
  \begin{array}{r@{\hspace{0.5em}}c@{\hspace{0.5em}}l}
    F_{\iq} \;=\; H_{\iq} & = & \pk(\hk(x_1), \hk(y_1), x_2)\; \land\\
           &&    \pk(x_2, \hk(y_2), \hk(x_1)).\\
    \NG_{\iq} & = & \lnot \pk(\hk(x_1), \hk(y_1), x_2) \lor
    \lnot \pk(x_2, \hk(y_2), \hk(x_1)).\\
    \xxsq & = & \{x_1, x_2\}.\\
    \yysq & = & \{y_1, y_2\}.\\
    \stz & = &
    \{x_1 \mapsto \ggk_2(\hk(\ffk_1(\ggk_1))),\,
      x_2 \mapsto \ggk_1,\\
      && \hparenc
      y_1 \mapsto \ffk_1(\hk(\ggk_2(\hk(\ffk_1(\ggk_1))))),\,
      y_2 \mapsto \ffk_1(\ggk_1)\}.
  \end{array}
  \]
  The variables $w_1,\ldots w_m$ are determined by $\stz$ as follows,
  where $m=4$:  
  \[w_1 = x_2,\, w_2 = y_2,\, w_3 = x_1,\, w_4 = y_1\]
  Hence, the interpolant is 
  \[\Qs\ H_{\iq}\; =\; \forall x_2 \exists y_2 \forall x_1 \exists y_1
  (\pk(\hk(x_1), \hk(y_1), x_2) \land \pk(x_2, \hk(y_2), \hk(x_1)).\]

  \noindent
  We proceed with the example by showing intermediate formulas used
  ``internally'' within the proof.  In $F_{\ig}$ we needed with $A$ and $B$
  two different instantiations of the single clause in $F_{\ic}$. Hence
  formula $F_{\ie}$ provides that clause in two ``copies'':
  \[
  \begin{array}{r@{\hspace{0.5em}}c@{\hspace{0.5em}}l}
    F_{\ie} & = & \pk(u^{\ie}_1, \hk(\ffk_1(u^{\ie}_1)), u^{\ie}_2))\; \land\\
          &&    \pk(u^{\ie}_3, \hk(\ffk_1(u^{\ie}_3)), u^{\ie}_4)).\\
    \NG_{\ie} & = &
    \lnot \pk(\hk(\ggk_2(v^{\ie}_1)), v^{\ie}_2, \ggk_1) \lor \lnot \pk(\ggk_1, v^{\ie}_1,
    \hk(\ggk_2(v^{\ie}_1))).\\
    \uus_{\ie} & = & \{u^{\ie}_1, u^{\ie}_2, u^{\ie}_3, u^{\ie}_4\}.\\
    \vvs_{\ie} & = & \{v^{\ie}_1, v^{\ie}_2\}.\\
    \sth & = & \{u^{\ie}_1 \mapsto \hk(\ggk_2(\hk(\ffk_1(\ggk_1)))),\,
    u^{\ie}_2 \mapsto \ggk_1,\,
    u^{\ie}_3 \mapsto \ggk_1,\,
    u^{\ie}_4 \mapsto \hk(\ggk_2(\hk(\ffk_1(\ggk_1)))),\\
    && \hparenc v^{\ie}_1 \mapsto \hk(\ffk_1(\ggk_1)),\,
    v^{\ie}_2 \mapsto \hk(\ffk_1(\hk(\ggk_2(\hk(\ffk_1(\ggk_1))))))\}.
  \end{array}
  \]
  A clausal formula $F_{\ia}$ that subsumes both $F_{\ie}$ and $F_{\iq}$ and
  an analogous $\NG_{\ia}$ that subsumes both $\NG_{\ie}$ and $\NG_{\iq}$
  along with the respective substitutions can be specified as:
  \[
  \begin{array}{r@{\hspace{0.5em}}c@{\hspace{0.5em}}l}
    F_{\ia} & = & \pk(u^{\ie}_1, \hk(v^{\ia}_1), u^{\ie}_2)\; \land\\
          &&    \pk(u^{\ie}_3, \hk(v^{\ia}_2), u^{\ie}_4).\\
    \NG_{\ie} & = &
    \lnot \pk(\hk(u^{\ia}_1), v^{\ie}_2, u^{\ia}_2) \lor
    \lnot \pk(u^{\ia}_2, v^{\ie}_1, \hk(u^{\ia}_1)).\\
    \vvs_{\ia} & = & \{v^{\ia}_1, v^{\ia}_2\}.\\
    \uus_{\ia} & = & \{u^{\ia}_1, u^{\ia}_2\}.\\
    \stv & = & \{v^{\ia}_1 \mapsto \ffk_1(u^{\ie}_1),\,
    v^{\ia}_2 \mapsto \ffk_1(u^{\ie}_3),\,
    u^{\ia}_1 \mapsto \ggk_2(v^{\ie}_1),\,
    u^{\ia}_2 \mapsto \ggk_1\}.\\
    \stren & = &
    \{u^{\ie}_1 \mapsto \hk(x_1),\,
      v^{\ia}_1 \mapsto y_1,\,
      u^{\ie}_2 \mapsto x_2,\,
      u^{\ie}_3 \mapsto x_2,\,
      v^{\ia}_2 \mapsto y_2,\,
      u^{\ie}_4 \mapsto \hk(x_1),\\
    && \hparenc u^{\ia}_1 \mapsto x_1,\,
      v^{\ie}_2 \mapsto \hk(y_1),\,
      u^{\ia}_2 \mapsto x_2,\,
      v^{\ie}_1 \mapsto \hk(y_2)\}.
  \end{array}
  \]

  \noindent
  The proof involves an induction where it is shown that for all $i \in
  \{0,\ldots,m\}$ it holds that
  \[F \entails \exists \ffs \Rs_i \forall
  \xxsq_i\, F_{\iq}\stski_i.\] The base case $i = 0$ is equal to: $F
  \;\entails\; \exists \ffs \forall \xxsq\, F_{\iq}\stskf$. The case $i=m$ is
  equal to $F \;\entails\; \Rs\, F_{\iq}$ and is used in the proof to justify
  the semantic property of the lifted interpolant. In our case $m$ is $4$.
  The substitution~$\stsk$ is determined by $\stz$ as follows:
  \[
  \begin{array}{r@{\hspace{0.5em}}c@{\hspace{0.5em}}l}
    \stsk & = &
    \{x_1 \mapsto  \ggk_2(\hk(y_2)),\,
    x_2 \mapsto  \ggk_1,\,
    y_1 \mapsto \ffk_1(\hk(x_1)),\,
    y_2 \mapsto \ffk_1(x_1)\}
  \end{array}
  \]
  Recall that in our example the \emph{ordered} set $\{w_1, w_2, w_3, w_4\}$
  is $\{x_2, y_2, x_1, y_1\}$. The substitutions~$\stsk_i$ used in the
  induction property are then:
  \[
  \begin{array}{r@{\hspace{0.5em}}c@{\hspace{0.5em}}l}
    \stsk|_{\yysq} = \stsk_0 = \stsk_1
    & = & \{y_2 \mapsto \ffk_1(x_1),\, y_1 \mapsto \ffk_1(\hk(x_1))\}.\\
    \stsk_2 = \stsk_3 & = & \{y_1 \mapsto \ffk_1(\hk(x_1))\}.\\
    \stsk_4 & = & \emptysubst.
  \end{array}
  \]
  We finish the example with showing the induction property for $i \in
  \{0,\ldots,m=4\}$, where changes in the matrix compared to the previous step
  are highlighted by underlining:
  \[
  \begin{array}{r@{\hspace{0.5em}}llll}
  i && \Rs_i & \hphantom{\forall}\xs_i &  F_{\iq}\stski_i \\\midrule
  0 & F \entails \exists \ffs && \forall x_2  x_1
  & (\pk(\hk(x_1), \hk(\ffk_1(\hk(x_1))), x_2) \land \pk(x_2, \hk(\ffk_1(x_1)), \hk(x_1))).\\
  1 & F \entails \exists  \ffs & \forall x_2 & \forall x_1\,
  & (\pk(\hk(x_1), \hk(\ffk_1(\hk(x_1))), x_2) \land \pk(x_2, \hk(\ffk_1(x_1)), \hk(x_1))).\\
  2 & F \entails \exists  \ffs & \forall x_2 \exists y_2 & \forall x_1\,
  & (\pk(\hk(x_1), \hk(\ffk_1(\hk(x_1))), x_2) \land \pk(x_2, \hk(\underline{y_2}), \hk(x_1))).\\
  3 & F \entails \exists  \ffs & \forall x_2 \exists y_2 \forall x_1 &\,
  & (\pk(\hk(x_1), \hk(\ffk_1(\hk(x_1))), x_2)  \land \pk(x_2, \hk(y_2), \hk(x_1))).\\
  4 & F \entails & \forall x_2 \exists y_2 \forall x_1 \exists y_1 &\,
  & (\pk(\hk(x_1), \hk(\underline{y_1}), x_2) \land \pk(x_2, \hk(y_2), \hk(x_1))).\\
  \end{array}
  \]
\end{examp}  

\medskip

We conclude this section with a proposition that shows some properties of
Craig-Lyndon interpolants constructed with the \CTIF procedure that go beyond
the requirements of a Craig-Lyndon interpolant (Definition~\ref{def-cli}), are
useful in certain applications, such as Theorem~\ref{thm-ipol-horn} below
and easily follow from the specification of the \CTIF procedure:
\begin{prop}[Properties of Interpolants Constructed with \CTIF]
  \label{prop-ipol-cip-properties}
  Let $\Qs F$ and $\Rs G$ be first-order sentences such that $\Qs$ and $\Rs$
  are quantifier prefixes, $F$ and $G$ are quantifier-free formulas and it
  holds that $\Qs F \entails \Rs G$. Then, by the \CTIF method a first-order
  sentence $\Ss H$ can be constructed such that $\Ss$ is a quantifier prefix,
  $H$ is a quantifier-free formula, $\Ss H$ is a Craig-Lyndon interpolant of
  $\Qs F$ and $\Rs G$, and it holds that:

  \smallskip
  \begin{enumerate}
    \item If there is an existential quantification in $\Ss$, then there is an
      existential quantification in $\Qs$ or there is a member of $\fun{F}$
      that is not in $\fun{G}$.

  \item If there is a universal quantification in $\Ss$, then there is a
    universal quantification in $\Rs$ or there is a member of $\fun{G}$
    that is not in $\fun{F}$.


  \item If $F_{\ic}$ and $\NG_{\ic}$ are clausal formulas obtained from
    clausifying $\Qs F$ and $\lnot \Rs G$, respectively, and $N$ is the root
    of a closed \sided ground tableau for $F_{\ic}$ and $\NG_{\ic}$, then
    $H\sigma = \nipol{N}$ for some substitution $\sigma$ whose domain is the
    set of the variables quantified in $\Ss$.
  \end{enumerate}
\end{prop}  
\begin{proof}
  Follows from the specification of Procedure~\ref{proc-cti-fol}
  and its correctness, Theorem~\ref{thm-lifting}. \qed
\end{proof}  
The first two items of Proposition~\ref{prop-ipol-cip-properties} concern
quantifiers in the interpolant in a coarse way, just with respect to their
kind, existential or universal, without taking dependencies on their order
into account. The third item states in essence that whenever for first-order
inputs there is a ground interpolant of the respective clausifications whose
formula has a certain structure, then there is a first-order interpolant of
the original inputs whose matrix has the same structure.

\section{Interpolant Lifting: Related Work}
\label{sec-lift-related}

The interpolant lifting of Procedure~\ref{proc-cti-fol} by replacing terms in
a ground interpolant $H_{\ig}$ with fresh variables $z_1,\ldots,z_n$ and
prepending a quantifier prefix $Q_1 z_1 \ldots Q_n z_n$ whose ordering
respects the subterm relationship among the replaced terms has been already
shown in essence by Huang \cite{huang:95}.  Although this interpolant lifting
can be expressed as a simple formula conversion, independently of any
particular calculus, its correctness seems not trivial to prove and subtle
issues arise.  For example, as observed in \cite{kovacs:17}, there is an error
in \cite{huang:95} that concerns equality handling.  Another example is a
version of interpolant lifting developed in \cite{bonacina:15:on:ipol} where
only constants are replaced by variables but which, as indicated in
\cite{bonacina:15:on:ipol}, does not generalize to compound terms in a way
that is compatible with other techniques shown there.  It seems that so far
two proofs for interpolant lifting with respect to compound terms can be found
in the literature: The proof of \cite[Theorem~15]{huang:95} and the proof of
\cite[Lemma~8.2.2]{baaz:11}, seemingly obtained independently.  Interpolant
lifting is called \name{abstraction} in \cite{baaz:11}.  Further discussions
and references can be found in
\cite{bonacina:15:on:ipol,kovacs:17,benedikt:2017}. Our use of lifting and our
correctness proof differs from the related methods and proofs described in the
literature \cite{huang:95,baaz:11,bonacina:15:on:ipol,kovacs:17} in two
important respects:
\begin{enumerate}
\item We apply lifting to a ground formula that actually \emph{is a
  Craig-Lyndon interpolant} of two intermediate ground formulas that relate in
  a certain way to the input sentences.  In contrast, lifting is applied in
  \cite{huang:95} to a so-called \name{relational interpolant} of the original
  input sentences, which is specified like a Craig interpolant, except that
  constraints on the functions need not to be satisfied.  Similarly less
  constrained variants of a Craig interpolant of the original input sentences
  are used as basis for lifting in \cite{baaz:11} (\name{weak interpolant})
  and in \cite{bonacina:15:on:ipol} (\name{provisional interpolant}).

\item Our proof of the correctness of interpolant lifting is \emph{independent
  of a particular calculus}.  The correctness proofs in \cite{huang:95} and
  \cite{baaz:11} are both based on modifying proofs as data structures,
  resolution proofs in the case of \cite{huang:95} and natural deduction
  proofs in the case of \cite{baaz:11}.  We assume more abstractly just
  Herbrand's theorem, expressed in the form that for an unsatisfiable clausal
  first-order formula a \emph{closed ground tableau can be constructed}, where
  terms are formed from input functions, Skolem functions and, if needed, an
  additional constant.  The tableau enters our method as ``given'', where the
  actual way in which it had been constructed is irrelevant.  Provers would
  typically operate on non-ground clauses and hand over a closed non-ground
  tableau which is instantiated to a ground tableau only just before
  extraction of the ground interpolant.  For practical implementation, this
  approach has the advantage that any system which computes a clausal tableau
  for an unsatisfiable first-order formula can be applied unaltered to the
  computation of first-order interpolants.

\end{enumerate}

We now look into the details of some interesting aspects of Huang's result in
\cite{huang:95} in comparison to ours.  There are similarities in the involved
formulas or resolution derivations, respectively, used internally in both
proofs: Huang's proof uses a conversion of the given resolution deduction to
what he calls \name{binary tree} deduction, where each clause is used at most
once. Analogously, in our formulas~$F_{\ie}$ and~$\NG_{\ie}$ of the proof of
Theorem~\ref{thm-lifting} each variable is instantiated to a ground term by
the substitution $\sth$. In Huang's proof, the binary tree deduction is
converted further to what he calls a \name{propositional} deduction, which
correspond to our ground formulas $F_{\ig} = F_{\ie}\sth$ and $\NG_{\ig} =
\NG_{\ie}\sth$.

In \cite{huang:95} equality handling with paramodulation is explicitly taken
into account, which, however, leads to the mentioned error in Huang's lifting
theorem \cite{kovacs:17}.
The proof of \cite[Lemma~8.2.2]{baaz:11} applies just to formulas without
equality.  A possibility to integrate equality handling into our method is
described in Sect.~\ref{sec-cli-related-equality}.

A minor difference between our and Huang's lifting technique is that Huang
orders variables in the quantifier prefix by the \emph{length} of the
associated terms, more constrained than the strict subterm relationship used
here.

In contrast to Huang's method for constructing relational interpolants, the
method of \cite{baaz:11} to construct weak interpolants involves certain cases
where quantified variables are introduced.  In Huang's relational interpolants
free variables are allowed, upon which extra quantifiers will be added after
lifting.  As indicated in \cite[p.~188]{huang:95}, this can be done in an
arbitrary way: the extra quantifiers can be existential or universal, at any
position in the prefix.  In our formalization, the base formulas used for
lifting have to be ground.  The effects described by Huang are subsumed by the
alternate possibilities to instantiate non-ground tableaux delivered by
provers as discussed in Sect.~\ref{sec-cli-related-choices}.

The input formulas in Huang's interpolation method are clausal formulas.  In
the symbolism of the proof of our Theorem~\ref{thm-lifting}, his method
computes interpolants of $\forall \uus_{\ic}\, F_{\ic}$ and $\exists
\vvs_{\ic}\, \lnot \NG_{\ic}$.  The handling of arbitrary first-order formulas
by Skolemization incorporated in our proof needs to be wrapped around Huang's
core theorem, which is, however not difficult: Staying in the symbolism of the
proof of our theorem, the set $\ffs \cup \ggs$ includes the involved Skolem
functions.  An interpolant of $\forall \uus_{\ic}\, F_{\ic}$ and $\exists
\vvs_{\ic}\, \lnot \NG_{\ic}$ in which -- after lifting -- no members of $\ffs
\cup \ggs$ occur is also an interpolant of $F$ and $G$, the original formulas
before Skolemization.

\section{Positive Hyper Tableaux and Interpolation from a Horn Sentence}
\label{sec-tab-constraints}

So far, our interpolant construction based on clausal tableaux applies to
arbitrarily structured closed clausal ground tableaux.  To obtain interpolants
that, in dependence of syntactic properties of the input formulas, have
specific syntactic properties beyond those required from Craig-Lyndon
interpolants, it is useful to consider clausal tableaux with structural
restrictions. Two basic restrictions are specified with the
following definition:
\begin{defn}[Tableau Properties: Regular, Leaf-Only]
  \label{def-tab-properties}
Define the following properties of clausal tableaux:

\sdlab{def-regular} \defname{Regular}: No node has an ancestor
with the same literal label.

\sdlab{def-leaf-only} \defname{Leaf-only} for a set $S$ of pairwise
non-complementary literals: Members of $S$ do not occur as literal labels of
inner nodes.
\end{defn}
\name{Regularity} is a well-known standard notion to avoid redundancies in
tableaux, see, e.g., \cite{handbook:tableaux:letz,handbook:ar:haehnle}.  Any
closed clausal tableaux for some clausal formula can be converted with a
tableau simplification to a regular closed clausal tableau for the same
formula (\cite{letz:habil}, see also Sect.~\ref{sec-access-convert}).
The \name{leaf-only} property can be applied to model constraints on clausal
tableaux that are constructed by ``bottom-up'' methods, as shown with
Definition~\ref{def-phyper} below.  In Sect.~\ref{sec-access-extract} we will
apply it together with a further tableau restriction to essentially simulate
non-clausal tableaux with clausal tableaux.  Any closed clausal tableau can be
transformed to a closed clausal tableau for the same formula that is leaf-only
for a given set of pairwise non-complementary literals, although the required
transformations are potentially expensive (see
Sect.~\ref{sec-access-convert}).

In the introduction we mentioned the important family of methods that can be
understood as constructing a clausal tableaux ``bottom-up'', in a
``forward-chaining'' manner, by starting with positive axioms and deriving
positive consequences, with the hyper tableaux calculus \cite{hypertab} as a
representative.  The following definition, expressed in terms of properties
from Definition~\ref{def-tab-properties}, renders structural constraints that
are typically observed by tableaux constructed with these methods:
\begin{defn}[Positive Hyper Tableau]
  \label{def-phyper}
  A clausal tableau that is regular and leaf-only for the set of all negative
  literals occurring as labels the tableau is called a \defname{positive hyper
    tableau}.
\end{defn}  
In a \emph{closed} positive hyper tableau the leafs are exactly the nodes with
negative literal label.  The term \name{positive hyper tableau} is from
\cite{handbook:ar:haehnle}, where methods that construct such tableaux are
investigated as specializations in a general setting of clausal tableau
methods with selection functions.  Availability of complete methods ensures
that for any unsatisfiable clausal formula a closed positive hyper tableau can
be constructed.  These construction methods typically observe further
constraints that are not modeled in Definition~\ref{def-phyper} since they are
not relevant for extracting interpolants from a given closed tableau. This
includes in particular that variable scopes are only clause-local and that
nodes labeled with a negative literal are always closed (``weakly
connected''), also during construction when the overall tableau is not yet
closed.

To demonstrate how the clausal tableau framework for first-order Craig-Lyndon
interpolation can be applied to derive further properties of constructed
interpolants in dependency of properties of the input formulas we now show
that, if the first interpolation argument is a Horn sentence, then for
arbitrary sentences as second arguments an interpolant that is a Horn sentence
can be constructed.  We first specify some syntactically characterized formula
classes:
\begin{defn}[Formula Classes: Universal, Existential, Positive, Negative, Horn]

  \sdlab{def-univ-ex} A sentence is called \defname{universal}
  (\defname{existential}) if it is a first-order sentence of the form $\Qs F$,
  where $\Qs$ is an individual quantifier prefix with only universal
  (existential) quantifications and $F$ is quantifier-free.

  \sdlab{def-pos-neg} A formula is called \defname{positive}
  (\defname{negative}) if and only if all occurrences of atoms in the formula
  have positive (negative) polarity.

  \sdlab{def-horn} A sentence is called \defname{Horn} if and only if it is a
  first-order sentence of the form $\Qs F$ where~$\Qs$ is a quantifier prefix
  and $F$ is a quantifier-free clausal formula with at most one positive
  literal in each of its clauses.
\end{defn}  

\noindent
Now the claimed property of interpolants where the first argument is Horn
can be made precise as follows:
\begin{thm}[Interpolation from a Horn Sentence]
  \label{thm-ipol-horn}
  Let $F, G$ be first-order sentences such that $F$ is Horn and $F\entails G.$
  Then a Craig-Lyndon interpolant~$H$ of $F$ and $G$ can be constructed such
  that
\begin{enumerate}
\item $H$ is a Horn sentence.
\item If $F$ and $G$ are universal sentences and $\fun{F} \subseteq \fun{G}$,
  then $H$ is a universal sentence.
\item If $F$ and $G$ are existential sentences and $\fun{G} \subseteq
  \fun{F}$, then $H$ is an existential sentence.
\end{enumerate}  
\end{thm}

\begin{proof}
  Let \name{is essentially a Horn ground formula} stand for \name{is a Horn
    ground formula or can be converted to an equivalent Horn ground formula by
    distributing disjunction upon conjunction}.  Existence of a closed
  two-sided positive hyper ground tableau for any clausification results of
  $F$ and $\lnot G$ follows from the completeness of proving methods that
  construct positive hyper tableaux.  Since $F$ is Horn, it can be clausified
  such that the respective clausal formula is Horn.  The theorem then follows
  from Proposition~\ref{prop-ipol-cip-properties} since, as we will show,
  if~$N$ is the root of a closed two-sided positive hyper ground tableau for
  clausal formulas~$F_\aaa$ and~$F_\bbb$ where~$F_\aaa$ is Horn, then the
  formula $\nipol{N}$ is a essentially a Horn ground formula.  We prove the
  latter claim by showing with induction on the tableau structure the
  following more general statement:
  \[\begin{arrayprf}
(\text{IP}) &
\text{For all nodes } N \text{ of a closed two-sided clausal positive hyper
  ground tableau}\\
& \text{for clausal formulas } F_\aaa \text{ and } F_\bbb
  \text{ where } F_\aaa \text{ is Horn
    the formula } \nipol{N}
  \text{ is}\\
& \text{essentially a Horn ground formula.}
\end{arrayprf}
\]
  For the base case where $N$ is a leaf it is immediate from
  Definition~\ref{def-ipol} that $\nipol{N}$ is a ground literal or a truth
  value constant and thus obviously a Horn ground formula.  To show the
  induction step, let $N$ be an inner node with children $N_1,\ldots,N_n$
  where $n \geq 1$.  As induction hypothesis assume that for all $i \in
  \{1,\ldots,n\}$ it holds that $\nipol{N_i}$ is essentially a Horn ground
  formula. We prove the induction step by showing that then also $\nipol{N}$
  is essentially a Horn ground formula:
  \begin{itemize}
    \item Case $\nside{N_1} = \aaa$: Observe that since the tableau is
      leaf-only for all negative literals it holds in this case for all $i \in
      \{1,\ldots,n\}$ such that $\nlit{N_i}$ is negative that either
      $\nipol{N} = \false$ or $\nipol{N}$ is a negative ground literal.
  \begin{itemize}
    \item Case $\nclause{N}$ is negative: Then $\nipol{N} = \bigvee_{i=1}^n
      \nipol{N_i}$ is a disjunction of negative ground literals, hence a Horn
      ground formula.
    \item Case $\nclause{N}$ is not negative: Since $F_\aaa$ is Horn,
      $\nclause{N}$ has exactly one child whose literal label is positive. Let
      $N_j$ with $j \in \{1,\ldots,n\}$ be that child.  By the induction
      hypothesis $\nipol{N_j}$ is essentially a Horn ground formula.  Since,
      as observed above, for all $i \in \{1,\ldots,n\} \setminus \{j\}$ the
      formula $\nipol{N_i}$ is $\false$ or a negative ground literal it
      follows that $\nipol{N} = \bigvee_{i=1}^n \nipol{N_i}$ is essentially a
      Horn ground formula.
  \end{itemize}
  
  \item Case $\nside{N_1} = \bbb$:  From the induction
  hypothesis it follows that $\nipol{N} = \bigwedge_{i=1}^n \nipol{N_i}$
  is essentially a Horn ground formula.
  \qed
  \end{itemize}
\end{proof}  

As we have seen with Theorem~\ref{thm-ipol-horn}, the framework for
interpolation based on clausal tableaux allows to prove the existence of
interpolants that meet certain constraints quite easily and, moreover, in a
constructive way that can be realized directly by practical automated
reasoning systems. An apparently weaker property has been shown in
\cite[\S~4]{mcnulty:uhorn} with techniques from model theory: For \emph{two}
universal Horn formulas there exists a universal Horn formula that is like a
Craig interpolant, except that function symbols occurring in it are not
constrained.

\vspace{-5pt} 
\section{Craig-Lyndon Interpolation: Refinements and Issues}
\label{sec-cli-related}

\subsection{Choices in Grounding and Side Assignment}
\label{sec-cli-related-choices} 

Procedure~\ref{proc-cti-fol} for the construction of first-order interpolants
leaves at several stages alternate choices that have effect on the formula
returned as interpolant.
We discuss some of these here, although a thorough investigation of ways to
integrate the exploration and evaluation of these into interpolant
construction seems a nontrivial topic on its own.

The first choice concerns the instantiation of variables in the tableaux
returned by provers. Typically, provers instantiate variables just as much
``as needed'' by the calculus to compute a closed tableau.  To match with our
interpolant lifting technique, variables in the literal labels of such
non-ground tableaux have to be instantiated by ground terms. There are
different possibilities to do so, all yielding a closed tableau for the input
clauses, but leading to different interpolants: A variable can be instantiated
by a term whose functions all occur in both interpolation inputs.  The term
will then occur in the interpolant.  Alternatively, the variable can be
instantiated by a term with a principal functor that has been introduced at
Skolemization, either of the first or of the second input formula, or that
occurs just in one the input formulas.  In the procedure description the fresh
constant~$k$ that is handled like a Skolem constant in the second input
formula has been introduced to have such a term available in any case.  By
interpolant lifting the term will then be replaced by a variable whose kind,
existential or universal, depends on the principal functor of the term, and
whose quantifier position in the prefix is constrained by subterms.  Of
course, also a combination of these two ways is possible, that is,
instantiating with a term whose principal functor occurs in both input
formulas but which has subterms with a functor that does not occur in one of
the input formulas.

Aside of these alternate possibilities that concern the instantiation of each
variable individually, there are also choices to instantiate different
variables by the same term or by different terms: Arbitrary subsets of the
free variables of the literal labels of the tableau can be instantiated with
the same ground term, leading in the interpolant to fewer quantified variables
but to more variable sharing.  In the description of
Procedure~\ref{proc-cti-fol} the fresh constant~$k$ is uniformly used to
instantiate all variables.

The second possibility for choice concerns the assignment of the
\name{side}~$\aaa$ and~$\bbb$ to tableau clauses.  Existing systems for
tableau construction typically would require changes to their internal data
structures to maintain such side information, which is undesirable. However,
assuming that sides are associated with the given input clauses, such systems
can be actually used unaltered to construct a \sided clausal tableau: Sides
can then be assigned to the clauses of the returned tableau ``in retrospect'',
by matching against the input clauses.  In some cases there are choices: A
tableau clause can be an instance of some input clause with side~$\aaa$ as
well as of some input clause with side~$\bbb$, or a clause can be present in
copies for each side.  In these cases it is possible to assign either side to
the tableau clause, where both assignments may lead to different interpolants.

\subsection{Goal-Sensitivity} 

Model elimination and the connection method are goal sensitive: They construct
a clausal tableau by starting with a clause from a designated subset of the
input clauses, the ``goal clauses''. Without loss of completeness the set of
negative clauses can, for example, be taken as goal clauses, or, if a theorem
is to be proven from a consistent set of axioms, the clauses representing the
(negated) theorem. It remains to be investigated what choices of goal clauses
are particularly useful for the computation of interpolants.

\subsection{Equality Handling}
\label{sec-cli-related-equality}

So far we considered only first-order logic \emph{without equality}.
Nevertheless, our method to compute interpolants can be used together with the
well-known encoding of equality as a binary predicate with axioms that express
its reflexivity, symmetry and transitivity as well as axioms that express
substitutivity of predicates and functions.
If the input formulas of interpolant computation involve equality, these
axioms have to be added.
The clauses of substitutivity axioms for predicate or function symbols that
occur only in the first (second) input formula then receive side~$\aaa$
($\bbb$).  The side of clauses of axioms for reflexivity, symmetry and
transitivity can be assigned arbitrarily, including the possibility to have
two copies of the clauses, one for each side.

For relational formulas, more can be said about the polarity in which equality
may occur in the interpolant in cases where it occurs only in the first
(second) of input formula: Then the clauses of axioms for reflexivity,
symmetry and transitivity can be assigned to the corresponding side to ensure
that in the computed interpolant equality only occurs positively
(negatively). This follows from the ``Lyndon property'', the condition that
predicates occur in the interpolant only in polarities in which they occur in
both input formulas, since in substitutivity clauses for predicates, which are
then the only clauses with equality literals whose side is $\bbb$ ($\aaa$),
equality only occurs negatively. Stronger possible constraints on interpolants
with respect to equality are stated in an interpolation theorem due to
Oberschelp and Fujiwara (see \cite{motohashi:84}).

The example from \cite{kovacs:17} to demonstrate the mentioned error in
\cite{huang:95} in presence of equality is finding an interpolant of $\rk(\ak)
\neq \rk(\bk)$ and $\ak \neq \bk$: Huang's proof would imply $\exists x\, x
\neq \rk(\bk) \entails \ak \neq \bk$, which does not hold in general. With our
suggested encoding we would obtain $\ak \neq \bk$ as ground interpolant of the
ground formulas $\rk(\ak) \neq \rk(\bk) \land (\ak \neq \bk \lor \rk(\ak) =
\rk(\bk))$ and $\ak \neq \bk$, and, because lifting has no effect, also
correctly as interpolant of the original inputs $\rk(\ak) \neq \rk(\bk)$ and
$\ak \neq \bk$.

\vspace{-6pt} 
\subsection{Preprocessing for Interpolation}
\label{sec-preproc}

Sophisticated preprocessing is a crucial component of automated reasoning
systems with high performance. While formula simplifications such as removal
of subsumed clauses and removal of tautological clauses preserve equivalence,
others only preserve unsatisfiability.  For example, \name{purity
  simplification}, that is, removal of clauses that contain a literal with a
predicate that occurs only in a single polarity in the formula.  Many
simplifications of the latter kind actually preserve not just
unsatisfiability, but, moreover, equivalence \emph{with respect to a set of
  predicates}, or, more precisely, a second-order equivalence
\begin{equation}
\label{eq-so-simp}
\exists p_1 \ldots \exists p_n\, F\; \equiv\;
\exists p_1 \ldots \exists p_n\, \f{simplify}(F),
\end{equation}
where $\f{simplify}(F)$ stands for the result of the simplification operation
applied to~$F$. One might say that the \emph{semantics of the predicates not
  in $\{p_1,\ldots,p_n\}$ is preserved} by the simplification. For the
computation of Craig-Lyndon interpolants it is possible to preprocess the
first as well as the negated second input formula independently from each
other in ways such that the semantics of the predicates occurring in both
formulas is preserved in this sense.  Preprocessors that support
simplification operations that can be parameterized with a set of predicates
whose semantics has to be preserved (see \cite[Section~2.5]{cw-pie} for a
discussion) can be applied for that purpose.

For clausal tableau methods some of these simplifications are particularly
relevant as they complement tableau construction with techniques which break
apart and join clauses and may thus introduce some of the benefits of
resolution.
Techniques for propositional logic that preserve
equivalence~(\ref{eq-so-simp}) for certain sets of predicates include variable
elimination by resolution and blocked clause elimination.  For first-order
generalizations of these, the handling of equality seems the most difficult
issue.  Predicate elimination can in general introduce equality also for
inputs without equality.  In a semantic framework where the Herbrand universe
is taken as domain this can be avoided to some degree, as shown in
\cite{cw-skp} with a variant of the SCAN algorithm \cite{scan} for predicate
elimination.  Blocked clause elimination in first-order logic
\cite{blocked:fol} comes in two variants, for formulas without and with
equality, respectively.

\vspace{-6pt} 
\subsection{Issues Related to Definer Predicates} 
Another way to utilize equivalence~(\ref{eq-so-simp}) is by \emph{introducing}
fresh ``definer'' predicates for example by structure-preserving normal forms
such as the Tseitin transformation and first-order generalizations of it
\cite{scott:twovars,tseitin,eder:def:85,plaisted:greenbaum}.  Actually, in our
approach to compute access interpolants this will play an important role.  If
disjoint sets of definer predicates are used for the first and for the second
interpolation input, then, by the definition of \name{Craig-Lyndon
  interpolant}, definer predicates do not occur in the interpolant.  In
certain situations, which need further investigation, it might be useful to
relax this constraint. For example, if in preprocessing two definers whose
associated subformulas are equal should be identified, even if one was
introduced for the first and the other for the second interpolation input.
Another example would be allowing definers occurring in the interpolant in
cases where this permits a condensed representation of a formula whose
equivalent without the definers would be much larger but straightforward to
obtain.

\vspace{-6pt} 
\subsection{Implementation -- Current State} 
\label{sec-cli-implem}
The \name{PIE} system \cite{cw-pie} includes an implementation of the
described approach to Craig-Lyndon interpolation.  Currently the
goal-sensitive first-order prover \name{CM} included with \name{PIE} is
supported as underlying theorem prover. Support for using also \name{Hyper}
\cite{cw-ekrhyper,cw-krhyper,hyper:2013} in that role has been
implemented in part.  The
clausal tableaux used for interpolant extraction are represented as Prolog
terms, providing a potential interface also to further provers.
Configurable preprocessing which respects preservation of predicate semantics
as required for interpolation is included.  \name{Symmetric} interpolation
\cite[Lemma~2]{craig:uses} (the name is due to \cite{mcmillan:symmetric}) with
consideration of predicate polarity is implemented as iterated interpolation
with two inputs.  Other implementations of interpolation will be discussed in
Sect.~\ref{sec-implem-qr} in the context of query reformulation.

\vspace{-6pt} 
\section{Access Interpolation with Clausal Tableaux: Overview and Basic Notions}
\label{sec-ai}

Access interpolation \cite{benedikt:book} is a recently introduced form of
interpolation for applications in query reformulation where the two input
formulas as well as the computed interpolants are in a fragment of first-order
logic, \name{first-order logic with relativized quantifiers}.  This fragment
allows to associate a \defname{binding pattern}, or \name{access pattern},
with each atom occurrence: a representation of its polarity, of its predicate,
and of a division of argument positions into \name{input} and \name{output
  positions}.  The technical framework for access interpolation has been
developed in \cite{benedikt:book} on the basis of Smullyan's non-clausal
tableaux \cite{smullyan:book:68,fitting:book}, which follow the formula
structure, proceeding from the overall input into subformulas, which allows to
integrate relativized quantifiers whose scope is a subformula in an elegant
way.  This correspondence to the formula structure is as such not available in
clausal tableaux, obtained after clausification, Skolemization and with
techniques targeted at automated processing that follow inner connection
structures instead of the formula structure.
The basic approach adopted here is to ``simulate'' aspects of Smullyan's
tableaux by clausal tableaux as much as needed for the extraction of
interpolants that are in the target fragment with relativized quantifiers.
This is achieved by first converting the input formulas into a structure
preserving normal form. Then there are two ways to proceed, which we will both
consider: The first is to compute a closed clausal tableau that is constrained
in a particular way such that an access interpolant can be extracted. The
second is to compute an arbitrary closed clausal tableau and convert it such
that it meets the constraints required to extract an access interpolant.
With the following Definitions~\ref{def-rqfo-main}--\ref{def-ai} we
recapitulate precise notions underlying access interpolation, adapted from
\cite{benedikt:book}:
\begin{defn}[\RQFO Formula]
  \label{def-rqfo-main}
  \
  
\sdlab{def-rqfo} The formulas of \name{first-order logic with relativized
  quantifiers}, briefly \name{\RQFO formulas}, are the relational formulas
that are generated by the grammar
\[\begin{array}{rcl}
F & ::= & \true \mid \false \mid F \land F \mid F \lor F \mid \forall \vs\, (\lnot R \lor F) \mid \exists \vs\, (R \land F),
\end{array}
\]
where in the last two grammar rules $\vs$ matches a (possibly empty) set of
variables and $R$ matches a relational atom in which all members of $\vs$
occur.

\smallskip

\sdlab{def-rqfo-neg} If $F$ is an \RQFO formula, then $\lnot F$ denotes the
\RQFO formula obtained from rewriting~$\lnot F$ exhaustively with equivalences
that propagate negation inwards, that is: $\lnot \true \equiv \false$; $\lnot
\false \equiv \true$; $\lnot (F \land G) \equiv \lnot F \lor \lnot G$; $\lnot
(F \lor G) \equiv \lnot F \land \lnot G$; $\lnot \forall \vs (\lnot A \lor F)
\equiv \exists \vs (A \land \lnot F)$; $\lnot \exists \vs (A \land F) \equiv
\forall \vs (\lnot A \lor \lnot F)$.
\end{defn}

\newcommand{\POLA}{\pi}

\begin{defn}[Binding Pattern and Related Notions]
\label{def-binding-pattern-main}

\sdlab{def-bpatt} A \defname{\bpatt} is a triple $\la \mathit{Sign},
\mathit{Predicate}, \mathit{InputPositions}\ra$, where $\mathit{Sign} \in
\{+,-\}$, $\mathit{Predicate}$ is a predicate and $\mathit{InputPositions}$ is
a set of numbers larger than or equal to~$1$ and smaller than or equal to the
arity of $\mathit{Predicate}$.  A \bpatt with sign~$+$ ($-$) is called
\defname{existential} (\defname{universal}).

\smallskip

\sdlab{def-covered} A \bpatt $\la S, P, O\ra$ is \defname{covered} by a
\bpatt $\la S^\prime, P^\prime, O^\prime\ra$ if and only if $S = S^\prime$, $P = P^\prime$ and $O^\prime
\subseteq O$. A set $B$ of \bpatts is \defname{covered} by a set of
\bpatts $B^\prime$ if and only if each member of $B$ is covered by some member of
$B^\prime$.

\smallskip

\sdlab{def-f-bpatts} The \defname{\bpatts}~$\bpp{F}$ of an \RQFO formula
$F$ 
is a set of \bpatts defined inductively as follows:
\vspace{-10pt}
\[
\begin{array}{r@{\hspace{0.5em}}c@{\hspace{0.5em}}l}
\bpp{\true}\; \eqdef\; \bpp{\false} & \eqdef & \{\}.\\
\bpp{G \land H}\; \eqdef\; \bpp{G \lor H} & \eqdef & \bpp{G} \cup \bpp{H}.\\
\bpp{\forall \vs\, (\lnot r(t_1,\ldots,t_n) \lor G)} & \eqdef
 & \{\la {-}, r, \{i \mid t_i \notin \vs \} \ra \} \cup \bpp{G}.\\
\bpp{\exists \vs\, (r(t_1,\ldots,t_n) \land G)} & \eqdef
 & \{\la {+}, r, \{i \mid t_i \notin \vs\}\} \ra \} \cup \bpp{G}.\\
\end{array}
\]
\end{defn}

\smallskip

\noindent
For example, if $F =\forall x\, (\lnot \rk(x) \lor \exists y \exists z\,
(\sk(x,y,z) \land \true))$, then $\bpp{F} = \{\la {-}, \rk, \{\}\ra,\; \la {+},
\sk, \{1\} \ra\}$.

\begin{defn}[Access Interpolant]
\label{def-ai}
Let $F, G$ be \RQFO sentences such that $F \entails G$. An \defname{access
  interpolant of $F$ and $G$} is an \RQFO sentence $H$ such that
\begin{enumerate}
\item \label{cond-ai-sem} $F \entails H \entails G$.
\item \label{cond-ai-lyndon} $\pred{H} \subseteq \pred{F} \cap \pred{G}$.
\item \label{cond-ai-relat} Every existential \bpatt of $H$ is covered by an
  existential \bpatt of $G$.  Every universal \bpatt of $H$ is covered by a
  universal \bpatt of $F$.\footnote{Compared to
    \cite[Thm.~3.12]{benedikt:book} in this definition of \name{access
      interpolant} from the condition~(\ref{cond-ai-relat}.) the explicit
    requirements that the predicate of an existential (universal) \bpatt of
    $H$ occurs positively in $F$ (negatively in $G$) have been dropped because
    these are already implied by condition~(\ref{cond-ai-lyndon}.).}
\item \label{cond-ai-const} $\const{H} \subseteq \const{F} \cap \const{G}$.
\end{enumerate}
\end{defn}

Our approach to compute access interpolants with clausal tableau resides on a
\name{structure preserving}, also called \name{definitional}, normal form
\cite{scott:twovars,tseitin,eder:def:85,plaisted:greenbaum} for clausifying
the two input \RQFO formulas.  Auxiliary ``definer'' predicates for
subformulas are introduced there. By using disjoint sets of definer predicates
for the conversion of each of the two input formulas it is ensured that
definer predicates do not occur in interpolants.  The normalization yields
only clauses of certain specific forms.
To specify the subformula definers we use the following common notions of
\name{subformula position} and \name{subformula at a position}, specialized to
\RQFO formulas by considering the relativizer literals not as subformulas on
their own but as belonging to the associated quantifications:
\begin{defn}[Position within an \RQFO Formula]
\ 

\sdlab{def-position} A \defname{position} of a subformula occurrence within an
\RQFO formula is a finite sequence of integers.

\smallskip

\sdlab{def-poss} The \defname{positions}~$\poss{F}$ of an \RQFO formula is a
set of positions defined inductively as follows: If $F$ is $\true$ or
$\false$, then $\poss{T} \eqdef \{\emptyseq\}$, if $F$ is of the form $F_1
\land F_2$ or $F_1 \lor F_2$, then $\poss{F} \eqdef \{\emptyseq\} \cup \{1p
\mid p \in \poss{F_1}\} \cup \{2p \mid p \in \poss{F_2}\}$ and if $F$ is of
the form $\forall \vs\, (\lnot R \lor F_1)$ or $\exists \vs\, (R \land F_1)$,
then $\poss{F} \eqdef \{\emptyseq\} \cup \{1p \mid p \in \poss{F_1}\}$.

\smallskip

\sdlab{def-subform-at} The \defname{subformula at} position~$p$ in an \RQFO
formula~$F$, in symbols $\subform{F}{p}$ is defined inductively as
$\subform{F}{\emptyseq} \eqdef F$; $\subform{F_1 \binop F_2}{iq} \eqdef
\subform{F_i}{q}$ for $\binop \in \{\land, \lor\}$ and $i \in \{1,2\}$;
$\subform{\forall \vs\, (\lnot R \lor F_1)}{1q} \eqdef \subform{\exists \vs\,
  (R \land F_1)}{1q} \eqdef \subform{F_1}{q}$.
\end{defn}
We assume a total order on the set of all variables, called the
\defname{standard order of variables}.  The following definition specifies
structure preserving conversions of \RQFO formulas that yield conjunctions of
first-order formulas of certain shapes.

\begin{defn}[Definitional Form of an \RQFO Formula]
  \label{def-defform} Let $F$ be an \RQFO formula.
  
\sdlab{def-sfd} For all positions $p \in \poss{F}$ let $\dx_p$ denote the
sequence of the members of $\var{\subform{F}{p}}$ ordered according to the
standard order of variables and let $D_p$ denote the atom $\f{d}_p(\dx_p)$,
where $\f{d}_p$ is a fresh predicate, also called a \defname{definer
  predicate}.  For all positions $p \in \poss{F}$ define the sentence
$\defp{p}$ depending on the form of $\subform{F}{p}$ as shown in the
following table:
\[
\begin{array}{l@{\hspace{1em}}l}
\subform{F}{p} & \defp{p}\\\midrule
\true & D_p \imp \true\\
\false  & D_p \imp \false\\
G \land H & \forall \dx_p\, (D_p \imp (D_{p1}\land D_{p2} ))\\
G \lor H & \forall \dx_p (D_p \imp (D_{p1} \lor D_{p2}))\\
\forall \xs\, (\lnot R \lor G) & \forall \dx_p\, (D_p \imp \forall \xs\, (\lnot R \lor D_{p1}))\\
\exists \xs\, (R \land G) & \forall \dx_p\, (D_p \imp \exists \xs\, (R \land D_{p1}))
\end{array}
\]

\smallskip

\sdlab{def-sfd-pos-neg} Define the following formula:
\[
\begin{array}{r@{\hspace{0.5em}}c@{\hspace{0.5em}}l}
  \DEFP{F} & \eqdef & \ddp_{\emptyseq}(\xs_{\emptyseq}) \land \bigwedge_{p
    \in \poss{F}} \defp{p}.\\
\end{array}
\]
\end{defn}

\medskip

\noindent
Structural normal forms that are like Definition~\ref{def-sfd} based on of
implications instead of equivalences are known as \name{Plaisted-Greenbaum
 form} \cite{plaisted:greenbaum}.  The semantic relationship between a
formula and its definitional form as specified in Definition~\ref{def-defform} is
captured by a second-order equivalence, which is easy to verify
with Ackermann's lemma \cite{ackermann:35,dls}:
\begin{prop}[Semantic Properties of the Definitional Form of an \RQFO Formula]
\label{prop-sem-def-nf}
Let~$F$ be an \RQFO formula, let $\{p_0, \ldots, p_n\} = \poss{F}$ and let
$\ddp_{p_1}, \ldots, \ddp_{p_n}$ be definer predicates as specified in
Definition~\ref{def-defform}. Then
\[\begin{array}{r@{\hspace{0.5em}}c@{\hspace{0.5em}}l}
F & \equiv & \exists
\ddp_{p_1} \ldots \exists \ddp_{p_n}\, \DEFP{F}.
\end{array}
\]
\end{prop}
Proposition~\ref{prop-sem-def-nf} allows to express the semantic
requirement~(\ref{cond-ai-sem}.) of the definition of \name{access
  interpolant} (Definition~\ref{def-ai}) in terms of the normalized formulas:
\begin{prop}[Semantic Property of Interpolants for
\RQFO Formulas in Definitional Form]
\label{prop-fhg-normalized}
Let $F, G$ be \RQFO formulas. Let $\{p_i \ldots p_m\} = \poss{F}$, let
$\{q_i \ldots q_n\} = \poss{G}$, and let predicates $\ddp_{p_1}, \ldots
\ddp_{p_m}$ and $\eeq_{q_1}, \ldots, \eeq_{q_n}$ be the definer predicates
introduced with forming $\DEFP{F}$ and $\DEFPNOT{G}$, respectively,
according to Definition~\ref{def-defform}.  Let $H$ be a formula such that $\pred{H}
\cap \{\ddp_{p_1}, \ldots \ddp_{p_m}, \eeq_{q_1}, \ldots, \eeq_{q_n}\} =
\emptyset$.  Then the following statements are equivalent:
\begin{enumerate}
\item $F \entails H \entails G$.
\item $\exists \ddp_{p_1} \ldots \exists \ddp_{p_m}\, \DEFP{F}\; 
\entails\; H\;
\entails\; \lnot \exists \eeq_{q_1} \ldots \exists \eeq_{q_n}\, \DEFPNOT{G}$.
\item $\DEFP{F}\; \entails\; H\; \entails\; \lnot \DEFPNOT{G}$.
\end{enumerate}
\end{prop}
As basis for computing an access interpolant we thus can take a closed \sided
clausal tableau for a clausal form of $\DEFP{F}$ as $F_\aaa$ and a clausal
form of $\DEFPNOT{G}$ as $F_\bbb$. The following lemma specifies the clause
forms obtained and introduces symbolic notation to refer to particular
literals, variables and Skolem functions occurring in them:

\begin{lem}[Definitional Clausification of an RQFO Formula]
\label{lem-input-clauseforms}
Let $F$ be an \RQFO formula.  For all $p \in \poss{F}$ let $\ddp_p$ denote the
definer predicate for $p$ introduced at forming $\DEFP{F}$, let $\dx_p$ denote
the sequence of the members of $\var{\subform{F}{p}}$ ordered according to the
standard order of variables, and let $D_p$ denote the atom
$\f{d}_p(\dx_p)$. For all $p \in \poss{F}$ where $\subform{F}{p}$ is of the
form $\forall x_1 \ldots \forall x_n\, (\lnot R \lor F^\prime)$ or $\exists
x_1 \ldots \exists x_n\, (R \land F^\prime)$ let~$R_p$ denote $R$ and let
$\vout_p$ denote $\{x_1, \ldots, x_n\}$.  For all $p \in \poss{F}$ where
$\subform{F}{p}$ is of the form $\exists x_1 \ldots \exists x_n\, (R \land
F^\prime)$ let $\skf_{\la p, 1\ra}, \ldots, \skf_{\la p, n\ra}$ be fresh
functions and let $\stcr_p$ be the substitution $\{x_1 \mapsto \skf_{\la p,
  1\ra}(\dx_p), \ldots, x_n \mapsto \skf_{\la p, n\ra}(\dx_p) \}$. Then
$\DEFP{F}$ is equivalent to the existential quantification upon Skolem
functions of the universal closure of a clausal formula, where the Skolem
functions are the introduced $\skf_{\la p, i\ra}$ and the clauses are of the
following forms, satisfying restrictions on arguments of atoms and free
variables as indicated:
\[\begin{array}
{r@{\hspace{1em}}l@{\hspace{1em}}l}
\text{No.} & \text{Clause Form} & \text{Restrictions}\\\midrule
1 & D_{\emptyseq} & \text{If } F
\text{ is a sentence, then } \parg{D_{\emptyseq}} = \emptyset\\
2 & \lnot D_p\\
3 & \lnot D_p \lor D_{p1} & \parg{D_{p1}} \subseteq \parg{\lnot D_p}\\
4 & \lnot D_p \lor D_{p2} & \parg{D_{p2}} \subseteq \parg{\lnot D_p}\\
5 & \lnot D_p \lor D_{p1} \lor D_{p2} 
  &  \parg{D_{p1} \lor D_{p2}} \subseteq \parg{\lnot D_p}\\
6 & \lnot D_p \lor \lnot R_p \lor D_{p1}
   &  \parg{D_{p1}} \subseteq \parg{\lnot D_p \lor \lnot R_p},\\
  &&  \parg{\lnot R_p \lor D_{p1}} \subseteq \parg{\lnot D_p} \cup \vout_p\\
7 & \lnot D_p \lor R_p\stcr_p
&  \var{R_p\stcr_p} \subseteq \parg{\lnot D_p}\\
8 & \lnot D_p \lor D_{p1}\stcr_p
  & \var{D_{p1}\stcr_p} \subseteq \parg{\lnot D_p}\\
\end{array}
\]
\end{lem}
\begin{proof}
The required clausal form would be obtained by common CNF transformation
methods, provided Skolemization is applied individually to each implication of
the form $\forall \dx_p\, (D_p \imp \exists \xs\, (R \land D_{p1}))$.  \qed
\end{proof}
The order within blocks 1--8 of Lemma~\ref{lem-input-clauseforms} corresponds
to the order in which clauses would be obtained by a straightforward CNF
translator applied on the definitional implications in the order displayed in
Definition~\ref{def-sfd}.

The applied variant of Skolemization is \name{inner} Skolemization
\cite{nonnengart:weidenbach:handbook}.  This follows because the universal
quantifications upon $\dx_p$ that precedes the quantification upon the
Skolemized variables $\vs$ is exactly upon the free variables of the argument
formula of the quantification upon $\vs$, that is, $R \land
\ddp_{p1}(\dx_{p1})$. Considering that the arguments of $\dx_p$ are exactly
the free variables of $\subform{F}{p}$ the applied Skolemization also
corresponds to inner Skolemization \emph{with respect to the original formula}
$F$ before translation to definitional form.

\section{Access Interpolant Extraction from Clausal Tableaux}
\label{sec-access-extract}

To permit extraction of access interpolants, clausal tableaux have to satisfy
certain restrictions that are specified with Definition~\ref{def-aitab} below.
Aside of the \name{regular}, \name{closed} and \name{leaf-only} properties,
which have already been specified, a further property is now needed:
\begin{defn}[Contiguous]
  \label{def-contiguous}
A clausal tableau is called \defname{contiguous} for an unordered pair of
literals if and only if whenever both members of the pair occur as literal
labels of two nodes on the same branch, one of the nodes is the parent of the
other.
\end{defn}
The \name{contiguous} property is used to represent relativized quantification
by handling conjuncts in the scope of an existential quantifier
simultaneously, specifically the atom that relativizes the quantified
variables and a second atom with a definer predicate that represents the
argument of the relativized quantification.  For this application the
\name{contiguous} property can be ensured by a tableau simplification,
Procedure~\ref{proc-contig} shown in Sect.~\ref{sec-access-convert}.  We have
now specified all prerequisites to define the constraint package on clausal
tableaux for access interpolation and call tableaux that satisfy it
\name{\acitx}, suggesting \name{ACcess Interpolation}:
\begin{defn}[\ACIT]
  \label{def-aitab}
  Let $F, G$ be \RQFO sentences.  An \defname{\acit for $F$ and $G$} is a
  closed \sided clausal ground tableau for two clausal formulas obtained from
  $F$ and~$G$ by clausifying $\DEFP{F}$ and $\DEFPNOT{G}$ as specified in
  Lemma~\ref{lem-input-clauseforms} that is regular, leaf-only for the
  set of all negative literals that occur as literal labels in it, and
  contiguous for all pairs of ground literals of that occur as literal labels
  in it and have, referring to the notation of
  Lemma~\ref{lem-input-clauseforms}, the form $\{R_p\stcr_p\sti,\;
  D_{p1}\stcr_p\sti\}$ for some position $p$ in $F$ or in $G$ and some ground
  substitution $\sti$.
\end{defn}
Note that an \acit is a special case of a closed positive hyper tableau
(Definition~\ref{def-phyper}).
The specification of the extraction of an access interpolant from a clausal
tableau involves a form of lifting that differs from the lifting described for
Craig-Lyndon interpolants with Procedure~\ref{proc-cti-fol}.  For access
interpolation lifting can not be performed globally on a ground interpolant
but on subformulas that correspond to the scopes of relativized
quantifiers. To specify this form of lifting we need further auxiliary
concepts that concern those occurrences of ground terms in a formula that are
as argument of an atom, in contrast to embedded in another term.  Symbolic
notation for referring to the set of terms with such occurrences as well as
for systematically replacing these occurrences with variables is provided.
Preconditions are made precise under which an entailment relationship between
formulas still holds after such a replacement by variables.
\begin{defn}[Set of Ground Arguments of Atoms]
  If $F$ is a formula, then $\garg{F}$ denotes the set of ground terms in
  $\parg{F}$.
\end{defn}
For example, if $x,y$ are variables and $\ak,\bk$ are constants, then
\[\garg{\forall x\, \pk(\ak,\ggk(\ak),\ggk(\bk),x,\ffk(y,\bk))} = \{\ak,
\ggk(\ak), \ggk(\bk)\}.\]  For
relational formulas~$F$ it holds that $\garg{F} = \const{F} = \fun{F}$.  Based
on $\garg{F}$, we define for injective substitution~$\sigma$ the following
restricted variant of $\revsubst{F}{\sigma}$:
\begin{defn}[Inverse Substitution of Ground Arguments of Atoms]
If $F$ is a formula and $\sigma$ is an injective ground substitution such that
$\rng{\sigma} \subseteq \garg{F}$, then let $\toprevsubst{F}{\sigma}$ denote
$F$ with all occurrences of members~$t$ of $\rng{\sigma}$ that are as argument
of an atom replaced with the variable mapped by $\sigma$ to $t$.
\end{defn}
While in $\revsubst{F}{\sigma}$ occurrences of terms that are not strict
subterms of some other member of $\rng{\sigma}$ are replaced, in
$\toprevsubst{F}{\sigma}$ only occurrences that are arguments of atoms are
replaced. The following proposition relates these two forms of ``inverse
substitution'':

\pagebreak 
\begin{prop}[Inverse Substitution of Arguments of Atoms and of Terms]
\label{prop-toprevsubst}
Let $F$ be a formula in which all non-ground terms are variables, let $\sigma$
be an injective ground substitution such that $\rng{\sigma} \subseteq \garg{F}$
and no member of $\dom{\sigma}$ occurs in $F$, and let $\gamma$ be an
injective substitution such that $\rng{\gamma} = \garg{F}$, no member of
$\dom{\gamma}$ occurs in $F$ and $\gamma|_{\dom{\sigma}} = \sigma$.  Then
\[\toprevsubst{F}{\sigma} = \revsubst{F}{\gamma}\gamma|_{\dom{\gamma} \setminus
\dom{\sigma}}.\]
\end{prop}

\noindent
The following proposition states a variant of
Proposition~\ref{prop-qu-entailment} where occurrences of possibly
\emph{complex} ground terms that themselves are not subterms of other terms
are replaced by quantified variables. We will apply it later to justify
lifting from ground terms introduced through Skolemization to quantified
variables.
\begin{prop}[Inessential Quantifications in Entailments for Terms]
\label{prop-inessential-qu-terms}
Let $F, G$ be formulas in which no non-ground terms with the exception of
variables occur.  Let $\sigma$ be a ground substitution such that
$\rng{\sigma} \subseteq \garg{F}$, $\rng{\sigma} \cap \garg{G} = \emptyset$ and
no member of $\dom{\sigma}$ occurs in $F$ or in $G$. Let $\xxs$ stand for
$\dom{\sigma}$.  Then
\[\exists \xxs\, 
\toprevsubst{F}{\sigma} \entails G\; \text{ if and only if }\; F \entails G.\]
\end{prop}

We are now equipped with the prerequisites to specify the extraction of an
access interpolant from an \acit, that is, a constructive mapping from an
\acit for two \RQFO sentences~$F$ and~$G$ such that $F \entails G$ to an access
interpolant of them.  The correctness of the mapping is then stated and proven
as Theorem~\ref{thm-correct-aipol}.
\begin{defn}[Access Interpolant Extraction from an \ACIT]
\label{def-aipol}
Let $F, G$ be \RQFO sentences such that $F \entails G$ and let $T$ be an \acit
for~$F$ and~$G$. For all inner nodes $N$ of $T$ define $\aipol{N}$ inductively
as follows, where $N_1, \ldots, N_k$ with $k \leq 1$ are the children of $N$,
and clause forms are understood as specified in
Lemma~\ref{lem-input-clauseforms}:
\begin{enumerate}[label={\roman*}.,leftmargin=2em]

\item \label{case-unit} Case $\nclause{N}$ is an instance
  of form~1:
  $\aipol{N} \eqdef \aipol{N_1}$.

\item \label{case-and-or} Case $\nclause{N}$ is an instance of
  one of forms~2--5 or~7--8:
  \begin{enumerate}[label={\alph*}.,leftmargin=1.5em,ref={\roman{enumi}.\alph*}]
    \item \label{case-or} Case $\nside{N_1} = \aaa$:
      \[\aipol{N} \eqdef \bigvee_{i = 2}^{k} \aipol{N_i}.\]
    \item \label{case-and} Case $\nside{N_1} = \bbb$:
      \[\aipol{N} \eqdef \bigwedge_{i = 2}^{k} \aipol{N_i}.\]
  \end{enumerate}
  
\item \label{case-quant} Case $\nclause{N}$ is an instance $\lnot D \lor \lnot
  R \lor D'$ of form~6: Since the tableau is closed and regular there is a
  unique ancestor $\ntgt{N_2}$ of $N_2$ with $\nlit{\ntgt{N_2}} =
  \du{\nlit{N_2}}$.
 
  \begin{enumerate}[label={\alph*}.,leftmargin=1.5em,ref={\roman{enumi}.\alph*}] 
    \item \label{case-samecol} 
      Case $\nside{\ntgt{N_2}} = \nside{N_1}$: $\aipol{N} \eqdef \aipol{N_3}$.
      
    \item \label{case-all} Case $\nside{N_1} = \aaa$ and $\nside{\ntgt{N_2}} =
      \bbb$:  Let \[\{t_1, \ldots, t_n\} \eqdef
      \garg{\lnot R} \setminus
      \garg{\DEFP{F} \land
        \nbranchA{N}},\footnote{Definition~\ref{def-aipol}
        is slightly different from a
        straightforward transfer of the corresponding specification in terms
        of tableau rules in \cite{benedikt:book}: In case~(\ref{case-all}) of
        Definition~\ref{def-aipol} the range $\{t_1, \ldots, t_n\}$ of $\stx$
        is specified as a subset of $\garg{\lnot R}$, whereas according
        to \cite[Figure2.8]{benedikt:book} one would expect $\{t_1, \ldots,
        t_n\} = \garg{\lnot R \lor \aipol{N_3}}$ $\setminus$
        $\garg{\DEFP{F} \land \nbranchA{N}}$.  The inclusion of
        $\aipol{N_3}$ on the left side of the $\setminus$ operator would,
        however, actually be redundant.  Analogous considerations hold for the
        case~(\ref{case-ex}).}
      \]
      let $v_1, \ldots, v_n$ be fresh variables, let $\stx \eqdef \{v_1
      \mapsto t_1, \ldots, v_n \mapsto t_n\}$, and define
      \[\aipol{N} \eqdef
      \toprevsubst{\forall v_1 \ldots \forall v_n\, (\lnot R \lor
        \aipol{N_3})}{\stx}.\]
      
    \item \label{case-ex} Case $\nside{N_1} = \bbb$ and $\nside{\ntgt{N_2}} =
      \aaa$: Let \[\{t_1, \ldots, t_n\} = \garg{R} \setminus
      \garg{\DEFPNOT{G} \land \nbranchB{N}},\] let $v_1, \ldots, v_n$ be fresh
      variables, let $\stx = \{v_1 \mapsto t_1, \ldots, v_n \mapsto t_n\}$,
      and define
      \[\aipol{N} \eqdef
      \toprevsubst{\exists v_1 \ldots \exists v_n\, 
        (R \land \aipol{N_3})}{\stx}.\]
\end{enumerate}
\end{enumerate}
\end{defn}
Although base cases are not explicitly distinguished in the inductive
definition of $\aipol{N}$, they are covered by the specification in
Definition~\ref{def-aipol}: If $\nclause{N}$ is an instance of form $2$ of
Lemma~\ref{lem-input-clauseforms}, then $k = 1$ and $\aipol{N} = \false$ or
$\aipol{N} = \true$, respectively.

Correctness of the access interpolant extraction according to
Definition~\ref{def-aipol} is stated with the following theorem:
\begin{thm}[Correctness of Access Interpolant Extraction from an \ACIT] 
\label{thm-correct-aipol}
  Let $F, G$ be \RQFO sentences such that $F \entails G$ and let $N$ be the
  root of an \acit for $F$ and~$G$. Then $\aipol{N}$ is an access interpolant
  of $F$ and $G$.
\end{thm}
Before we can proof this theorem, we need some auxiliary concepts and
propositions. 
An \acit is based on the conjunction of two clausal formulas, each obtained
from one of the two input sentences.  \name{Global position specifiers} allow
to refer unambiguously to each literal occurrence and further items in this
conjunction:
\begin{defn}[Global Position Specifier]
  \label{def-globalpos}
  Consider an \acit for \RQFO sentences $F$ and $G$.  It is a clausal tableau
  for two clausal formulas obtained by clausifying $\DEFP{F}$ and
  $\DEFPNOT{G}$.  For $\XS \in \{\LL, \RR\}$ define $D_{\XS p}$, $R_{\XS p}$,
  $\xxs_{\XS p}$, $\vvs_{\XS p}$, $\skf_{\la {\XS p}, i\ra}$, $\sigma_{\XS p}$
  to denote $D_p$, $R_p$, $\xxs_p$, $\vvs_p$, $\skf_{\la p, i\ra}$,
  $\sigma_p$, respectively, as specified in
  Definition~\ref{lem-input-clauseforms}, in case $\XS = \LL$ referring to $p
  \in \poss{F}$ and clauses obtained from $\DEFP{F}$ and in case $\XS = \RR$
  referring to $p \in \poss{\lnot G}$ and clauses obtained from
  $\DEFPNOT{G}$. Position specifiers of the form $\XS p$, where $\XS \in
  \{\LL, \RR\}$ and $p$ denotes a position as specified in
  Definition~\ref{def-position} are called \defname{global position
    specifiers}. The symbol~$\XS$ is called the \defname{side} of a global
  position identifier~$\XS p$.
\end{defn}

\noindent
To mimic the $\delta$-rule of non-clausal tableaux we specify the notion of
\name{introducer literal} and \name{introducer node} associated with each
``Skolem term'', that is, ground term whose principal functor is a Skolem
function:
\begin{defn}[Ground Term Introducers]
  Let $T$ be an \acit. The \defname{introducer literals} for a ground term
  $\skf_{\la \XS p, i\ra}(\dx_{\XS p})\sti$ occurring in a literal label of a
  tableau node are $R_{\XS p}\stcr_{\XS p}\sti$, and $D_{\XS p1}\stcr_{\XS
    p}\sti$.  An \name{introducer node} for a ground term is a node whose
  literal label is an introducer literal for the term.
\end{defn}
The following proposition shows a relationship of occurrences of Skolem terms
and their introducers that holds for \acitx:
\begin{prop}[Precedence of Ground Term Introducers in \ACITX]
\label{prop-aux-skolem-strong} 
Let~$N$ be a node of an \acit and let $t \in \garg{\nlit{N}}$ where the
principal functor of $t$ is a Skolem function $\skf_{\la \XS p, i\ra}$.  Then
$N$ is an introducer node for $t$ or $N$ has an ancestor that is an introducer
node for $t$.
\end{prop}
\begin{proof}
Assume that $t \in \garg{\nlit{N}}$ and $N$ is not an introducer node for
$t$. We show that~$N$ then has an ancestor~$N^\prime$ such that $t \in
\garg{\nlit{N^\prime}}$.  The proposition then follows from finiteness of the
tableau branch length.  Numbers of clause forms refer to
Lemma~\ref{lem-input-clauseforms}. Let $\nparent{N}$ denote the parent of $N$.
Then:

\begin{enumerate}[label={\roman*}.,leftmargin=2em]
\item \label{item-intro-dp} If $\nlit{N}$ is negative, then there must exist
  an ancestor $N^\prime$ of $N$ with $\nlit{N^\prime} = \du{\nlit{N}}$, and
  thus $t \in \garg{\nlit{N^\prime}}$, as claimed.

\item Else, if $\nlit{N}$ is of the form $D_{\XS p}\sti$, then
  $\nclause{\nparent{N}}$ must be an instance of a clause of one of the forms
  2--6. (Form~1 can be excluded since $\parg{D_{\XS\emptyseq}} = \emptyset$,
  contradicting our assumption $t \in \garg{\nlit{N}}$.) In all cases it
  can be verified that $\garg{D_{\XS p}\sti} \subseteq \garg{E}$, where~$E$ is
  the disjunction of the negative literals in $\nclause{\nparent{N}}$. Hence
  $N$ must have a sibling $N''$ with a negative literal label and such that $t
  \in \garg{\nlit{N''}}$. The existence of an ancestor $N^\prime$ of $N$ as
  claimed then follows from~(\ref{item-intro-dp}).
 
\item Else $\nlit{N}$ must be of the form $R_{\XS p}\stcr_p\sti$ or $D_{\XS
  p1}\stcr_p\sti$ and $\nclause{\nparent{N}}$ must be an instance of a
  clause of form~7 or 8.  Because $N$ is not an introducer node for $t$ it
  follows that $t \in \{x\sti \mid x \in \var{R_{\XS p}\stcr_{\XS p}}\}$ or
    $t \in \{x\sti \mid x \in \var{D_{\XS p1}\stcr_{\XS p}}\}$,
    respectively. With $t \in \garg{\nlit{N}}$ it follows from the
    specification of clause forms~7 and~8 that $t \in \garg{\lnot D_{\XS
        p}\sti}$.
  Because $N$ has a sibling whose literal label is $\lnot D_{\XS p}\sti$, the
  existence of an ancestor~$N^\prime$ of~$N$ as claimed follows
  from~(\ref{item-intro-dp}). \qed
\end{enumerate}
\end{proof}

\noindent
In the proof of Theorem~\ref{thm-correct-aipol} semantic and syntactic
properties of intermediate formulas constructed during access interpolant
extraction need to be considered. The notions of \name{\RQFO formula} and
\name{access interpolant} as such are not adequate to express the relevant
properties of these intermediate formula, but generalizations of them, defined
as follows:
\begin{defn}[\RQFOT Formula, Weak Access Interpolant]
\ 

\sdlab{def-rqfot} Formulas of \name{first-order logic with relativized
  quantifiers and ground terms}, briefly \name{\RQFOT formulas}, are defined
like \RQFO formulas (Definition~\ref{def-rqfo}) with the exception that as arguments
of atoms not just variables and constants, but also ground terms with function
symbols of arbitrary arity are allowed.

\smallskip

\sdlab{def-wai} Let $F,G$ be \RQFO sentences and let $F^\prime,G^\prime$ be
quantifier-free first-order formulas such that $F \land F^\prime \entails G
\lor G^\prime$.  A \defname{weak access interpolant} of the quadruple $\la F,
F^\prime, G, G^\prime\ra$ is an \RQFOT sentence~$H$ such that
\begin{enumerate}
\item \label{cond-wai-sem} $F \land F^\prime \entails H \entails G \lor G^\prime$.
\item \label{cond-wai-lyndon} $\pred{H} \subseteq \pred{F} \cap \pred{G}$.
\item \label{cond-wai-relat} Every existential \bpatt of $H$ is covered by an
  existential \bpatt of $G$.  Every universal \bpatt of $H$ is covered by a
  universal \bpatt of $F$.
\item \label{cond-wai-const} 
$\garg{H} \subseteq \garg{F \land F^\prime} \cap \garg{G \lor G^\prime}.$
\end{enumerate}
\end{defn}
Access interpolants are special cases of weak access interpolants:
\begin{prop}[Weak and Standard Access Interpolants]
\label{prop-weak-std-ai}
Let $F,G$ be \RQFO sentences. A formula $H$ is an access interpolant of $F$
and $G$ if and only if $H$ is a weak access interpolant of $\la F, \true, G,
\false \ra$.
\end{prop}
\begin{proof}
Easy to see from the definitions of \name{weak access interpolant}
(Definition~\ref{def-wai}) and \name{access interpolant} (Definition~\ref{def-ai}).
\qed
\end{proof}

We are now ready to prove the core property that underlies the correctness of
the access interpolant extraction from \acitx:
\begin{lem}[Core Invariant of Access Interpolant Extraction from
 \ACITX]
  \label{lem-aipol-invariant}
  Let $F, G$ be \RQFO sentences such that $F \entails G$ and let $T$ be an
  \acit for $F$ and~$G$.  For all inner nodes~$N$ of $T$ the formula
  $\aipol{N}$ is a weak access interpolant of $\la \DEFP{F}, \nbranchA{N},
  \lnot \DEFPNOT{G}, \lnot \nbranchB{N} \ra$.
\end{lem}

\begin{proof}
  By induction on the tableau structure. The property
  to show for all nodes $N$ of the tableau is:
  \begin{tabularlabtext}
  (IP) & If $N$ is an inner node, then 
    $\aipol{N}$ is a weak access interpolant of $\la \DEFP{F}, \nbranchA{N},
  \lnot \DEFPNOT{G}, \lnot \nbranchB{N} \ra$.
  \end{tabularlabtext}
  In the base case where $N$ is a leaf it satisfies (IP) trivially. To prove
  the induction step assume as induction hypothesis that $N$ is an inner node
  with children $N_1,\ldots, N_k$ where $k \geq 1$ and that (IP) holds for all
  children, that is:
  \begin{tabularlabtext}
    (IH) & For all $i \in \{1,\ldots,k\}$ it holds that if $N_i$ is an inner node,
    then $\aipol{N_i}$ is a weak access interpolant of

    \medskip
    
    \hspace*{\fill}$\la \DEFP{F}, \nbranchA{N_i}, \lnot \DEFPNOT{G},
    \lnot \nbranchB{N_i} \ra.$\hspace*{\fill}
\end{tabularlabtext}
  We prove the induction step by showing that (IH) implies that (IP) holds for
  $N$, that is, $\aipol{N}$ is a weak access interpolant of
\[\la \DEFP{F}, \nbranchA{N}, \lnot \DEFPNOT{G}, \lnot \nbranchB{N}\ra.\]
We will now prove this for the case where the children of $N$ have side label
$\aaa$ for all possible forms of $\nclause{N}$ according to
Lemma~\ref{lem-input-clauseforms}. The case where the children have side label
$\bbb$ can be shown analogously.
We thus assume that $N$ is an inner node of $T$ and that the children of $N$
have side label $\aaa$.  The following general lemma is then easy to
verify:
\begin{tabularlabtext}
  (LR) & For all $i \in \{1,\ldots,k\}$ it holds
  that $\nbranchB{N} = \nbranchB{N_i}$.
\end{tabularlabtext}
For all clause forms with exception of form~1 the literal $\nlit{N_1}$ must be
an instance of a literal of the form $\lnot D_{\LL p}$.  Since the tableau is
closed and leaf-only it follows for ground substitutions~$\sti$ such that
$\lnot D_{\LL p}\sti = \nlit{N_1}$ that:
\begin{tabularlabtext}
(LD) & $D_{\LL p}\sti$ (which is equal to $\du{\nlit{N_1}}$)
occurs as a conjunct in $\nbranchA{N}$.
\end{tabularlabtext}

The formula $\nclause{N}$ must be an instance of a clause of one of the forms
listed in Lemma~\ref{lem-input-clauseforms}.  We now consider each possible
case in subsections headed with the respective clause forms.  For each case we
verify that $\aipol{N}$ satisfies the characteristics
\ref{cond-wai-sem}--\ref{cond-wai-const} of \name{weak access interpolant}
according to Definition~\ref{def-wai}.  We label the respective subproofs with
\name{\waicond~\ref{cond-wai-sem} left}, \name{\waicond~\ref{cond-wai-sem}
  right}, \name{\waicond~\ref{cond-wai-lyndon}},
\name{\waicond~\ref{cond-wai-relat}}, and
\name{\waicond~\ref{cond-wai-const}}, respectively.  Condition
\name{\waicond~\ref{cond-wai-sem}} is there split up into a \name{left}
component, that is, $\DEFP{F} \land \nbranchA{N} \entails \aipol{N}$ and a
\name{right} component, that is, $\aipol{N} \entails \lnot \DEFPNOT{G} \lor
\lnot \nbranchB{N}$, or, equivalently, expressed as contrapositive,
$\DEFPNOT{G} \land \nbranchB{N} \entails \lnot \aipol{N}$.
If appropriate, proof steps are shown in tabular symbolic form followed by
explanations. 

\casesection{Clause Form 1}

For this clause form $\aipol{N}$ is defined as $\aipol{N_1}$.  This case can
be proven with a simplified variant of the proof for clause forms~2--5 below.
The role of $N_2$ in that other proof is taken here by $N_1$ and properties
$\DEFP{F} \entails D_{\LL \emptyseq}$ as well as $\garg{D_{\LL \emptyseq}}
= \emptyset$ can be utilized.

\casesection{Clause Form 2--5}

For these clause forms $\aipol{N}$ is defined as $\bigvee_{i=2}^{k}
\aipol{N_i}$.  Let~$\sti$ be a ground substitution such that $\dom{\sti} =
\parg{\lnot D_{\LL p}}$ and $\lit{N_1} = \lnot D_{\LL p}\sti$.

\begin{caselist}

  \casepara{\waicond~\ref{cond-wai-sem} left}
\prlReset{case-or-sem-left}
\[
\begin{arrayprf}
\prl{1} & \DEFP{F} \entails \forall \dx_{\LL p}\, (D_{\LL p} \imp
D_{\LL p1} \lor \ldots \lor D_{\LL p(k-1)}).\\
\prl{2} & \DEFP{F} \entails D_{\LL p}\sti \imp
D_{\LL p1}\sti \lor \ldots \lor D_{\LL p(k-1)}\sti.\\
\prl{3} & \nbranchA{N} \entails D_{\LL p}\sti.\\
\prl{4} & \DEFP{F} \land \nbranchA{N} \entails
D_{\LL p1}\sti \lor \ldots \lor D_{\LL p(k-1)}\sti.\\
\prl{5} & \DEFP{F} \land \nbranchA{N} \entails
\aipol{N_2} \lor \ldots \lor \aipol{N_k}.\\
\prl{6} & \DEFP{F} \land \nbranchA{N} \entails
\aipol{N}.
\end{arrayprf}
\]

Entailment~\pref{1} holds since its right side is a conjunct of its left
side. Entailment~\pref{2} follows from~\pref{1} by instantiating universally
quantified variables. Entailment~\pref{3} follows from~(LD).
Entailment~\pref{4} follows from~\pref{3} and~\pref{2}. Entailment~\pref{5}
follows from~\pref{4} and~(IH). Entailment~\pref{6} follows from~\pref{5} and
the definition of $\aipol{N}$ for the considered clause forms.

\casepara{\waicond~\ref{cond-wai-sem} right}
\prlReset{case-or-sem-right}
\[
\begin{arrayprf}
\prl{1} & \DEFPNOT{G} \land \nbranchB{N} \entails 
\lnot \aipol{N_2} \land \ldots \land \lnot \aipol{N_k}.\\
\prl{2} & \DEFPNOT{G} \land \nbranchB{N} \entails
\lnot \aipol{N}.\\
\end{arrayprf}
\]
Entailment~\pref{1} follows from~(IH) and~(LR).  Entailment~\pref{2} follows
from~\pref{1} and the definition of $\aipol{N}$ for the considered clause
forms.

\casepara{\waicond~\ref{cond-wai-lyndon} and~\ref{cond-wai-relat}}
Immediate from~(IH) and the definition of $\aipol{N}$ for the considered
clause forms.

\casepara{\waicond~\ref{cond-wai-const}}
\prlReset{case-or-const}
\[
\begin{arrayprf}
\prl{1} & \garg{D_{\LL pi}\sti} \subseteq
\garg{\lnot D_{\LL p}\sti}, \text{ for all } i \in \{1, \ldots, k-1\}.\\
\prl{2} & \garg{\nbranchA{N}} = \garg{\nbranchA{N_i}}, 
      \text{ for all } i \in \{2, \ldots, k\}.\\
\prl{3} & \garg{\nbranchB{N}} = \garg{\nbranchB{N_i}},
      \text{ for all } i \in \{2, \ldots, k\}.\\
\prl{4} & \garg{\aipol{N}}\;  \subseteq\\
      & \garg{\DEFP{F} \land 
        \nbranchA{N}} \cap \garg{\lnot \DEFPNOT{G} \lor \lnot \nbranchB{N}}.
\end{arrayprf}
\]
Subsumption~\pref{1} follows from the specification of the considered clause
forms.  Equality~\pref{2} follows from~\pref{1}, since for $i \in
\{2,\ldots,k\}$ it holds that $\nlit{N_i} = D_{\LL
  p(i-1)}\sti$. Equality~\pref{3} follows from~(LR).  Subsumption~\pref{4}
follows from~(IH), \pref{3} and~\pref{2}, given that $\garg{\aipol{N}} =
\bigcup_{i=2}^{k} \garg{\aipol{N_i}}$.

\end{caselist}

\casesection{Clause Form~6, Case
  {\bfseries \upshape \sffamily
    side(tgt({\itshape \rmfamily N}$\mathbf{_2}$)) = L}$\,$}
In this case $\aipol{N}$ is defined as $\aipol{N_3}$.  Le $\sti$ be a ground
substitution such that $\dom{\sti} = \parg{\lnot D_{\LL p}} \cup \vout_{\LL p}$ and
$\nclause{N} = (\lnot D_{\LL p} \lor \lnot R_{\LL p} \lor D_{\LL p1})\sti$.
We note the following lemma, which can be derived similarly as~(LD):
\begin{tabularlabtext}
(L3) & The literal $R_{\LL p}\sti$ (that is, $\du{\nlit{N_2}}$)
 occurs as a conjunct in $\nbranchA{N}$.
\end{tabularlabtext}

\begin{caselist}

\casepara{\waicond~\ref{cond-wai-sem} left}
\prlReset{case-samecol-wai-sem-left}
\[
\begin{arrayprf}
\prl{1} & \DEFP{F} \land \nbranchA{N_3} \entails \aipol{N_3}.\\
\prl{2} & \DEFP{F} \land \nbranchA{N} \land D_{\LL p1}\sti \entails \aipol{N_3}.\\
\prl{3} & \DEFP{F} \entails 
  \forall \dx_{\LL p} (D_{\LL p} \imp \forall \vout_{\LL p} (\lnot R_{\LL p} \lor D_{\LL p1})).\\
\prl{4} & \DEFP{F} \entails \lnot D_{\LL p}\sti \lor \lnot R_{\LL p}\sti \lor D_{\LL p1}\sti.\\
\prl{5} & \nbranchA{N} \entails D_{\LL p}\sti.\\
\prl{6} & \nbranchA{N} \entails R_{\LL p}\sti.\\
\prl{7} & \DEFP{F} \land \nbranchA{N} \entails \aipol{N_3}.\\
\prl{8} & \DEFP{F} \land \nbranchA{N} \entails \aipol{N}.\\
\end{arrayprf}
\]
Entailment~\pref{1} follows from~(IH).  Entailment~\pref{2} is obtained from~\pref{1} by
expressing $\nbranchA{N_3}$ with its last conjunct made explicit.
Entailment~\pref{3} holds since its right side is a conjunct of its left side.
Entailment~\pref{4} follows from \pref{3} by instantiating universally quantified
variables. Entailments~\pref{5} and \pref{6} follow from~(LD) and~(L3), respectively.
Entailment~\pref{7} follows from~\pref{4}--\pref{6} and \pref{2}.  Step~\pref{8} follows from~\pref{7} since
$\aipol{N} = \aipol{N_3}$.

\casepara{\waicond~\ref{cond-wai-sem} right}
Immediate from (IH) and (LR) since  $\aipol{N} = \aipol{N_3}$.

\casepara{\waicond~\ref{cond-wai-lyndon} and~\ref{cond-wai-relat}}
Immediate from~(IH) since $\aipol{N} = \aipol{N_3}$.

\casepara{\waicond~\ref{cond-wai-const}}
\prlReset{case-samecol-wai-const}
\[
\begin{arrayprf}
\prl{1} & \garg{\lnot R_{\LL p}\sti} \cup \garg{\lnot D_{\LL p}\sti} \subseteq
\garg{\nbranchA{N}}.\\
\prl{2} & \garg{D_{\LL p1}\sti} \subseteq \garg{\lnot R_{\LL p}\sti} \cup \garg{\lnot D_{\LL p}\sti}.\\
\prl{3} & \garg{\nbranchA{N}} = \garg{\nbranchA{N_3}}.\\
\prl{4} & \garg{\nbranchB{N}} = \garg{\nbranchB{N_3}}.\\
\prl{5} & \garg{\aipol{N}}\; \subseteq\\
& \garg{\DEFP{F} \land \nbranchA{N}}
\cap \garg{\lnot \DEFPNOT{G} \lor \lnot \nbranchB{N}}.
\end{arrayprf}
\]
Subsumption~\pref{1} follows from~(LD) and~(L3).  Subsumption~\pref{2} follows
from the definition of clause form~6. Equality~\pref{3} follows from~\pref{2}
and \pref{1} since $\nbranchA{N_3} = \nbranchA{N} \land D_{\LL p1}\sti$.
Equality~\pref{4} follows from~(LR). Subsumption~\pref{5} follows from~(IH),
\pref{4} and~\pref{3}, since $\aipol{N} = \aipol{N_3}$.

\end{caselist}

\casesection{Clause Form~6, Case
  {\bfseries \upshape \sffamily
    side(tgt({\itshape \rmfamily N}$\mathbf{_2}$)) = R}$\,$}

Let $\sti$ be a ground substitution such that $\dom{\sti} = \parg{\lnot D_{\LL
    p}} \cup \vout_{\LL p}$ and $\nclause{N} = (\lnot D_{\LL p} \lor \lnot
R_{\LL p} \lor D_{\LL p1})\sti$.  Formula $\aipol{N}$ is for this case then
defined as \[\forall v_1 \ldots \forall v_n\, \toprevsubst{(\lnot R_{\LL
    p}\sti \lor \aipol{N_3})}{\stx},\] where $\{t_1, \ldots, t_n\} =
\garg{\lnot R_{\LL p}\sti} \setminus \garg{\DEFP{F} \land \nbranchA{N}}$,
$v_1, \ldots, v_n$ are fresh variables, and $\stx = \{v_1 \mapsto t_1, \ldots,
v_n \mapsto t_n\}$.

\begin{caselist}
  \casepara{\waicond~\ref{cond-wai-sem} left}
  \prlReset{case-all-sem-left}%
\[
\begin{arrayprf}
\prl{1} & \DEFP{F} \land \nbranchA{N_3} \entails \aipol{N_3}.\\
\prl{2} & \DEFP{F} \land \nbranchA{N} \land D_{\LL p1}\sti 
\entails \aipol{N_3}.\\
\prl{3} & \DEFP{F} \entails 
\forall \dx_{\LL p}\, (D_{\LL p} \imp \forall \vout_{\LL p}\, (\lnot R_{\LL p} \lor D_{\LL p1})).\\
\prl{4} & \DEFP{F} \entails \lnot D_{\LL p}\sti \lor \lnot R_{\LL p}\sti \lor D_{\LL p1}\sti.\\
\prl{5} & \nbranchA{N} \entails D_{\LL p}\sti.\\
\prl{6} & \DEFP{F} \land \nbranchA{N} 
  \entails \lnot R_{\LL p}\sti \lor \aipol{N_3}.\\
\prl{7} & \DEFP{F} \land \nbranchA{N} 
  \entails \forall v_1 \ldots \forall v_n 
   \toprevsubst{(\lnot R_{\LL p}\sti \lor \aipol{N_3})}{\stx}.\\
\prl{8} & \DEFP{F} \land \nbranchA{N} \entails \aipol{N}.
\end{arrayprf}
\]
Entailments~\pref{1}--\pref{5} follow in the same way as in the as shown above
for clause form~6, case $\nside{\ntgt{N_2}} = \aaa$,
\name{\waicond~\ref{cond-wai-sem} left}.  Entailment~\pref{6} follows
from~\pref{5}, \pref{4} and~\pref{2}.  Given the specified properties of
$\stx$, entailment~\pref{7} follows from~\pref{6} by
Proposition~\ref{prop-inessential-qu-terms}.  Step~\pref{8} is obtained
from~\pref{7} by contracting the definition of $\aipol{N}$ for the considered
case.

\casepara{\waicond~\ref{cond-wai-sem} right}
\prlReset{case-all-sem-right}%
\[
\begin{arrayprf}
\prl{1} & \DEFPNOT{G} \land \nbranchB{N_3} \entails \lnot \aipol{N_3}.\\
\prl{3} & \nbranchB{N} \entails R_{\LL p}\sti.\\
\prl{4} & \DEFPNOT{G} \land \nbranchB{N} \entails 
R_{\LL p}\sti \land \lnot \aipol{N_3}.\\
\prl{5} & \DEFPNOT{G} \land \nbranchB{N} \entails 
      \toprevsubst{\exists v_1 \ldots \exists v_n\, 
      (R_{\LL p}\sti \land \lnot \aipol{N_3})}{\stx}.\\
\prl{6} & \DEFPNOT{G} \land \nbranchB{N} \entails 
      \toprevsubst{\lnot \forall v_1 \ldots \forall v_n\, 
      (\lnot R_{\LL p}\sti \lor \aipol{N_3})}{\stx}.\\
\prl{7} & \DEFPNOT{G} \land \nbranchB{N} \entails \lnot \aipol{N}.\\
\end{arrayprf}
\]
Entailment~\pref{1} follows from~(IH).  Entailment~\pref{3} holds since
$\nlit{N_2} = \lnot R_{\LL p}\sti$ and $\nside{\ntgt{N_2}} = \bbb$.
Entailment~\pref{4} follows from~\pref{3}, (LR) and~\pref{1}.
Entailment~\pref{5} follows from~\pref{4}, since the formula on the right side
of \pref{5} is entailed by the formula on the right side of~\pref{4}. The
formula on the right side of~\pref{6} is equivalent to that on the right side
of~\pref{5}.  Entailment~\pref{7} is obtained from~\pref{6} by contracting the
definition of $\aipol{N}$ for the considered case.

\casepara{\waicond~\ref{cond-wai-lyndon}}
The formula $\aipol{N}$ contains, compared to $\aipol{N_3}$ one additional
predicate occurrence, a negative occurrence of the predicate of $R_{\LL
  p}\sti$.  This predicate occurs in an instance of a clause of form~6, hence
negatively in $\DEFP{F}$. It also occurs positively in the literal label
of $\ntgt{N_2}$, where $\nside{\ntgt{N_2}} = \bbb$, hence positively in a
clause of form~7 obtained from normalizing $\DEFPNOT{G}$, hence negatively in
$\lnot \DEFPNOT{G}$.

\casepara{\waicond~\ref{cond-wai-relat}}
From~(IH) it follows that all existential \bpatts of $\aipol{N_3}$ are covered
by $\lnot \DEFPNOT{G}$ and all universal \bpatts of $\aipol{N_3}$ are covered
by $\DEFP{F}$.  That the \bpatts of $\aipol{N}$, defined as
\[\forall v_1
\ldots \forall v_n\, \toprevsubst{(\lnot R_{\LL p}\sti \lor
  \aipol{N_3})}{\stx},\] are also covered in that way by $\lnot \DEFPNOT{G}$
and $\DEFP{F}$ then follows if the outermost quantification of $\aipol{N}$
is covered by the quantification upon $\forall \vout_{\LL p}$ in the formula
$\forall \dx_{\LL p} (D_{\LL p} \imp \forall \vout_{\LL p}\, (\lnot R_{\LL p}
\lor D_{\LL p1}))$, which is a conjunct of $\DEFP{F}$. This, in turn,
follows if each member of the range of $\stx$ occurs in $R_{\LL p}\sti$ in an
argument position that is an ``output position'' of $R_{\LL p}$, that is, the
argument of $R_{\LL p}$ at that position is a member of $\vout_{\LL p}$.
We show the latter statement.  Let $t$ be a member of the range of $\stx$.
From the definition of $\stx$ it follow that $t \in \garg{\lnot R_{\LL
    p}\sti}$.  Assume that $t$ occurs in $\lnot R_{\LL p}\sti$ in a
``non-output'' position, that is, at an argument position of $R_{\LL p}\sti$
at which the argument of $R_{\LL p}$ is no member of $\vout_{\LL p}$.  From
the definition of the considered clause form~6 it follows that then $t \in
\garg{\lnot D_{\LL p}\sti}$ or $t \in \garg{\lnot R_{\LL p}}$. By~(LD) and
since $\garg{\lnot R_{\LL p}} \subseteq \garg{\DEFP{F}}$ it follows that
$t \in \garg{\nbranchA{N}} \cup \garg{\DEFP{F}}$.  From the specification
of $\stx$ it follows that its $\rng{\stx} \cap (\garg{\nbranchA{N}} \cup
\garg{\DEFP{F}}) = \emptyset$. Hence $t \notin \rng{\stx}$, contradicting
our initial presumption about~$t$. Thus $t$ must occur in $\garg{\lnot R_{\LL
    p}\sti}$ at an ``output position'' of~$R_{\LL p}$.

\casepara{\waicond~\ref{cond-wai-const}}
\prlReset{case-all-const}
\centerlong{
\begin{arrayprflong}
\prl{1} & \garg{\aipol{N_3}} \subseteq \garg{\DEFP{F} \land
    \nbranchA{N_3}}.\\
\prl{2} & \garg{\aipol{N_3}} \subseteq \garg{\DEFP{F} \land
    \nbranchA{N} \land D_{\LL p1}\sti}.\\
\prl{3} & \garg{D_{\LL p1}\sti} \subseteq \garg{\lnot D_{\LL p}\sti} \cup \garg{\lnot R_{\LL p}\sti}
\subseteq \garg{\nbranchA{N}} \cup \garg{\lnot R_{\LL p}\sti}.\\
\prl{4} & \garg{\aipol{N_3}} \subseteq 
\garg{\DEFP{F} \land
    \nbranchA{N}} \cup \garg{{\lnot R_{\LL p}\sti}}.\\
\prl{5} & \garg{\aipol{N}} \subseteq \garg{\aipol{N_3}}.\\
\prl{6} & \garg{\aipol{N}} \cap
(\garg{\lnot R_{\LL p}\sti} \setminus 
 (\garg{\DEFP{F} \land \nbranchA{N}})) = \emptyset.\\
\prl{7} & \garg{\aipol{N}} \subseteq \garg{\DEFP{F} \land
    \nbranchA{N}}.\\
\prl{8} & \garg{\aipol{N}} \subseteq \garg{\lnot \DEFPNOT{G} \lor \lnot \nbranchB{N}}.\\
\prl{9} & \garg{\aipol{N}}\; \subseteq\\
    & \garg{\DEFP{F} \land \nbranchA{N}} \cap
      \garg{\lnot (\DEFPNOT{G} \land \nbranchB{N})}.
\end{arrayprflong}}
Subsumption~\pref{1} follows from~(IH). Subsumption~\pref{2} is obtained
from~\pref{1} by expressing $\nbranchA{N_3}$ with its last conjunct made
explicit.  Subsumptions~\pref{3} follow from the definition of clause form~6
and~(LD).  Subsumption~\pref{4} follows from~\pref{2}
and~\pref{3}. Subsumptions~\pref{5} and~\pref{6} follow from the definition of
$\aipol{N}$ in the considered case.  Subsumption~\pref{7} follows
from~\pref{6}, \pref{5} and~\pref{4}.  Subsumption~\pref{8} can be shown as
follows: Let $t$ be a member of $\garg{\aipol{N}}$. From the definition of
$\aipol{N}$ for the considered case it follows that $t \in \garg{\lnot R_{\LL
    p}\sti} \cup \garg{\aipol{N_3}}$.  If $t \in \garg{\lnot R_{\LL p}\sti}$,
then, since $\nlit{N_2} = \lnot R_{\LL p}\sti$ it also holds that $t \in
\nlit{\ntgt{N_2}}$. Thus, because $\nside{\ntgt{N_2}} = \bbb$ it then holds
that $t \in \garg{\nbranchB{N}}$.  If $t \in \garg{\aipol{N_3}}$, then,
from~(IH) it follows that $t \in \garg{\lnot \DEFPNOT{G} \lor \lnot
  \nbranchB{N_3}}$, hence, by~(LR), $t \in \garg{\lnot \DEFPNOT{G} \lor \lnot
  \nbranchB{N}}$, which completes the proof of~\pref{8}.
Subsumption~\pref{9} follows from~\pref{8} and~\pref{7}.
\qed
\end{caselist}

\vspace{-10pt} 
\casesection{Clause Form~7}

For this clause forms $\aipol{N}$ is defined as $\bigvee_{i=2}^{k} \aipol{N_i}
= \aipol{N_2}$.  Let $\sti$ be a ground substitution such that $\dom{\sti} =
\parg{\lnot D_{\LL p}}$ and $\lit{N_1} = \lnot D_{\LL p}\sti$.

\begin{caselist}
  \casepara{\waicond~\ref{cond-wai-sem} left}
\prlReset{case-or7-sem-left}%
\[
\begin{arrayprf}
\prl{1} & \DEFP{F} \land \nbranchA{N_2} \entails \aipol{N_2}.\\
\prl{2} & \DEFP{F} \land \nbranchA{N} 
\land R_{\LL p}\stcr_{\LL p}\sti \entails \aipol{N_2}.\\
\prl{gamma} & \gamma \eqdef (\stcr_{\LL p}\sti)|_{\vout_{\LL p}}.\\
\prl{3} & \DEFP{F} \land \nbranchA{N^{\prime\prime}} \land R_{\LL p}\stcr_{\LL p}\sti
  \land D_{\LL p1}\stcr_{\LL p}\sti \entails \aipol{N_2}.\\
\prl{5} & \rng{\gamma}
\cap \garg{\nbranchA{N^{\prime\prime}}} = \emptyset.\\
\prl{6} & \rng{\gamma}
\cap \garg{\nbranchB{N_2}} = \emptyset.\\
\prl{7} & \rng{\gamma} \cap \garg{\DEFP{F}} = \emptyset.\\
\prl{8} & \rng{\gamma} \cap \garg{\DEFPNOT{G}} = \emptyset.\\
\prl{9} & \rng{\gamma} \cap \garg{\aipol{N_2}} = \emptyset.\\
\prl{10} & \exists \vout_{\LL p} 
\toprevsubst{(\DEFP{F} \land \nbranchA{N^{\prime\prime}} \land 
  R_{\LL p}\stcr_{\LL p}\sti \land D_{\LL p1}\stcr_{\LL p}\sti)}{\gamma} \entails \aipol{N_2}.\\
\prl{11} & \DEFP{F} \land \nbranchA{N^{\prime\prime}} \land \exists \vout_{\LL p}
  \toprevsubst{(R_{\LL p}\stcr_{\LL p}\sti \land D_{\LL p1}\stcr_{\LL p}\sti)}{\gamma}
  \entails \aipol{N_2}.\\
\prl{12} & \DEFP{F} \land \nbranchA{N} \land \exists \vout_{\LL p}
  \toprevsubst{(R_{\LL p}\stcr_{\LL p}\sti \land D_{\LL p1}\stcr_{\LL p}\sti)}{\gamma}
   \entails \aipol{N_2}.\\
\prl{13} & \DEFP{F} \land \nbranchA{N} \land \exists \vout_{\LL p}
   (R_{\LL p}\sti \land D_{\LL p1}\sti) \entails \aipol{N_2}.\\
\prl{14} & \DEFP{F} \entails \forall \dx_{\LL p} (D_{\LL p} \imp \exists \vout_{\LL p}\, (R_{\LL p}
  \land D_{\LL p1})).\\
\prl{15} & \DEFP{F} \entails D_{\LL p}\sti \imp \exists \vout_{\LL p}\, (R_{\LL p}\sti
  \land D_{\LL p1}\sti).\\
\prl{16} & \DEFP{F} \land \nbranchA{N} \entails D_{\LL p}\sti.\\
\prl{17} & \DEFP{F} \land \nbranchA{N} \entails \exists \vout_{\LL p}\, (R_{\LL p}\sti
  \land D_{\LL p1}\sti).\\
\prl{18} & \DEFP{F} \land \nbranchA{N} \entails \aipol{N_2}.\\
\prl{19} & \DEFP{F} \land \nbranchA{N} \entails \aipol{N}.\\
\end{arrayprf}
\]
Entailment~\pref{1} follows from~(IH).  Entailment~\pref{2} is obtained
from~\pref{1} by expressing $\nbranchA{N_2}$ with its last conjunct made
explicit.  The substitution~$\gamma$, defined in \pref{gamma}, is an injection
and $R_{\LL p}\stcr_{\LL p}\sti$ as well as $D_{\LL p1}\stcr_{\LL p}\sti$ are
the introducer literals for exactly the members of $\rng{\gamma}$.  From the
contiguity property of the tableau it follows that if a node with literal
label $D_{\LL p1}\stcr_{\LL p}\sti$ is an ancestor of~$N_2$, then it is the
parent of $N_2$, that is, $N$.  Hence there exists node $N^{\prime\prime}$
which is the parent of $N$ or identical with $N$ (depending on whether
$\nlit{N} = D_{\LL p1}\stcr_{\LL p}\sti$) such that~\pref{3} holds and, by
Proposition~\ref{prop-aux-skolem-strong}, also~\pref{5} and~\pref{6} hold.
(Entailment~\pref{3} holds also if $\nlit{N} \neq D_{\LL p1}\stcr_{\LL
  p}\sti$.  The conjunct $D_{\LL p1}\stcr_{\LL p}\sti$ then just redundantly
strengthens the left side.)  Equalities~\pref{7} and~\pref{8} hold since the
principal functor of all members of $\rng{\gamma}$ is a Skolem functor and
thus does occur neither in $\DEFP{F}$ nor in $\DEFPNOT{G}$.
Equality~\pref{9} follows from~\pref{8}, \pref{6} and~(IH).
Entailment~\pref{10} follows from \pref{9} and~\pref{3} by
Proposition~\ref{prop-inessential-qu-terms}.  Entailment~\pref{11} follows
from~\pref{10}, \pref{7} and~\pref{5}.  Entailment~\pref{12} follows
from~\pref{11} since $N^{\prime\prime}$ is identical to $N$ or the parent of
$N$.  Entailment~\pref{13} follows from~\pref{12} since $\toprevsubst{R_{\LL
    p}\stcr_{\LL p}\sti}{\gamma} = R_{\LL p}\sti$ and $\toprevsubst{D_{\LL
    p1}\stcr_{\LL p}\sti}{\gamma} = D_{\LL p1}\sti$.  Entailment~\pref{14}
holds since its right side is a conjunct of its left side.
Entailment~\pref{15} follows from~\pref{14} by instantiating universally
quantified variables. Entailment~\pref{16} follows from~(LD).
Entailment~\pref{17} follows from~\pref{16} and~\pref{14}.
Entailment~\pref{18} follows from~\pref{17}
and~\pref{13}. Entailment~\pref{19} follows from~\pref{18} since $\aipol{N} =
\aipol{N_2}$.

\casepara{\waicond~\ref{cond-wai-sem} right}
Immediate from (IH) and (LR) since $\aipol{N} = \aipol{N_2}$.

\casepara{\waicond~\ref{cond-wai-lyndon} and~\ref{cond-wai-relat}}
Immediate from (IH) since $\aipol{N} = \aipol{N_2}$.

\casepara{\waicond~\ref{cond-wai-const}}
\prlReset{case-or7-sem-const}%
\[
\begin{arrayprf}
\prl{1} & t \in \garg{\aipol{N_2}} \cap \garg{R_{\LL p}\stcr_{\LL p}\sti}.\\
\prl{2} & t \notin \rng{\gamma}.\\
\prl{3} & t \in \garg{\lnot D_{\LL p}\sti} \cup \garg{R_{\LL p}}.\\
\prl{4} & t \in \garg{\nbranchA{N}} \cup \garg{R_{\LL p}}.\\
\prl{5} & \garg{\aipol{N_2}} \cap \garg{R_{\LL p}\stcr_{\LL p}\sti} 
\subseteq \garg{\nbranchA{N}} \cup \garg{R_{\LL p}}.\\
\prl{6} & \garg{\aipol{N_2}} \subseteq \garg{\DEFP{F} \land
  \nbranchA{N}}.\\
\prl{7} & \garg{\aipol{N_2}} \subseteq \garg{\lnot \DEFPNOT{G} \lor
  \lnot \nbranchB{N}}.\\
\prl{8} & \garg{\aipol{N}}\; \subseteq\\
    & \garg{\DEFP{F} \land
  \nbranchA{N}} \cap \garg{\lnot \DEFPNOT{G} \lor \lnot \nbranchB{N}}.
\end{arrayprf}
\]
Let $t$ be a term that satisfies \pref{1}. Let substitution~$\gamma$ be
defined as in step~\prefglobal{case-or7-sem-left:gamma} of the proof of
\name{\waicond~\ref{cond-wai-sem} left} above.  Statement~\pref{2} then
follows from step~\prefglobal{case-or7-sem-left::9} of that proof. By
\pref{2}, the literal $R_{\LL p}\stcr_{\LL p}\sti$ is not an introducer
literal for~$t$. Hence~\pref{3} follows from the definition of clause
form~7. Statement~\pref{4} follows from~\pref{3} and~(LD).
Subsumption~\pref{5} then follows since \pref{1} implies~\pref{4}, for all
ground terms~$t$.  Subsumption~\pref{6} follows from~(IH) and~\pref{5} because
$\garg{R_{\LL p}} \subseteq \garg{\DEFP{F}}$ and $\nbranchA{N_2} =
\nbranchA{N} \land R_{\LL p}\stcr_{\LL p}\sti$.  Subsumption~\pref{7} follows
from~(IH) and~(LR).  Since $\aipol{N} = \aipol{N_2}$, subsumption~\pref{8}
follows from~\pref{7} and~\pref{6}.

\end{caselist}

\casesection{Clause Form~8} Can by show in the same way as for clause form~7,
with the roles of $R_{\LL p}\stcr_{\LL p}\sti$ and $D_{\LL p1}\stcr_{\LL
  p}\sti$ switched. \qed
\end{proof}

Theorem~\ref{thm-correct-aipol} stated above
can now be proven on the basis of Lemma~\ref{lem-aipol-invariant}:

\begin{thmref}
  {Correctness of Access Interpolant Extraction from an \ACIT}
  {\ref{thm-correct-aipol}}
  Let $F, G$ be \RQFO sentences such that $F \entails G$ and let $N$ be the
  root of an \acit for $F$ and~$G$. Then $\aipol{N}$ is an access interpolant
  of $F$ and $G$.
\end{thmref}
\begin{proof}
From Lemma~\ref{lem-aipol-invariant} and Proposition~\ref{prop-weak-std-ai} it follows
that $\aipol{N}$ is an access interpolant of $\DEFP{F}$ and $\lnot
\DEFPNOT{G}$.  With Proposition~\ref{prop-fhg-normalized} if follows that
$\aipol{N}$ is an access interpolant of $F$ and $G$.  \qed
\end{proof}

\section{Ensuring the Requirements on \ACITX}
\label{sec-access-convert}

An \acit has clauses of specific forms according to
Lemma~\ref{lem-input-clauseforms} and certain structural properties, namely,
it is closed, regular, leaf-only for the set of all negative literals
occurring as literal labels, and contiguous for certain pairs of literals.  A
closed positive hyper tableau has all these structural properties, with
exception of the contiguity requirement.  Hence, a closed positive hyper
tableau whose clauses match the forms of Lemma~\ref{lem-input-clauseforms}
that also satisfies the required contiguity property can be directly used to
extract an access interpolant.  Actually, contiguity can in this case be
ensured with an inexpensive tableau transformation, shown as
Procedure~\ref{proc-contig} below.

Closed clausal tableaux with arbitrary structure can be restructured to meet
the structural properties required by \acitx with a series of tableau
conversions that we will now specify.  All of them preserve closedness and for
all of them the clauses of the converted tableau are clauses of the respective
input tableau. With exception of the conversion that ensures the
\name{leaf-only} property all considered tableau conversions are
\emph{simplifications}, that is, procedures that require typically linear and
at most polynomial effort.  Termination is for these conversions easy to see.
For the potentially expensive \name{leaf-only} conversion we state it
explicitly as a proposition and provide a proof.  Examples that illustrate the
conversions will be given in Sect.~\ref{sec-access-convert-examples}.  First
we need to specify an additional auxiliary tableau property:
\begin{defn}[Eager]
\label{def-eager}  
A clausal tableaux is called \defname{eager} if and only if no closed node is
a descendant of another closed node.
\end{defn}
The \name{eager} property is typically ensured implicitly by tableau
construction calculi, since there it is pointless to attach children to a
closed node.  In addition, the leaf-only property for the set of all negative
literals, which is presupposed for positive hyper tableaux and \acit, implies
eagerness.
Operations such as instantiating literal labels and tableau structure
transformations as considered here might, however, result in non-eager
tableaux, also for eager inputs, such that it is useful to take the this
property here explicitly into account.

The following conversions to ensure eagerness and regularity are described as
destructive tableau manipulation procedures. The procedure for ensuring
regularity is from \cite{letz:habil} and is illustrated by
Fig.~\ref{fig-removal-irregularities}.  Both conversions can be considered as
tableau simplifications.
\begin{proc}[Removal of Uneagerness]
\label{proc-simp-eager}
\ 

\algoinput A clausal tableau.

\algoskip
\algomethod Repeat the following operation until the resulting tableau is
eager: Select an inner node $N$ that is closed.  Remove the edges originating
in $N$.

\algoskip

\algooutput An eager clausal tableau, whose clauses are also clauses of the
input tableau.  The following properties of the input tableau are preserved:
closed, regular, leaf-only.
\end{proc}

\begin{proc}[Removal of Irregularities {\cite[Section~2.1.3]{letz:habil}}]
\label{proc-simp-regular}
\ 

\algoinput A clausal tableau.

\algoskip

\algomethod Repeat the following operation until the resulting tableau is
regular: Select a node~$N$ in the tableau with an ancestor $N^\prime$ such
that $\nlit{N^\prime} = \nlit{N}$. Remove the edges originating in the parent
$N^{\prime\prime}$ of $N$ and replace them with the edges originating in $N$.

\algoskip

\algooutput A regular clausal tableau whose clauses are also clauses of the
input tableau. The following properties of the
input tableau are preserved: closed, eager, leaf-only.
\end{proc}

\renewcommand{\tableauscale}{0.7}

\begin{figure}
\centering
\caption{Tableau simplification step for removal of irregularities with
  Procedure~\ref{proc-simp-regular} \cite{letz:habil}.  Node $N$ in the procedure
  description corresponds to $N_i$, in the figure, where $1 \leq i \leq k$. It
  is possible that $N^{\prime\prime} = N^\prime$.  A triangle below a node
  represents the edges originating in the node (which might be none) together
  with the descendants of the node and all edges between them.}
\label{fig-removal-irregularities}
\begin{tikzpicture}[scale=\tableauscale,
    baseline=(a.north),
    sibling distance=5em,level distance=10ex,
    every node/.style = {transform shape,anchor=mid}]]
    \node (a) {$N^\prime$}
    child { node {$N^{\prime\prime}$} 
      edge from parent[dotted,thick]
      [sibling distance=4.0em,level distance=10ex]
      child { node (3) {$N_1$} \tria{$T_1$} edge from parent[solid,thin] }
      child[missing] { node {$N_1$} }
      child { node (4) {$N_{i-1}$}  \tria{$T_{i-1}$} edge from parent[solid,thin] }
      child { node {$N_i$} \tria{$T_i$} edge from parent[solid,thin] }
      child { node (5) {$N_{i+1}$} \tria{$T_{i+1}$} edge from parent[solid,thin] }
      child[missing] { node {$N_{i+1}$}  }
      child { node (6) {$N_{k}$} \tria{$T_k$} edge from parent[solid,thin] }
    };
    \path (3) -- node[auto=false]{\ldots} (4);
    \path (5) -- node[auto=false]{\ldots} (6);
\end{tikzpicture}
\raisebox{-13ex}{$\;\;\;\rewrite\;\;\;$}
\begin{tikzpicture}[scale=\tableauscale,
    baseline=(a.north),
    sibling distance=5em,level distance=10ex,
    every node/.style = {transform shape,anchor=mid}]]
    \node (a) {$N^\prime$}
    child { node {$N^{\prime\prime}$} 
      \tria{$T_i$} 
      edge from parent[dotted,thick] };
\end{tikzpicture}
\end{figure}
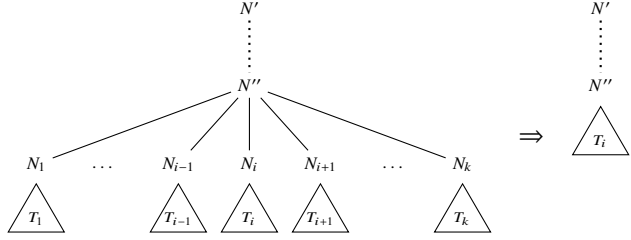

The following conversion ensures the \name{leaf-only} property.  It is again
specified as a procedure that destructively manipulates a tableau.  We use
there the notion of a \defname{fresh copy} of an ordered tree~$T$, which is an
ordered tree~$T^\prime$ with fresh nodes and edges, related to $T$ through a
bijection $c$ such that any node $N$ of $T$ has the same labels (e.g., literal
label and side label) as node $c(N)$ of $T^\prime$ and such that the $i$-th
edge originating in node $N$ of $T$ ends in node~$M$ if and only if the $i$-th
edge originating in node $c(N)$ of $T^\prime$ ends in node $c(M)$.  The
procedure is illustrated by Fig.~\ref{fig-leafonly}
and~\ref{fig-leafonly-fixing}. Its termination is then shown with
Proposition~\ref{proc-leafonly-terminates}.

\begin{proc}[Leaf-Only Conversion]
\label{proc-leafonly}

\algoinput A closed, eager and regular clausal tableau and a set~$S$ of pairwise
non-comple\-mentary literals that occur as literal labels of nodes of the
tableau.

\algoskip

\algomethod Repeat the following operations until the tableau is leaf-only
for~$S$:
\begin{enumerate}
\item \label{step-proc-pick} Let $N$ be the inner node whose literal label is
  in $S$ that is first visited by traversing the tableau in pre-order.  Let
  $N^\prime$ be the parent of $N$.

\item Create a fresh copy~$U$ of the subtree rooted at $N^\prime$.  In~$U$
  remove the edges that originate in the node corresponding to $N$.

\item \label{step-proc-attach} Remove the edges originating in $N^\prime$ and
  replace them with the edges originating in $N$.

\item \label{step-proc-fix} For each leaf descendant~$M$ of $N^\prime$ with
  $\nlit{M} = \du{\nlit{N}}$: Create a fresh copy~$U^\prime$ of $U$. Change
  the origin of the edges originating in the root of $U^\prime$ to~$M$.

\item \label{step-proc-simp} Ensure eagerness and regularity by simplifying
  with Procedure~\ref{proc-simp-eager} and \ref{proc-simp-regular}.
\end{enumerate}

\algooutput A closed, eager and regular clausal tableau whose clauses are also
clauses of the input tableau and which is leaf-only for $S$.
\end{proc}

\begin{figure}[b]
\centering
\caption{Conversion step for ensuring the leaf-only property with
  Procedure~\ref{proc-leafonly}.  Node~$N$ in the procedure description
  corresponds to $N_i$ in the figure, where $1 \leq i \leq k$.  A triangle
  below a node represents the edges originating in the node (which might be
  none, except for $T_i$) together with the descendants of the node and all
  edges between them. Triangle~$T_i^\prime$ is obtained from $T_i$ with steps
  illustrated in Fig.~\ref{fig-leafonly-fixing}.}
\label{fig-leafonly}

\begin{tikzpicture}[scale=\tableauscale,
    baseline=(a.north),
    sibling distance=5em,level distance=10ex,
    every node/.style = {transform shape,anchor=mid}]]
    \node (a) {$N^\prime$}
          [sibling distance=4.0em,level distance=10ex]
          child { node (3) {$N_1$} \tria{$T_1$} }
          child[missing] { node {$N_1$}  }
          child { node (4) {$N_{i-1}$}  \tria{$T_{i-1}$} }
          child { node {$N_i$} \tria{$T_i$} }
          child { node (5) {$N_{i+1}$} \tria{$T_{i+1}$} }
          child[missing] { node {$N_{i+1}$}  }
          child { node (6) {$N_{k}$} \tria{$T_k$} };
    \path (3) -- node[auto=false]{\ldots} (4);
    \path (5) -- node[auto=false]{\ldots} (6);
\end{tikzpicture}
\raisebox{-7ex}{$\;\;\rewrite\;\;$}
\begin{tikzpicture}[scale=\tableauscale,
    baseline=(a.north),
    sibling distance=5em,level distance=10ex,
    every node/.style = {transform shape,anchor=mid}]]
    \node (a) {$N^\prime$} \tria{$T_i^\prime$};
\end{tikzpicture}
\end{figure}
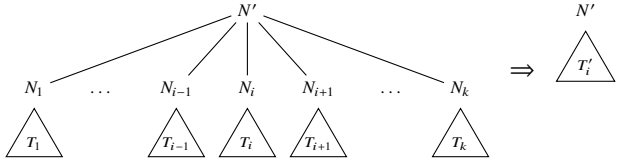

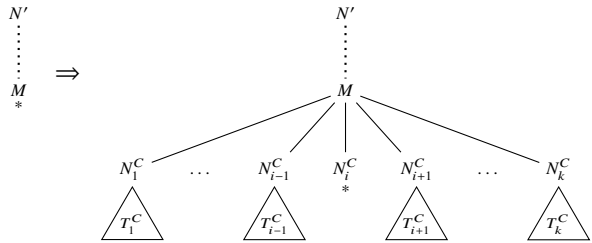
\begin{figure}
\centering
\caption{Conversion step of a leaf $M$ with $\nlit{M} = \du{\nlit{N}}$ in
  Procedure~\ref{proc-leafonly} to obtain $T_i^\prime$ from $T_i$ (see
  Fig.~\ref{fig-leafonly}).  An asterisk indicates leaves for which it is
  ensured that they are closed.  Node $N$ in the procedure description
  corresponds to $N_i$ here.  Superscripts $C$ indicate that copies of the
  subtrees referenced in Fig.~\ref{fig-leafonly} are used.  Node $N^\prime$,
  depicted also in Fig.~\ref{fig-leafonly}, is included here just to indicate
  explicitly that all affected nodes $M$ are descendants of $N^\prime$.}
\label{fig-leafonly-fixing}

\begin{tikzpicture}[scale=\tableauscale,
    baseline=(a.north),
    sibling distance=5em,level distance=10ex,
    every node/.style = {transform shape,anchor=mid}]]
    \node (a) {$N^\prime$}
    child { \closednode {$M$} 
    edge from parent[dotted,thick]};
\end{tikzpicture}
\raisebox{-7ex}{$\;\;\rewrite\;\;$}
\begin{tikzpicture}[scale=\tableauscale,
    baseline=(a.north),
    sibling distance=5em,level distance=10ex,
    every node/.style = {transform shape,anchor=mid}]]
    \node (a) {$N^\prime$}
    child { node {$M$}
      [sibling distance=4.0em,level distance=10ex]
      edge from parent[dotted,thick]
      child { node (3) {$N_1^C$} \tria{$T_1^C$} edge from parent[solid,thin] }
      child[missing] { node {$N_1$}  }
      child { node (4) {$N_{i-1}^C$}  \tria{$T_{i-1}^C$} edge from parent[solid,thin] }
      child { \closednode {$N_i^C$} edge from parent[solid,thin] }
      child { node (5) {$N_{i+1}^C$} \tria{$T_{i+1}^C$} edge from parent[solid,thin] }
      child[missing] { node {$N_{i+1}^C$}  }
      child { node (6) {$N_{k}^C$} \tria{$T_k^C$} edge from parent[solid,thin] }
    };
    \path (3) -- node[auto=false]{\ldots} (4);
    \path (5) -- node[auto=false]{\ldots} (6);
\end{tikzpicture}
\end{figure}

\begin{prop}[Termination of Leaf-Only Conversion]
\label{proc-leafonly-terminates}
Procedure~\ref{proc-leafonly} terminates.
\end{prop}

\begin{proof}
We give a measure that strictly decreases in each round of the procedure.
Consider a single round of the steps~1.--5. of Procedure~\ref{proc-leafonly} with
$N$ and $N^\prime$ as determined in step~1. Then:
\begin{enumerate}[label={\roman*}.,leftmargin=2em]
\item \label{item-proc-leafonly-allbelow} All tableau modifications made in
  the round are in the subtree rooted at~$N^\prime$.
\item \label{item-proc-leafonly-toleaf} At finishing the round all descendants
  of $N^\prime$ with the same literal label as~$N$ are leaves.
\item \label{item-proc-leafonly-nonewnonleaf} All literal labels of inner
  nodes that are descendants of $N^\prime$ and are different from
  $\du{\nlit{N}}$ at finishing the round are already literal labels of inner
  nodes that are descendants of~$N^\prime$ when entering the round.
\end{enumerate}
We can now specify the measure that strictly decreases in each round of
Procedure~\ref{proc-leafonly}. For a node $N$ define $\nbadlits{N}$ as the set
of literal labels that occur in inner (i.e., non-leaf) descendants of~$N$ and
are members of~$S$.  From the above items~(\ref{item-proc-leafonly-toleaf})
and~(\ref{item-proc-leafonly-nonewnonleaf}) it follows that for $N^\prime$ as
determined in step~1 of Procedure~\ref{proc-leafonly} the cardinality of
$\nbadlits{N^\prime}$ is strictly decreased in a round of steps~1.--5. of the
procedure.  However, a different node might be determined as~$N^\prime$ in
step~1 of the next round.  To specify a globally decreasing measure we define
a further auxiliary notion: Let $N_n$ be a node whose ancestors are in
root-to-leaf order the nodes $N_1,\ldots,N_{n-1}$. Define $\ncode{N_n}$ as the
string $I_1\ldots I_n \omega$ of numbers, where for $i \in \{1,\ldots,n\}$ the
number $I_i$ is the number of right siblings of $N_i$.  With
item~(\ref{item-proc-leafonly-allbelow}) it then follows that the following
string of numbers, determined at step~1 of a round, is strictly reduced from
round to round w.r.t. the lexicographical order of strings of numbers:
 \[\ncode{N^\prime}|\nbadlits{N^\prime}|.\]
Regularity ensures that the length of the strings to be considered can not be
larger than the finite number of literal labels of nodes of the input tableau
plus~$3$ (a leading $0$ for the root, which has no literal label; $\omega$;
and $|\nbadlits{N^\prime}|$).  With the lexicographical order restricted to
strings up to that length we have a well-order and the strict reduction
ensures termination.
\qed
\end{proof}

The following conversion ensures contiguity as far as required for \acitx.  It
is illustrated by Fig.~\ref{fig-simp-contig}.

\begin{proc}[Ensuring Contiguity in Special Cases]
\label{proc-contig}
\ 

\algoinput An eager and regular clausal tableau and a set $\SetOfS$ of
unordered pairs of literals such that for each such pair $\{L_1, L_2\}$ it
holds that:
\begin{itemize}
\item $L_1$ and $L_2$ occur as literal labels of nodes of the tableau.
\item There is a literal $L_0$ such that all clauses of the tableau in which
  $L_1$ or $L_2$ occur as literals are of the form $L_0 \lor L_1$ or $L_0 \lor
  L_2$.
\item All nodes of the tableau with $L_0$ as literal label are leaves.
\end{itemize}

\algomethod Repeat the following until the resulting tableau is contiguous for
all members of~$\SetOfS$:
\begin{enumerate}
\item \label{proc-contig-step-select} Select an inner node $N$ that has a
  descendant $M$ such that $\{\nlit{N}, \nlit{M}\} \in \SetOfS$ and there is a
  third node that is a descendant of $N$ and an ancestor of $M$.
\item Create fresh nodes $M_0^\prime$ and $M^\prime$ where $M_0^\prime$ has
  the same label values (i.e., the literal label and, if applicable, the side
  label) as the sibling of $M$, and $M^\prime$ has the same label values as
  $M$.
\item Remove the outgoing edges from $N$ and attach them to $M^\prime$.
\item \label{proc-contig-step-final-trafo} Add $M_0^\prime$ and $M^\prime$ as
  children to $N$.
\item Apply Procedure~\ref{proc-simp-regular} to ensure regularity.
\end{enumerate}

\algooutput An eager and regular clausal tableau whose clauses are also
clauses of the input tableau and which is contiguous for all members of the
input set~$\SetOfS$. The following further properties of the input tableau are
preserved: closed, leaf-only for a set of literals that does not contain
members of the pairs in $\SetOfS$.
\end{proc}

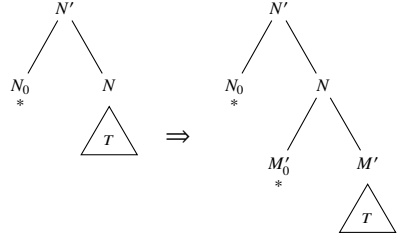
\begin{figure}
\centering
\caption{First steps of a round for establishing contiguity with
  Procedure~\ref{proc-contig}. Node~$M$ and the third node between $N$ and $M$
  mentioned in the procedure definition are descendants of $N$ and not shown
  explicitly.  Nodes $N_0$ and $M_0^\prime$ have the same labels. The depicted
  conversion in
  steps~\ref{proc-contig-step-select}--\ref{proc-contig-step-final-trafo} is
  followed by the application of Procedure~\ref{proc-simp-regular} to re-establish
  regularity.}
\label{fig-simp-contig}
\begin{tikzpicture}[scale=\tableauscale,
    baseline=(a.north),
    sibling distance=5em,level distance=10ex,
    every node/.style = {transform shape,anchor=mid}]]
    \node (a) {$N^{\prime}$}
    child { \closednode {$N_0$} } 
    child { node {$N$} \tria{$T$} };
\end{tikzpicture}
\raisebox{-13ex}{$\;\;\;\rewrite\;\;\;$}
\begin{tikzpicture}[scale=\tableauscale,
    baseline=(a.north),
    sibling distance=5em,level distance=10ex,
    every node/.style = {transform shape,anchor=mid}]]
    \node (a) {$N^{\prime}$}
    child { \closednode {$N_0$} } 
    child { node {$N$} 
      child { \closednode {$M_0^\prime$} }
      child { node {$M^\prime$} \tria{$T$} }
    };
\end{tikzpicture}
\end{figure}

Termination of Procedure~\ref{proc-contig} follows since the number of nodes
that can be selected in step~\ref{proc-contig-step-select} strictly decreases
in each round.  Like Procedure~\ref{proc-simp-eager} and
\ref{proc-simp-regular}, the procedure can be considered as a tableau
simplification.

The procedures defined in this section suggest to apply them in the presented
order, that is, Procedure~\ref{proc-simp-eager} (eagerness),
Procedure~\ref{proc-simp-regular} (regularity), Procedure~\ref{proc-leafonly}
(leaf-only property) and Procedure~\ref{proc-contig} (contiguity) to the
closed clausal tableau obtained by a prover from the structure preserving
clausifications of $\DEFP{F}$ and $\DEFPNOT{G}$. The converted tableau is
then an \acit, suitable for extracting the access interpolant according to
Definition~\ref{def-aipol}. If the closed clausal tableau obtained by the
prover is already a positive hyper tableau, then it is, of course, sufficient
ensure contiguity with Procedure~\ref{proc-contig}.

\section{Examples for Conversion to \ACITX}
\label{sec-access-convert-examples}

In this section the definitional normalization of \RQFO formulas for access
interpolation and the conversion of closed clausal tableaux for them to \acitx
is illustrated with examples. We consider computing an access interpolant for
the single \RQFO sentence
\begin{equation}
\label{eq-examp-fg}
F \;=\; G \;=\;\forall x\, (\lnot \rk(x) \lor \exists y\, (\sk(x,y) \land \true))
\end{equation}
in the role of both interpolation inputs.  Of course, the sentence itself is
then trivially also an access interpolant. Nevertheless, with this example
different structuring possibilities of clausal tableaux as obtained by provers
and the effects of the conversions show up.

The definitional normal forms of $\DEFP{F}$ and $\DEFPNOT{G}$, conjoined
together, yield the following clausal formula, where $\ffk$ and $\ggk$ are
Skolem functions, as basis for interpolant computation.  The respective clause
form according to Lemma~\ref{lem-input-clauseforms} is there annotated in the
right column.  Clauses obtained from $\DEFP{F}$ are shown against
\taaatxt{grey background.}
\begin{equation}
\label{eq-formula-examp-xy}
\begin{array}{llllc}
\rowcolor{tcolaaabg}
\taaa{\dk_{\LL \emptyseq}} && \taaa{\land} & \hspace*{3em} & 1\\
\rowcolor{tcolaaabg}
\taaa{(\lnot \dk_{\LL \emptyseq} \lor \lnot \rk(x) \lor \dk_{\LL 1}(x))} && \taaa{\land} && 5\\
\rowcolor{tcolaaabg}
\taaa{(\lnot \dk_{\LL 1}(x) \lor \sk(x,\ffk(x)))} && \taaa{\land} && 7\\
\tbbb{\dk_{\RR \emptyseq}} && \tbbb{\land} && 9 \\
\tbbb{(\lnot \dk_{\RR \emptyseq} \lor \rk(\ggk))} && \tbbb{\land} && 14\\
\tbbb{(\lnot \dk_{\RR \emptyseq} \lor \dk_{\RR 1}(\ggk))} && \tbbb{\land} && 15\\
\tbbb{(\lnot \dk_{\RR 1}(x) \lor \lnot \sk(x,y) \lor \dk_{\RR 11})} && \tbbb{\land} && 11\\
\tbbb{\lnot \dk_{\RR 11}} &&&& 10\\
\end{array}
\end{equation}
Using global position specifiers (Definition~\ref{def-globalpos}), the value
of some of the symbolic designators in Lemma~\ref{lem-input-clauseforms} is as
follows: $R_{\LL \emptyset} = R_{\RR \emptyset} = \rk$, $R_{\LL 1} = R_{\RR 1}
= \sk$, $\sigma_{\LL 1} = \{y \mapsto \skf(x)\}$, $\sigma_{\RR 0} = \{x
\mapsto \skg\}$.  The Skolem functions $\f{f}_{\la \LL 1, 1\ra}$ and
$\f{f}_{\la \RR \emptyset, 1\ra}$ are expressed by $\ffk$ and $\ggk$,
respectively, for readability.

In the examples shown below we will consider closed clausal tableaux for the
clausal formula~(\ref{eq-formula-examp-xy}), where the tableau clauses are the
following instances of the clauses of
formula~(\ref{eq-formula-examp-xy}). Again the respective clause form
according to Lemma~\ref{lem-input-clauseforms} is annotated in the right
column.
\begin{equation}
\label{eq-clauses-examp}
\begin{array}{llc}
\rowcolor{tcolaaabg}
\taaa{\dk_{\LL \emptyseq}} & \hspace*{3em} & 1\\
\rowcolor{tcolaaabg}
\taaa{\lnot \dk_{\LL \emptyseq} \lor \lnot \rk(\ggk) \lor \dk_{\LL 1}(\ggk)} && 5\\
\rowcolor{tcolaaabg}
\taaa{\lnot \dk_{\LL 1}(\ggk) \lor \sk(\ggk,\ffk(\ggk))} && 7\\
\tbbb{\dk_{\RR \emptyseq}} && 9 \\
\tbbb{\lnot \dk_{\RR \emptyseq} \lor \rk(\ggk)} && 14\\
\tbbb{\lnot \dk_{\RR \emptyseq} \lor \dk_{\RR 1}(\ggk)} && 15\\
\tbbb{\lnot \dk_{\RR 1}(\ggk) \lor \lnot \sk(\ggk,\ffk(\ggk)) \lor \dk_{\RR 11}} && 11\\
\tbbb{\lnot \dk_{\RR 11}} && 10\\
\end{array}
\end{equation}

\noindent
To qualify as \acit the ground tableau then has to be leaf-only for the set
\begin{equation}
\label{eq-examp-leafonly-set}
\{\taaa{\lnot \dk_{\LL \emptyseq},\, \lnot \dk_{\LL 1}(\ggk),\,
\lnot \rk(\ggk)},\, \tbbb{\lnot \dk_{\RR \emptyseq},\, \lnot \dk_{\RR 1}(\ggk),\, \lnot \dk_{\RR 11},\, \lnot \sk(\ggk,\ffk(\ggk))}\}
\end{equation}
and contiguous for the pair
\begin{equation}
\label{eq-examp-contiguous-set}
\{\tbbb{\rk(\ggk),\, \dk_{\RR 1}(\ggk)}\}.
\end{equation}

As noted in Sect.~\ref{sec-access-convert}, positive hyper tableaux which
satisfy a certain contiguity condition are already \acitx.  Such a tableau is
typically constructed by ``bottom-up'' calculi that would start with the
positive ``root definers'' $\dk_{\LL \emptyseq}$ and $\dk_{\RR \emptyseq}$ and
proceed by ``applying'' clauses like rules that fire in a forward-chaining
manner, that is, extending a branch only with a clause whose negative literals
all have complements in the branch.  The following tableau gives an example:
\begin{examp}[Positive Hyper Tableau]
\label{examp-hyper}
Figure~\ref{fig-tab-ex-4} shows a closed positive hyper tableau for the
clausal formula~(\ref{eq-formula-examp-xy}) that is an \acit for $F$ and $G$
and thus allows direct extraction of an access interpolant.  Nodes with side
label $\aaa$ are shown with \taaatxt{grey background.}

\medskip

\renewcommand{\extabscale}{0.7}
\renewcommand{\extabld}{9ex}

\begin{tableaufig}{Example~\ref{examp-hyper} -- Positive Hyper Tableau}
\label{fig-tab-ex-4}
\begin{tikzpicture}[scale=\extabscale,
    baseline=(a.north),
    sibling distance=5em,level distance=\extabld,
    every node/.style = {transform shape,anchor=mid}]]
    \node (a) {$\taaa{\dk_{\LL \emptyseq}}$}
    child { node {$\tbbb{\dk_{\RR \emptyseq}}$ }
      child { node {$\tbbb{\lnot \dk_{\RR \emptyseq}}$} }
      child { node {$\tbbb{\rk(\ggk)}$}
        child { node {$\tbbb{\lnot \dk_{\RR \emptyseq}}$} }
        child { node {$\tbbb{\dk_{\RR 1}(\ggk)}$} 
          child { node {$\taaa{\lnot \dk_e}$} }
          child { node {$\taaa{\lnot \rk(\ggk)}$} }
          child { node {$\taaa{\dk_{\LL 1}(\ggk)}$ } 
            child { node {$\taaa{\lnot \dk_{\LL 1}(\ggk)}$} }
            child { node {$\taaa{\sk(\ggk,\ffk(\ggk))}$} 
              child { node {$\tbbb{\lnot \dk_{\RR 1}(\ggk)}$} }
              child { node {$\tbbb{\lnot \sk(\ggk,\ffk(\ggk))}$} }
              child { node {$\tbbb{\dk_{\RR 11}}$} 
                child { node {$\tbbb{\lnot \dk_{\RR 11}}$}}
                }
              }
            }
          }
        }
    };
\end{tikzpicture}
\end{tableaufig}

\end{examp}

The remaining examples shown in this section follow start from ``connection
tableaux'', or, more precisely, \name{tightly connected} tableaux (see, e.g.,
\cite{letz:stenz:handbook}): Each inner node with exception of the root has a
child with complementary literal label.  Such tableaux are constructed from
provers based on model elimination or the connection method, which maintain
the \name{tightly connected} property throughout tableau construction.
Typically they build the tableau ``top-down'' in a goal-sensitive way by
starting in a theorem proving setting with a clause obtained from the theorem
in contrast to the axioms.  This connectedness property of the tableau
returned by provers might get lost by our conversion to \acitx. Moreover,
also the weaker property of \name{path connectedness}, that is, among siblings
(except for the root and its children) there exists a node that has an
\emph{ancestor} with complementary literal label, is not ensured by the
conversions.

\renewcommand{\extabscale}{0.5}
\renewcommand{\extabld}{10ex}

\begin{examp}[Connection Tableau I]
\label{examp-conn-1}
Figure~\ref{fig-tab-ex-1-1} shows a closed tightly connected clausal tableau
for the clausal formula~(\ref{eq-formula-examp-xy}).  Nodes with side label
$\aaa$ are shown with \taaatxt{grey background.} Edges that connect nodes with
complementary literal labels are emphasized.  The node picked as $N$ in the
\emph{next} round of Procedure~\ref{proc-leafonly} is marked by a surrounding
rectangle.  Figure~\ref{fig-tab-ex-1-2} shows the result of applying a round
of Procedure~\ref{proc-leafonly}.  Again the node picked as $N$ in the next
round is marked. Further rounds yield the tableaux of
Fig.~\ref{fig-tab-ex-1-3} and Fig.~\ref{fig-tab-ex-1-4}.  The latter is
leaf-only for the set~(\ref{eq-examp-leafonly-set}) of literals, but not
contiguous for the pair~(\ref{eq-examp-contiguous-set}).  The literals that
are chosen as $N$ and $M$ in Procedure~\ref{proc-contig} are displayed in oval
markings. The result of applying Steps~1.-4. of Procedure~\ref{proc-contig} is
then shown in Fig.~\ref{fig-tab-ex-1-5}.  The tableau now also is contiguous
for the pair~(\ref{eq-examp-contiguous-set}), but violates regularity with the
nodes marked by a flag. The regularity simplification of
Procedure~\ref{proc-simp-regular} finally yields the tableau in
Fig.~\ref{fig-tab-ex-1-6}, which is an \acit and actually identical to the
positive hyper tableau in Fig.~\ref{fig-tab-ex-4}.

\renewcommand{\extabld}{9ex}
\renewcommand{\extabscale}{0.7}

\begin{tableaufigTwoCol}{Example~\ref{examp-conn-1} -- Stage 1}
\label{fig-tab-ex-1-1}
\begin{tikzpicture}[scale=\extabscale,
    baseline=(a.north),
    sibling distance=9em,level distance=\extabld,
    every node/.style = {transform shape,anchor=mid}]]
    \node (a) {$\taaa{\dk_{\LL \emptyseq}}$}
    child { node {$\taaa{\lnot \dk_{\LL \emptyseq}}$} \linked }
    child { \nodeMarkLeafonly {$\taaa{\lnot \rk(\ggk)}$} 
      [sibling distance=5em]
      child { node {$\tbbb{\lnot \dk_{\RR \emptyseq}}$}
        child { node {$\tbbb{\dk_{\RR \emptyseq}}$} \linked }
      }
      child { node {$\tbbb{\rk(\ggk)}$} \linked }
    }
    child { node {$\taaa{\dk_{\LL 1}(\ggk)}$}
      [sibling distance=5em]
      child { node {$\taaa{\lnot \dk_{\LL 1}(\ggk)}$} \linked }
      child { node {$\taaa{\sk(\ggk,\ffk(\ggk))}$} 
        child { node {$\tbbb{\lnot \dk_{\RR 1}(\ggk)}$}
          child { node {$\tbbb{\lnot \dk_{\RR \emptyseq}}$} 
            child { node {$\tbbb{\dk_{\RR \emptyseq}}$} \linked }
          }
          child { node {$\tbbb{\dk_{\RR 1}(\ggk)}$} \linked }
        }
        child { node {$\tbbb{\lnot \sk(\ggk,\ffk(\ggk))}$} \linked }
        child { node {$\tbbb{\dk_{\RR 11}}$} 
          child { node {$\tbbb{\lnot \dk_{\RR 11}}$} \linked }
        }
      }
    };
\end{tikzpicture}
\end{tableaufigTwoCol}
\begin{tableaufigTwoCol}{Example~\ref{examp-conn-1} -- Stage 2}
\label{fig-tab-ex-1-2}
\begin{tikzpicture}[scale=\extabscale,
    baseline=(a.north),
    sibling distance=10em,level distance=\extabld,
    every node/.style = {transform shape,anchor=mid}]]
    \node (a) {$\taaa{\dk_{\LL \emptyseq}}$}
    child { \nodeMarkLeafonly {$\tbbb{\lnot \dk_{\RR \emptyseq}}$}
      child { node {$\tbbb{\dk_{\RR \emptyseq}}$}  \linked }
    }
    child { node {$\tbbb{\rk(\ggk)}$} 
      [sibling distance=5em]
      child { node {$\taaa{\lnot \dk_{\LL \emptyseq}}$} }
      child { node {$\taaa{\lnot \rk(\ggk)}$ } \linked
        [sibling distance=5em]
      }
      child { node {$\taaa{\dk_{\LL 1}(\ggk)}$}
        [sibling distance=5em]
        child { node {$\taaa{\lnot \dk_{\LL 1}(\ggk)}$} \linked }
        child { node {$\taaa{\sk(\ggk,\ffk(\ggk))}$} 
          child { node {$\tbbb{\lnot \dk_{\RR 1}(\ggk)}$}
            child { node {$\tbbb{\lnot \dk_{\RR \emptyseq}}$} 
              child { node {$\tbbb{\dk_{\RR \emptyseq}}$} \linked }
            }
            child { node {$\tbbb{\dk_{\RR 1}(\ggk)}$} \linked }
          }
          child { node {$\tbbb{\lnot \sk(\ggk,\ffk(\ggk))}$} \linked }
          child { node {$\tbbb{\dk_{\RR 11}}$} 
            child { node {$\tbbb{\lnot \dk_{\RR 11}}$} \linked }
          }
        }
      }
    };
\end{tikzpicture}
\end{tableaufigTwoCol}

\begin{tableaufigTwoCol}{Example~\ref{examp-conn-1} -- Stage 3}
\label{fig-tab-ex-1-3}
\begin{tikzpicture}[scale=\extabscale,
    baseline=(a.north),
    sibling distance=5em,level distance=\extabld,
    every node/.style = {transform shape,anchor=mid}]]
    \node (a) {$\taaa{\dk_{\LL \emptyseq}}$}
    child { node {$\tbbb{\dk_{\RR \emptyseq}}$}
      child { node {$\tbbb{\lnot \dk_{\RR \emptyseq}}$} \linked
      }
      child { node {$\tbbb{\rk(\ggk)}$} 
        [sibling distance=5em]
        child { node {$\taaa{\lnot \dk_{\LL \emptyseq}}$} }
        child { node {$\taaa{\lnot \rk(\ggk)}$ } \linked
          [sibling distance=5em]
        }
        child { node {$\taaa{\dk_{\LL 1}(\ggk)}$}
          [sibling distance=5em]
          child { node {$\taaa{\lnot \dk_{\LL 1}(\ggk)}$} \linked }
          child { node {$\taaa{\sk(\ggk,\ffk(\ggk))}$} 
            child { \nodeMarkLeafonly {$\tbbb{\lnot \dk_{\RR 1}(\ggk)}$}
              child { node {$\tbbb{\lnot \dk_{\RR \emptyseq}}$} 
              }
              child { node {$\tbbb{\dk_{\RR 1}(\ggk)}$} \linked }
            }
            child { node {$\tbbb{\lnot \sk(\ggk,\ffk(\ggk))}$} \linked }
            child { node {$\tbbb{\dk_{\RR 11}}$} 
              child { node {$\tbbb{\lnot \dk_{\RR 11}}$} \linked }
            }
          }
        }
      }
    };
\end{tikzpicture}
\end{tableaufigTwoCol}
\begin{tableaufigTwoCol}{Example~\ref{examp-conn-1} -- Stage 4}
\label{fig-tab-ex-1-4}
\begin{tikzpicture}[scale=\extabscale,
    baseline=(a.north),
    sibling distance=5em,level distance=\extabld,
    every node/.style = {transform shape,anchor=mid}]]
    \node (a) {$\taaa{\dk_{\LL \emptyseq}}$}
    child { node {$\tbbb{\dk_{\RR \emptyseq}}$}
      child { node {$\tbbb{\lnot \dk_{\RR \emptyseq}}$} \linked
      }
      child { \nodeMarkContig {$\tbbb{\rk(\ggk)}$} 
        [sibling distance=5em]
        child { node {$\taaa{\lnot \dk_{\LL \emptyseq}}$} }
        child { node {$\taaa{\lnot \rk(\ggk)}$ } \linked
          [sibling distance=5em]
        }
        child { node {$\taaa{\dk_{\LL 1}(\ggk)}$}
          [sibling distance=5em]
          child { node {$\taaa{\lnot \dk_{\LL 1}(\ggk)}$} \linked }
          child { node {$\taaa{\sk(\ggk,\ffk(\ggk))}$} 
            child { node {$\tbbb{\lnot \dk_{\RR \emptyseq}}$} 
            }
            child { \nodeMarkContig {$\tbbb{\dk_{\RR 1}(\ggk)}$} 
              child { node {$\tbbb{\lnot \dk_{\RR 1}(\ggk)}$ } \linked
              }
              child { node {$\tbbb{\lnot \sk(\ggk,\ffk(\ggk))}$}  }
              child { node {$\tbbb{\dk_{\RR 11}}$} 
                child { node {$\tbbb{\lnot \dk_{\RR 11}}$} \linked }
              }
            }
          }
        }
      }
    };
\end{tikzpicture}
\end{tableaufigTwoCol}

\begin{tableaufigTwoCol}{Example~\ref{examp-conn-1} -- Stage 5}
\label{fig-tab-ex-1-5}
\begin{tikzpicture}[scale=\extabscale,
    baseline=(a.north),
    sibling distance=5em,level distance=\extabld,
    every node/.style = {transform shape,anchor=mid}]]
    \node (a) {$\taaa{\dk_{\LL \emptyseq}}$}
    child { node {$\tbbb{\dk_{\RR \emptyseq}}$}
      child { node {$\tbbb{\lnot \dk_{\RR \emptyseq}}$} \linked
      }
      child { node  {$\tbbb{\rk(\ggk)}$} 
        [sibling distance=5em]
        child { node  {$\tbbb{\lnot \dk_{\RR \emptyseq}}$ }}
        child { \nodeMarkRegular {$\tbbb{\dk_{\RR 1}(\ggk)}$} 
          child { node {$\taaa{\lnot \dk_{\LL \emptyseq}}$} }
          child { node {$\taaa{\lnot \rk(\ggk)}$ }
            [sibling distance=5em]
          }
          child { node {$\taaa{\dk_{\LL 1}(\ggk)}$}
            [sibling distance=5em]
            child { node {$\taaa{\lnot \dk_{\LL 1}(\ggk)}$} \linked }
            child { node {$\taaa{\sk(\ggk,\ffk(\ggk))}$} 
              child { node {$\tbbb{\lnot \dk_{\RR \emptyseq}}$} 
              }
              child { \nodeMarkRegular {$\tbbb{\dk_{\RR 1}(\ggk)}$} 
                child { node {$\tbbb{\lnot \dk_{\RR 1}(\ggk)}$ } \linked
                }
                child { node {$\tbbb{\lnot \sk(\ggk,\ffk(\ggk))}$}  }
                child { node {$\tbbb{\dk_{\RR 11}}$} 
                  child { node {$\tbbb{\lnot \dk_{\RR 11}}$} \linked }
                }
              }
            }
          }
        }
      }
    };
\end{tikzpicture}
\end{tableaufigTwoCol}
\begin{tableaufigTwoCol}{Example~\ref{examp-conn-1} -- Stage 6}
\label{fig-tab-ex-1-6}
\begin{tikzpicture}[scale=\extabscale,
    baseline=(a.north),
    sibling distance=5em,level distance=\extabld,
    every node/.style = {transform shape,anchor=mid}]]
    \node (a) {$\taaa{\dk_{\LL \emptyseq}}$}
    child { node {$\tbbb{\dk_{\RR \emptyseq}}$}
      child { node {$\tbbb{\lnot \dk_{\RR \emptyseq}}$} \linked
      }
      child { node  {$\tbbb{\rk(\ggk)}$} 
        [sibling distance=5em]
        child { node  {$\tbbb{\lnot \dk_{\RR \emptyseq}}$ }}
        child { node {$\tbbb{\dk_{\RR 1}(\ggk)}$} 
          child { node {$\taaa{\lnot \dk_{\LL \emptyseq}}$} }
          child { node {$\taaa{\lnot \rk(\ggk)}$ }
            [sibling distance=5em]
          }
          child { node {$\taaa{\dk_{\LL 1}(\ggk)}$}
            [sibling distance=5em]
            child { node {$\taaa{\lnot \dk_{\LL 1}(\ggk)}$} \linked }
            child { node {$\taaa{\sk(\ggk,\ffk(\ggk))}$} 
              child { node {$\tbbb{\lnot \dk_{\RR 1}(\ggk)}$ }
              }
              child { node {$\tbbb{\lnot \sk(\ggk,\ffk(\ggk))}$}  \linked }
              child { node {$\tbbb{\dk_{\RR 11}}$} 
                child { node {$\tbbb{\lnot \dk_{\RR 11}}$} \linked }
              }
            }
          }
        }
      }
    };
\end{tikzpicture}
\end{tableaufigTwoCol}
\end{examp}

\begin{examp}[Connection Tableau II]
\label{examp-conn-2}
Like Example~\ref{examp-conn-1}, this example starts with a closed tightly
connected clausal tableau for the clausal formula~(\ref{eq-formula-examp-xy})
and proceeds in rounds of Procedure~\ref{proc-leafonly}
(Fig.~\ref{fig-tab-ex-2-1}--\ref{fig-tab-ex-2-3}), steps 1.-4. of
Procedure~\ref{proc-contig} (Fig.~\ref{fig-tab-ex-2-4}) and regularity
simplification with Procedure~\ref{proc-simp-regular} to an \acit
(Fig.~\ref{fig-tab-ex-2-5}).

\renewcommand{\extabscale}{0.7}
\renewcommand{\extabld}{10ex}

\begin{tableaufigTwoCol}{Example~\ref{examp-conn-2} -- Stage 1}
\label{fig-tab-ex-2-1}
\begin{tikzpicture}[scale=\extabscale,
    baseline=(a.north),
    sibling distance=5em,level distance=\extabld,
    every node/.style = {transform shape,anchor=mid}]]
    \node (a) {$\tbbb{\dk_{\RR \emptyseq}}$}
    child { node {$\tbbb{\lnot \dk_{\RR \emptyseq}}$} \linked }
    child { node {$\tbbb{\rk(\ggk)}$} 
      child {  \nodeMarkLeafonly {$\taaa{\lnot \dk_{\LL \emptyseq}}$}
        child { node {$\taaa{\dk_{\LL \emptyseq}}$} \linked }
      }
      child { node {$\taaa{\lnot \rk(\ggk)}$} \linked }
      child { node {$\taaa{\dk_{\LL 1}(\ggk)}$} 
        child { node {$\taaa{\lnot \dk_{\LL 1}(\ggk)}$} \linked}
        child { node {$\taaa{\sk(\ggk,\ffk(\ggk))}$ }
          child { node {$\tbbb{\lnot \dk_{\RR 1}(\ggk)}$}
            child { node {$\tbbb{\lnot \dk_{\RR \emptyseq}}$} 
            }
            child { node {$\tbbb{\dk_{\RR 1}(\ggk)}$} \linked }
          }
          child { node {$\tbbb{\lnot \sk(\ggk,\ffk(\ggk))}$} \linked }
          child { node {$\tbbb{\dk_{\RR 11}}$} 
            child { node {$\tbbb{\lnot \dk_{\RR 11}}$} \linked }
          }
        }
      }
    };
\end{tikzpicture}
\end{tableaufigTwoCol}    
\begin{tableaufigTwoCol}{Example~\ref{examp-conn-2} -- Stage 2}
\label{fig-tab-ex-2-2}
\begin{tikzpicture}[scale=\extabscale,
    baseline=(a.north),
    sibling distance=5em,level distance=\extabld,
    every node/.style = {transform shape,anchor=mid}]]
    \node (a) {$\tbbb{\dk_{\RR \emptyseq}}$}
    child { node {$\tbbb{\lnot \dk_{\RR \emptyseq}}$} \linked }
    child { node {$\tbbb{\rk(\ggk)}$}
      child { node {$\taaa{\dk_{\LL \emptyseq}}$}
        child { node {$\taaa{\lnot \dk_{\LL \emptyseq}}$} \linked
        }
        child { node {$\taaa{\lnot \rk(\ggk)}$} }
        child { node {$\taaa{\dk_{\LL 1}(\ggk)}$} 
          child { node {$\taaa{\lnot \dk_{\LL 1}(\ggk)}$} \linked}
          child { node {$\taaa{\sk(\ggk,\ffk(\ggk))}$ }
            child { \nodeMarkLeafonly {$\tbbb{\lnot \dk_{\RR 1}(\ggk)}$}
              child { node {$\tbbb{\lnot \dk_{\RR \emptyseq}}$} 
              }
              child { node {$\tbbb{\dk_{\RR 1}(\ggk)}$} \linked }
            }
            child { node {$\tbbb{\lnot \sk(\ggk,\ffk(\ggk))}$} \linked }
            child { node {$\tbbb{\dk_{\RR 11}}$} 
              child { node {$\tbbb{\lnot \dk_{\RR 11}}$} \linked }
            }
          }
        }
      }
    };
\end{tikzpicture}
\end{tableaufigTwoCol}

\begin{tableaufigTwoCol}{Example~\ref{examp-conn-2} -- Stage 3}
\label{fig-tab-ex-2-3}
\begin{tikzpicture}[scale=\extabscale,
    baseline=(a.north),
    sibling distance=5em,level distance=\extabld,
    every node/.style = {transform shape,anchor=mid}]]
    \node (a) {$\tbbb{\dk_{\RR \emptyseq}}$}
    child { node {$\tbbb{\lnot \dk_{\RR \emptyseq}}$} \linked }
    child { \nodeMarkContig {$\tbbb{\rk(\ggk)}$}
      child { node {$\taaa{\dk_{\LL \emptyseq}}$}
        child { node {$\taaa{\lnot \dk_{\LL \emptyseq}}$} \linked
        }
        child { node {$\taaa{\lnot \rk(\ggk)}$} }
        child { node {$\taaa{\dk_{\LL 1}(\ggk)}$} 
          child { node {$\taaa{\lnot \dk_{\LL 1}(\ggk)}$} \linked}
          child { node {$\taaa{\sk(\ggk,\ffk(\ggk))}$ }
            child { node {$\tbbb{\lnot \dk_{\RR \emptyseq}}$} 
            }
            child { \nodeMarkContig {$\tbbb{\dk_{\RR 1}(\ggk)}$}
              child { node {$\tbbb{\lnot \dk_{\RR 1}(\ggk)}$}  \linked
              }
              child { node {$\tbbb{\lnot \sk(\ggk,\ffk(\ggk))}$} }
              child { node {$\tbbb{\dk_{\RR 11}}$} 
                child { node {$\tbbb{\lnot \dk_{\RR 11}}$} \linked }
              }
            }
          }
        }
      }
    };
\end{tikzpicture}
\end{tableaufigTwoCol} 
\begin{tableaufigTwoCol}{Example~\ref{examp-conn-2} -- Stage 4}
\label{fig-tab-ex-2-4}
\begin{tikzpicture}[scale=\extabscale,
    baseline=(a.north),
    sibling distance=5em,level distance=\extabld,
    every node/.style = {transform shape,anchor=mid}]]
    \node (a) {$\tbbb{\dk_{\RR \emptyseq}}$}
    child { node {$\tbbb{\lnot \dk_{\RR \emptyseq}}$} \linked }
    child { node {$\tbbb{\rk(\ggk)}$}
      child { node {$\tbbb{\lnot \dk_{\RR \emptyseq}}$} }
      child { \nodeMarkRegular {$\tbbb{\dk_{\RR 1}(\ggk)}$} 
        child { node {$\taaa{\dk_{\LL \emptyseq}}$}
          child { node {$\taaa{\lnot \dk_{\LL \emptyseq}}$} \linked
          }
          child { node {$\taaa{\lnot \rk(\ggk)}$} }
          child { node {$\taaa{\dk_{\LL 1}(\ggk)}$} 
            child { node {$\taaa{\lnot \dk_{\LL 1}(\ggk)}$} \linked}
            child { node {$\taaa{\sk(\ggk,\ffk(\ggk))}$ }
              child { node {$\tbbb{\lnot \dk_{\RR \emptyseq}}$} 
              }
              child { \nodeMarkRegular {$\tbbb{\dk_{\RR 1}(\ggk)}$}
                child { node {$\tbbb{\lnot \dk_{\RR 1}(\ggk)}$}  \linked
                }
                child { node {$\tbbb{\lnot \sk(\ggk,\ffk(\ggk))}$} }
                child { node {$\tbbb{\dk_{\RR 11}}$} 
                  child { node {$\tbbb{\lnot \dk_{\RR 11}}$} \linked }
                }
              }
            }
          }
        }
      }
    };
\end{tikzpicture}
\end{tableaufigTwoCol} 

\begin{tableaufig}{Example~\ref{examp-conn-2} -- Stage 5}
\label{fig-tab-ex-2-5}
\begin{tikzpicture}[scale=\extabscale,
    baseline=(a.north),
    sibling distance=5em,level distance=\extabld,
    every node/.style = {transform shape,anchor=mid}]]
    \node (a) {$\tbbb{\dk_{\RR \emptyseq}}$}
    child { node {$\tbbb{\lnot \dk_{\RR \emptyseq}}$} \linked }
    child { node {$\tbbb{\rk(\ggk)}$}
      child { node {$\tbbb{\lnot \dk_{\RR \emptyseq}}$} }
      child { node {$\tbbb{\dk_{\RR 1}(\ggk)}$} 
        child { node {$\taaa{\dk_{\LL \emptyseq}}$}
          child { node {$\taaa{\lnot \dk_{\LL \emptyseq}}$} \linked
          }
          child { node {$\taaa{\lnot \rk(\ggk)}$} }
          child { node {$\taaa{\dk_{\LL 1}(\ggk)}$} 
            child { node {$\taaa{\lnot \dk_{\LL 1}(\ggk)}$} \linked}
            child { node {$\taaa{\sk(\ggk,\ffk(\ggk))}$ }
              child { node {$\tbbb{\lnot \dk_{\RR 1}(\ggk)}$}  \linked
              }
              child { node {$\tbbb{\lnot \sk(\ggk,\ffk(\ggk))}$} }
              child { node {$\tbbb{\dk_{\RR 11}}$} 
                child { node {$\tbbb{\lnot \dk_{\RR 11}}$} \linked }
              }
            }
          }
        }
      }
    };
\end{tikzpicture}
\end{tableaufig} 

\end{examp}

\pagebreak
\begin{examp}[Connection Tableau III]
\label{examp-conn-3}
Like Example~\ref{examp-conn-1} and~\ref{examp-conn-2}, this example starts
with a closed tightly connected clausal tableau for the clausal
formula~(\ref{eq-formula-examp-xy}) and proceed in rounds of
Procedure~\ref{proc-leafonly}
(Fig.~\ref{fig-tab-ex-3-1}--\ref{fig-tab-ex-3-5}), steps 1.-4. of
Procedure~\ref{proc-contig} (Fig.~\ref{fig-tab-ex-3-6}) and regularity
simplification with Procedure~\ref{proc-simp-regular} to an \acit
(Fig.~\ref{fig-tab-ex-3-7}).

\renewcommand{\extabscale}{0.62}
\renewcommand{\extabld}{12ex}

\begin{tableaufigTwoCol}{Example~\ref{examp-conn-3} -- Stage 1}
\label{fig-tab-ex-3-1}
\begin{tikzpicture}[scale=\extabscale,
    baseline=(a.north),
    sibling distance=5em,level distance=\extabld,
    every node/.style = {transform shape,anchor=mid}]]
    \node (a) {$\tbbb{\dk_{\RR \emptyseq}}$}
    child { node {$\tbbb{\lnot \dk_{\RR \emptyseq}}$ } \linked }
    child { node {$\tbbb{\dk_{\RR 1}(\ggk)}$} 
      [sibling distance=9em]
      child { node {$\tbbb{\lnot \dk_{\RR 1}(\ggk)}$} \linked }
      child { \nodeMarkLeafonly {$\tbbb{\lnot \sk(\ggk,\ffk(\ggk))}$} 
        [sibling distance=6em]
        child { node {$\taaa{\lnot \dk_{\LL 1}(\ggk)}$} 
          [sibling distance=5em]
          child { node {$\taaa{\lnot \dk_{\LL \emptyseq}}$}
            child { node {$\taaa{\dk_{\LL \emptyseq}}$} \linked }
          }
          child { node {$\taaa{\lnot \rk(\ggk)}$}
            [sibling distance=5em]
            child { node {$\tbbb{\lnot \dk_{\RR \emptyseq}}$} }
            child { node {$\tbbb{\rk(\ggk)}$} \linked}
          }
          child { node {$\taaa{\dk_{\LL 1}(\ggk)}$} \linked}
        }
        child { node {$\taaa{\sk(\ggk,\ffk(\ggk))}$} \linked }
      }
      child { node {$\tbbb{\dk_{\RR 11}}$} 
        child { node {$\tbbb{\lnot \dk_{\RR 11}}$} \linked }
      }
    };
\end{tikzpicture}
\end{tableaufigTwoCol}
\begin{tableaufigTwoCol}{Example~\ref{examp-conn-3} -- Stage 2}
\label{fig-tab-ex-3-2}
\begin{tikzpicture}[scale=\extabscale,
    baseline=(a.north),
    sibling distance=5em,level distance=\extabld,
    every node/.style = {transform shape,anchor=mid}]]
    \node (a) {$\tbbb{\dk_{\RR \emptyseq}}$}
    child { node {$\tbbb{\lnot \dk_{\RR \emptyseq}}$ } \linked }
    child { node {$\tbbb{\dk_{\RR 1}(\ggk)}$} 
      [sibling distance=14em]
      child { \nodeMarkLeafonly {$\taaa{\lnot \dk_{\LL 1}(\ggk)}$} 
        [sibling distance=5em]
        child { node {$\taaa{\lnot \dk_{\LL \emptyseq}}$}
          child { node {$\taaa{\dk_{\LL \emptyseq}}$} \linked }
        }
        child { node {$\taaa{\lnot \rk(\ggk)}$}
          [sibling distance=5em]
          child { node {$\tbbb{\lnot \dk_{\RR \emptyseq}}$} }
          child { node {$\tbbb{\rk(\ggk)}$} \linked}
        }
        child { node {$\taaa{\dk_{\LL 1}(\ggk)}$} \linked}
      }
      child { node {$\taaa{\sk(\ggk,\ffk(\ggk))}$} 
        [sibling distance=5em]
        child { node {$\tbbb{\lnot \dk_{\RR 1}(\ggk)}$} }
        child { node {$\tbbb{\lnot \sk(\ggk,\ffk(\ggk))}$} \linked }
        child { node {$\tbbb{\dk_{\RR 11}}$} 
          child { node {$\tbbb{\lnot \dk_{\RR 11}}$} \linked }
        }
      }
    };
\end{tikzpicture}
\end{tableaufigTwoCol}

\begin{tableaufigTwoCol}{Example~\ref{examp-conn-3} -- Stage 3}
\label{fig-tab-ex-3-3}
\begin{tikzpicture}[scale=\extabscale,
    baseline=(a.north),
    sibling distance=5em,level distance=\extabld,
    every node/.style = {transform shape,anchor=mid}]]
    \node (a) {$\tbbb{\dk_{\RR \emptyseq}}$}
    child { node {$\tbbb{\lnot \dk_{\RR \emptyseq}}$ } \linked }
    child { node {$\tbbb{\dk_{\RR 1}(\ggk)}$} 
      [sibling distance=9em]
      child { \nodeMarkLeafonly {$\taaa{\lnot \dk_{\LL \emptyseq}}$}
        child { node {$\taaa{\dk_{\LL \emptyseq}}$} \linked }
      }
      child { node {$\taaa{\lnot \rk(\ggk)}$}
        [sibling distance=5em]
        child { node {$\tbbb{\lnot \dk_{\RR \emptyseq}}$} }
        child { node {$\tbbb{\rk(\ggk)}$} \linked}
      }
      child { node {$\taaa{\dk_{\LL 1}(\ggk)}$}
        [sibling distance=5em]
        child { node {$\taaa{\lnot \dk_{\LL 1}(\ggk)}$} \linked }
        child { node {$\taaa{\sk(\ggk,\ffk(\ggk))}$} 
          [sibling distance=5em]
          child { node {$\tbbb{\lnot \dk_{\RR 1}(\ggk)}$} }
          child { node {$\tbbb{\lnot \sk(\ggk,\ffk(\ggk))}$} \linked }
          child { node {$\tbbb{\dk_{\RR 11}}$} 
            child { node {$\tbbb{\lnot \dk_{\RR 11}}$} \linked }
          }
        }
      }
    };
\end{tikzpicture}
\end{tableaufigTwoCol}
\begin{tableaufigTwoCol}{Example~\ref{examp-conn-3} -- Stage 4}
\label{fig-tab-ex-3-4}
\begin{tikzpicture}[scale=\extabscale,
    baseline=(a.north),
    sibling distance=5em,level distance=\extabld,
    every node/.style = {transform shape,anchor=mid}]]
    \node (a) {$\tbbb{\dk_{\RR \emptyseq}}$}
    child { node {$\tbbb{\lnot \dk_{\RR \emptyseq}}$ } \linked }
    child { node {$\tbbb{\dk_{\RR 1}(\ggk)}$} 
      child { node {$\taaa{\dk_{\LL \emptyseq}}$}
        [sibling distance=9em]
        child { node {$\taaa{\lnot \dk_{\LL \emptyseq}}$} \linked
        }
        child { \nodeMarkLeafonly {$\taaa{\lnot \rk(\ggk)}$}
          [sibling distance=5em]
          child { node {$\tbbb{\lnot \dk_{\RR \emptyseq}}$} }
          child { node {$\tbbb{\rk(\ggk)}$} \linked}
        }
        child { node {$\taaa{\dk_{\LL 1}(\ggk)}$}
          [sibling distance=5em]
          child { node {$\taaa{\lnot \dk_{\LL 1}(\ggk)}$} \linked }
          child { node {$\taaa{\sk(\ggk,\ffk(\ggk))}$} 
            [sibling distance=5em]
            child { node {$\tbbb{\lnot \dk_{\RR 1}(\ggk)}$} }
            child { node {$\tbbb{\lnot \sk(\ggk,\ffk(\ggk))}$} \linked }
            child { node {$\tbbb{\dk_{\RR 11}}$} 
              child { node {$\tbbb{\lnot \dk_{\RR 11}}$} \linked }
            }
          }
        }
      }
    };
\end{tikzpicture}
\end{tableaufigTwoCol}

\begin{tableaufigTwoCol}{Example~\ref{examp-conn-3} -- Stage 5}
\label{fig-tab-ex-3-5}
\begin{tikzpicture}[scale=\extabscale,
    baseline=(a.north),
    sibling distance=5em,level distance=\extabld,
    every node/.style = {transform shape,anchor=mid}]]
    \node (a) {$\tbbb{\dk_{\RR \emptyseq}}$}
    child { node {$\tbbb{\lnot \dk_{\RR \emptyseq}}$ } \linked }
    child { \nodeMarkContig {$\tbbb{\dk_{\RR 1}(\ggk)}$} 
      child { node {$\taaa{\dk_{\LL \emptyseq}}$}
        [sibling distance=5em]
        child { node {$\tbbb{\lnot \dk_{\RR \emptyseq}}$} }
        child { \nodeMarkContig {$\tbbb{\rk(\ggk)}$} 
          [sibling distance=5em]
          child { node {$\taaa{\lnot \dk_{\LL \emptyseq}}$}
          }
          child { node {$\taaa{\lnot \rk(\ggk)}$}  \linked
          }
          child { node {$\taaa{\dk_{\LL 1}(\ggk)}$}
            [sibling distance=5em]
            child { node {$\taaa{\lnot \dk_{\LL 1}(\ggk)}$} \linked }
            child { node {$\taaa{\sk(\ggk,\ffk(\ggk))}$} 
              [sibling distance=5em]
              child { node {$\tbbb{\lnot \dk_{\RR 1}(\ggk)}$} }
              child { node {$\tbbb{\lnot \sk(\ggk,\ffk(\ggk))}$} \linked }
              child { node {$\tbbb{\dk_{\RR 11}}$} 
                child { node {$\tbbb{\lnot \dk_{\RR 11}}$} \linked }
              }
            }
          }
        }
      }
    };
\end{tikzpicture}
\end{tableaufigTwoCol}
\begin{tableaufigTwoCol}{Example~\ref{examp-conn-3} -- Stage 6}
\label{fig-tab-ex-3-6}
\begin{tikzpicture}[scale=\extabscale,
    baseline=(a.north),
    sibling distance=5em,level distance=\extabld,
    every node/.style = {transform shape,anchor=mid}]]
    \node (a) {$\tbbb{\dk_{\RR \emptyseq}}$}
    child { node {$\tbbb{\lnot \dk_{\RR \emptyseq}}$ } \linked }
    child { node {$\tbbb{\dk_{\RR 1}(\ggk)}$} 
      child { node {$\tbbb{\lnot \dk_{\RR \emptyseq}}$} }
      child { \nodeMarkRegular {$\tbbb{\rk(\ggk)}$}       
        child { node {$\taaa{\dk_{\LL \emptyseq}}$}
          [sibling distance=5em]
          child { node {$\tbbb{\lnot \dk_{\RR \emptyseq}}$} }
          child { \nodeMarkRegular {$\tbbb{\rk(\ggk)}$} 
            [sibling distance=5em]
            child { node {$\taaa{\lnot \dk_{\LL \emptyseq}}$}
            }
            child { node {$\taaa{\lnot \rk(\ggk)}$}  \linked
            }
            child { node {$\taaa{\dk_{\LL 1}(\ggk)}$}
              [sibling distance=5em]
              child { node {$\taaa{\lnot \dk_{\LL 1}(\ggk)}$} \linked }
              child { node {$\taaa{\sk(\ggk,\ffk(\ggk))}$} 
                [sibling distance=5em]
                child { node {$\tbbb{\lnot \dk_{\RR 1}(\ggk)}$} }
                child { node {$\tbbb{\lnot \sk(\ggk,\ffk(\ggk))}$} \linked }
                child { node {$\tbbb{\dk_{\RR 11}}$} 
                  child { node {$\tbbb{\lnot \dk_{\RR 11}}$} \linked }
                }
              }
            }
          }
        }
      }
    };
\end{tikzpicture}
\end{tableaufigTwoCol}

\begin{tableaufig}{Example~\ref{examp-conn-3} -- Stage 7}
\label{fig-tab-ex-3-7}
\begin{tikzpicture}[scale=\extabscale,
    baseline=(a.north),
    sibling distance=5em,level distance=\extabld,
    every node/.style = {transform shape,anchor=mid}]]
    \node (a) {$\tbbb{\dk_{\RR \emptyseq}}$}
    child { node {$\tbbb{\lnot \dk_{\RR \emptyseq}}$ } \linked }
    child { node {$\tbbb{\dk_{\RR 1}(\ggk)}$} 
      child { node {$\tbbb{\lnot \dk_{\RR \emptyseq}}$} }
      child { node {$\tbbb{\rk(\ggk)}$}       
        child { node {$\taaa{\dk_{\LL \emptyseq}}$}
          [sibling distance=5em]
          child { node {$\taaa{\lnot \dk_{\LL \emptyseq}}$}
          }
          child { node {$\taaa{\lnot \rk(\ggk)}$}  \linked
          }
          child { node {$\taaa{\dk_{\LL 1}(\ggk)}$}
            [sibling distance=5em]
            child { node {$\taaa{\lnot \dk_{\LL 1}(\ggk)}$} \linked }
            child { node {$\taaa{\sk(\ggk,\ffk(\ggk))}$} 
              [sibling distance=5em]
              child { node {$\tbbb{\lnot \dk_{\RR 1}(\ggk)}$} }
              child { node {$\tbbb{\lnot \sk(\ggk,\ffk(\ggk))}$} \linked }
              child { node {$\tbbb{\dk_{\RR 11}}$} 
                child { node {$\tbbb{\lnot \dk_{\RR 11}}$} \linked }
              }
            }
          }
        }
      }
    };
\end{tikzpicture}
\end{tableaufig}

\end{examp}

\section{Access Interpolation: Refinements, Issues and Related Work}

\label{sec-ai-conclusion}

\subsection{Alternate Tableaux -- Alternate Interpolants}

In general there are different closed clausal tableaux for a single given
unsatisfiable clausal formula. A different interpolant would be extracted from
each tableau. For an application such as query reformulation some of these
interpolants might be more preferable than others.  For example, a query might
be preferably reformulated in terms of more specific access patterns such that
available parameter instantiations are best utilized when the reformulated
query is evaluated.

In \cite{toman:wedell:book,toman:2015:tableaux,toman:2017} query reformulation
is indeed based on computing alternate interpolants -- each considered as
representing a query plan -- and comparing them with a cost function. This
approach has been refined in \cite{toman:2015:tableaux,toman:2017} with a
condensed representation of a set of tableaux in a single structure.  An
advanced system that interleaves the generation of a pair of such condensed
tableaux, one for each of the two interpolation input formulas, with detecting
when their combination would be closed is described in \cite{toman:2017}.

Enumeration of closed clausal tableaux for a given set of clauses is quite
natural for goal-sensitive clausal tableau methods such as model elimination
and the connection method which operate with backtracking anyways.  If such a
calculus is not stopped after finishing construction of a closed clausal
tableau it backtracks to generate further closed tableaux (\name{CM}
\cite{cw-pie} can for example be configured in that way).  However,
requirements for theorem proving and for the computation of interpolants as
query plans seem to contradict: Theorem provers typically aim to prevent the
search for alternate proofs as much as possible without compromising
completeness or experimental success, whereas, if interpolants are considered
as query plans with associated costs it would make sense to compare even
proofs with ``trivial'' differences, for example, if they correspond to
interpolants just distinguished by a different order of conjuncts in a
subformula.

Heuristics that can be configured to give priority to tableaux that are
preferred with respect to the application seem in general useful. For
goal-sensitive provers based on model elimination or on the connection
calculus the order in which clauses are picked for inclusion in the tableau is
relatively easy to influence.  With backtracking and iterative deepening
completeness is preserved by ensuring that lower ranked clauses will
eventually be considered as long as no closed tableau with more highly ranked
clauses has been found. A common heuristics is to rank clauses by their
length, shortest first.  Application specific orderings can, for example,
effect that clauses associated with more specific reformulations are given
priority.\footnote{The \name{CM} prover included with the \name{PIE} system
  \cite{cw-pie} supports this with an experimental option. For an example see
  Section~2.2 in
  \url{http://cs.christophwernhard.com/pie/downloads/pie/scratch/scratch_views_lit.pdf}.}

\subsection{Preprocessing Issues and Resolution}

Preserving the second-order equivalence~(\ref{eq-so-simp}) as discussed in
Sect.~\ref{sec-preproc} for preprocessing inputs of Craig-Lyndon
interpolation is too weak for access interpolation, as the latter depends on
further constraints on the clause form.  It seems, however, possible to
generalize some of the clause forms of Lemma~\ref{lem-input-clauseforms} such
that they are closed under certain preprocessing steps that break apart and
join clauses. Forms~\mbox{2--5}, for example, could be easily generalized to a
single form with one negative and an arbitrary number of positive definer
literals whose arguments, and hence also free variables, all occur in the
negative literal. This generalized form is closed under resolution. The
forms~2--5 are already handled in a generalized way together in proof of
Lemma~\ref{lem-aipol-invariant}.  Form~6 could be generalized by allowing
further positive literals. However, unrestricted resolution among such clauses
may yield clauses with multiple negative ``$R$''-literals (literals that are
not definers). Exploring ways to restrict resolution for these clauses is an
issue for future research.

\subsection{Application to Query Reformulation} The actual specification of
axiom schemas used as inputs of interpolation tasks for query reformulation is
beyond the scope of this paper. We refer to the literature, in particular to
\cite{benedikt:book,benedikt:2017} and \cite{toman:wedell:book}.
Nevertheless, certain aspects of the axiom schemas in \cite{benedikt:book}
seem to be closely connected to the precise way in which input formulas are
expressed and processed by a theorem prover, suggesting further investigations
in the context of our method.  This concerns in particular the consideration
of access interpolation for ``non-Boolean'' queries, that is, queries whose
results are relations of non-zero arity, discussed in
\cite[Sec.~3.8]{benedikt:book} and the axiom schema $\f{AltAcSch}^{\equi}$
from \cite[Section~3.6]{benedikt:book}, which is apparently obtained from a more
abstract specification $\f{AcSch}^{\equi}$ with the involvement of unfolding
predicates, raising the question whether a more condensed representation
without the need of unfolding is possible.

The considered approach to query reformulation is not principally limited to
relational database queries, but applies to logic-based knowledge
representation mechanisms in general.  So far, the main application direction
is to optimize a given query with respect to a given access schema.  The
approach can, however, in also be applied conversely to determine from given
parameterized queries an access schema that would be required to answer these
queries.  That inferred access schema can then be used to determine caches and
precomputed indexes as basis for evaluating the queries at a later time.

\subsection{Implementations of Interpolation for Query
  Reformulation} 
\label{sec-implem-qr}
In \cite{benedikt:2017} different approaches to implement query reformulation
based on Craig-Lyndon interpolation have been experimentally investigated,
however only for ground inputs.  One approach considered there was an
extension of the general first-order prover \name{Vampire} that supports
interpolant computation \cite{vampire:interpol:2012}. Apparently it does not
ensure the polarity constraint on Craig-Lyndon interpolants (the ``Lyndon
property''), had many timeouts and produced poor result formulas, indicating
that the requirements of interpolation in verification, the main objective of
the \name{Vampire} extension, and in query reformulation are quite
different.\footnote{Aside of \name{Vampire}, also the \name{Princess} theorem
  prover (\url{http://www.philipp.ruemmer.org/princess.shtml})
  \cite{ruemmer:ipol:jar:2011} supports Craig interpolation, mainly with
  respect to theories targeted at applications in verification.  The
  \name{PIE} system is another first-order prover with support for Craig
  interpolation (see Sect.~\ref{sec-cli-implem}).}  Better results were
obtained with methods for interpolant extraction from resolution proofs, based
on algorithms from \cite{huang:95,bonacina:11,mcmillan:2003}.  The
\name{MathSAT} SMT solver \cite{mathsat:13} and the \name{E} first-order
prover \cite{eprover:13} have been used to compute the underlying resolution
proofs (apparently \name{E} was the only first-order resolution prover that
produced sufficiently detailed proofs). \name{MathSAT} was there superior to
\name{E}. For extraction of the propositional interpolants an optimized
variant of the algorithm of \cite{huang:95} introduced in \cite{benedikt:2017}
as well as the method of \cite{mcmillan:2003} showed best. Also a method based
on the chase technique, implemented with \name{DLV} \cite{dlv} as model
generator, has been evaluated in \cite{benedikt:2017}.  In essence, a
\emph{uniform} interpolant (that is, the result of predicate elimination) is
computed there by removing literals whose predicate is not allowed in the
interpolant from a disjunctive normal form.  In cases where the expected
reformulation is a disjunctive normal form this approach is only slightly
worse than the best resolution-based approach.
It seems that access interpolation has so far not been implemented.  As
already mentioned, the approach of
\cite{toman:wedell:book,toman:2015:tableaux} has been implemented with
advanced dedicated techniques \cite{toman:2017}.  Small examples with a simple
form of axiom schemas that are processed by Craig-Lyndon interpolation on the
basis of a general first-order clausal tableau prover come with the \name{PIE}
system \cite{cw-pie}.\footnote{%
  \url{http://cs.christophwernhard.com/pie/downloads/pie/scratch/scratch_access_demo_01.pdf},
  \url{http://cs.christophwernhard.com/pie/downloads/pie/scratch/scratch_views_lit.pdf}}

\section{Conclusion}
\label{sec-conclusion}

We investigated the computation of Craig-Lyndon interpolants and of access
interpolants, a recent form of interpolation with applications in query
reformulation, by means of clausal tableau methods.  Aspects of the elegance
of an established interpolant construction based on non-clausal tableaux were
combined with the suitability of clausal tableaux for machine processing.  The
framework of clausal tableaux as a basis in contrast to resolution leads to a
natural way to decompose the overall task of computing first-order
interpolants into subtasks, which seems useful for theoretical considerations
as well as practical implementation.  This new modularization concerns three
different aspects: the lifting from ground interpolants to quantified
formulas, the roles of ``local'' versus ``non-local'' techniques, and the
interplay of proof search with structural requirements on the proof:

\begin{enumerate}

\item \label{abstr-lifting} Computation of Craig-Lyndon interpolants is
  performed on the basis of clausal tableaux in two stages, like in most
  resolution-based interpolation methods. In contrast to these, however, the
  second stage, the lifting to a quantified first-order formula, is performed
  here on an actual Craig-Lyndon interpolant of a finite unsatisfiable subset
  of the Herbrand expansion of the Skolemized and clausified input formulas
  instead of some formula that is characterized as ``almost'' a Craig-Lyndon
  interpolant of the original input formulas.  This allows to consider the
  lifting conversion more abstractly, just based on Herbrand's theorem,
  independently from a particular calculus.  For practical implementation, it
  means that any method that computes a closed clausal tableau for an
  unsatisfiable first-order input formula can be directly applied to
  interpolant computation, without need to alter its internal workings,
  benefiting directly from refinements and efficient data structures.

\item \label{abstr-local} Differently from resolution, the construction of a
  clausal tableau does not involve breaking apart and joining
  clauses. However, such operations can be performed during preprocessing for
  a clausal tableau prover.  We then get an overall picture of the
  tableau-based interpolation where first operations that break apart and join
  clauses, like resolution, are performed only ``locally'', that is, on each
  of the two input formulas individually.  These operations must preserve the
  semantics of the predicates and functions that are allowed in the
  interpolant, but can eliminate or semantically alter other predicates and
  functions.  After this preprocessing stage is completed, for example,
  because no further conversion operations are possible or because potential
  further operations would increase the formula size in an undesired way, the
  actual tableau construction comes in to handle the ``non-local'' joint
  processing of both preprocessed inputs.

\item \label{abstr-struct} In our approach to the computation of access
  interpolants the underlying clausal tableaux are required to meet specific
  structural constraints.  ``Bottom-up'' methods for constructing clausal
  tableaux, such as the hyper tableau calculus, typically can be configured to
  compute tableaux that satisfy these constraints. For other calculi the
  tableau construction is split in two phases: First, a theorem prover
  computes an arbitrarily structured closed clausal tableau.  Second, the
  structure of the tableau output by the prover is converted such that it
  satisfies the required restrictions. This conversion is, however,
  potentially costly since it involves steps that might involve duplication of
  subtableaux.
\end{enumerate}  

These aspects suggest a number of challenging follow-up questions. With
respect to aspect~(\ref{abstr-lifting}.): Can the justification of lifting for
Craig-Lyndon interpolants also be applied to resolution-based interpolation
methods?  Can this lifting method, which prepends a single quantifier prefix
to the whole formula be reconciled with requirements of relativized
quantification as in access interpolation, where quantifier scopes that are
limited to subformulas seem essential?  With respect to~(\ref{abstr-local}.):
Does the observation of the roles of ``local'' versus ``non-local'' inferences
in interpolation based on clausal tableaux indicate some interesting property
where clausal tableaux and resolution diverge?  Is interpolation a field for
which clausal tableau methods are better suited than resolution in some
substantial sense?  With respect to~(\ref{abstr-struct}.): Are there
interesting implications of the ``calculus that preserves a structure'' versus
``calculus has more freedom followed by potentially costly conversion''
approaches?  Can it be shown that the first approach implies that the costs of
conversion have to be incorporated in essence into the proof search?

The presented material provides foundations to practically implement
Craig-Lyndon interpolation and access interpolation on the basis of a variety
of machine-oriented theorem-proving methods for first-order logic that can be
considered as constructing a closed clausal tableaux.  These fall into two
families, goal-oriented ``top-down'' methods such as model elimination and the
connection method, and data-oriented ``bottom-up'' methods such as the hyper
tableau calculus. A first implementation of Craig-Lyndon interpolation on the
basis of a prover of the first family is already available \cite{cw-pie}.
Since the presented methods and most discussed refinements incorporate
first-order provers that compute clausal tableaux abstractly, without imposing
requirements on tableau construction methods, it should only be a short way
from the foundations provided with this work to experimental evaluations.

\subsubsection*{Acknowledgments.}
This research was supported by Deutsche Forschungsgemeinschaft with
grant~\mbox{WE~5641/1-1} for the project \name{The Second-Order Approach and its
  Application to View-Based Query Processing} hosted by Tech\-ni\-sche
Uni\-ver\-si\-tät Dresden, Germany.

\bibliographystyle{spmpsci}
\bibliography{bibelim10short}

\closeout\keysfile
\closeout\keyslogfile
\end{document}